\documentclass[journal,onecolumn]{IEEEtran}
\usepackage{cite}
\usepackage{amsmath,amssymb,amsfonts,amsthm, color}
\usepackage{graphicx}
\usepackage{textcomp}
\usepackage{xcolor}
\usepackage{enumitem}
\usepackage{algorithm2e}
\SetKwInput{KwData}{Encoder}
\SetKwInput{KwResult}{The result}
\usepackage{float}
\newtheorem*{theorem*}{Theorem}
\newtheorem{theorem}{Theorem} 
\newtheorem{corollary}{Corollary} 
 
\newtheorem{proposition}{Proposition} 
 
\newtheorem{lemma}{Lemma} 
\newtheorem{example}{Example}
\newtheorem{claim}{Claim}
\usepackage[utf8]{inputenc}
\usepackage{xcolor}
\usepackage{soul}
\usepackage[noend]{algpseudocode}
\usepackage{mathtools}
\usepackage{cuted}
\usepackage{caption}

\newcommand{\C}{{\mathcal C}}

\newcommand{\E}{{\mathcal E}}

\newcommand{\cO}{{\mathcal O}}

\newcommand{\R}{{\mathcal R}}

\renewcommand{\S}{{\mathcal S}}
\renewcommand{\P}{{\mathcal P}}
\renewcommand{\O}{{\mathcal O}}

\newcommand{\ba}{{\boldsymbol a}}

\newcommand{\bv}{{\boldsymbol v}}

\newcommand{\bs}{{\boldsymbol s}}

\newcommand{\bw}{{\boldsymbol{w}}}

\newcommand{\bx}{{\boldsymbol{x}}}

\ifCLASSINFOpdf
\else
\fi
\begin{document}
%
\title{Coding for Polymer-Based Data Storage}
%
%
%

\author{Srilakshmi~Pattabiraman,~\IEEEmembership{Student Member,~IEEE,}
        Ryan~Gabrys,~\IEEEmembership{Member,~IEEE,}
        and~Olgica~Milenkovic,~\IEEEmembership{Fellow,~IEEE}
\thanks{S. Pattabiraman, and O. Milenkovic are with the Department
of Electrical and Computer Engineering, University of Illinois, Urbana-Champaign, Urbana, IL, 61801 USA e-mails: sp16@illinois.edu, milenkov@illinois.edu.}
\thanks{R. Gabrys is with the University of California, San Diego, San Diego, CA, 92093, USA e-mail: ryan.gabrys@gmail.com.}
\thanks{The work was funded by the DARPA Molecular Informatics program, the SemiSynBio program of the NSF and SRC, and the
NSF CIF grant number 1618366. }%
\thanks{A part of this work was presented at the \textit{2019 IEEE Information Theory Workshop}, Visby, Gotland, Sweden. Another part was submitted to the \textit{2020 IEEE International Symposium
on Information Theory}.}%
\thanks{The authors acknowledge the insightful comments made by the reviewers that significantly improved the content of the manuscript and presentation of the results.}%
}

%
%

\markboth{March~2020}%
{Pattabiraman \MakeLowercase{\textit{et al.}}: Coding for Polymer-Based Data Storage}
%



\maketitle

\begin{abstract}
Motivated by polymer-based data-storage platforms that use chains of binary synthetic polymers as the recording media and read the content via tandem mass spectrometers, we propose a new family of codes that allows for both unique string reconstruction and correction of multiple mass errors. We consider two approaches: The first approach pertains to asymmetric errors and it is based on introducing redundancy that scales linearly with the number of errors and logarithmically with the length of the string. The construction allows for the string to be uniquely reconstructed based only on its erroneous substring composition multiset. The key idea behind our unique reconstruction approach is to interleave (shifted) Catalan-Bertrand paths with arbitrary binary strings and ``reflect'' them so as to force prefixes and suffixes of the same length to have different weights. The asymptotic code rate of the scheme is one, and decoding is accomplished via a simplified version of the Backtracking algorithm used for the Turnpike problem. For symmetric errors, we use a polynomial characterization of the mass information and adapt polynomial evaluation code constructions for this setting. In the process, we develop new efficient decoding algorithms for a constant number of composition errors. 
\end{abstract}


%
\IEEEpeerreviewmaketitle

\section{Introduction}
Current digital storage systems are facing numerous obstacles in terms of scaling the storage density and allowing for in-memory based computations~\cite{zhirnov2016nucleic}. To offer storage densities at nanoscale, several molecular storage paradigms have recently been put forward in 
~\cite{al2017mass,goldman2013towards,grass2015robust,yazdi2015rewritable,yazdi2017portable}. One promising line of work with low storage cost and readout latency is the work in~\cite{al2017mass}, which proposes using synthetic polymers for storing user-defined information and reading the content via tandem mass spectrometry (MS/MS) techniques. More precisely, binary data is encoded using poly(phosphodiester)s, synthesized through automated phosphoramidite chemistry in such a way that the two bits $0$ and $1$ are represented by molecules of different masses that are stitched together into strings of fixed length. To read the encoded data, phosphate bonds are broken, and MS/MS readers are used to estimate the masses of the fragmented polymer and reconstruct the recorded string, as illustrated in the simplified scheme shown in Figure~\ref{fig:spectrometry}.
\begin{figure}[h]
\centering
\includegraphics[scale=0.3]{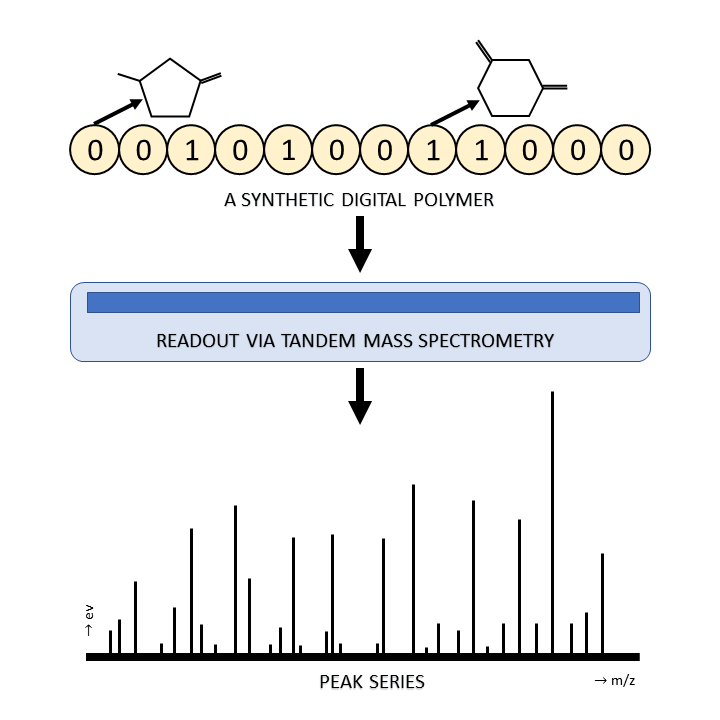} 
\caption{The scheme is adapted from~\cite{al2017mass}. The top figure depicts a binary string synthesized using phosphoramidite chemistry. The bottom image is an illustration of \emph{peak series} or MS Spectrum obtained by MS/MS readout of the digital polymer. The peak series plots the charge at the detection plates (in eV) against the ratio of the mass number of the ion and its charge number (m/z). The charge normalization is often removed through calibration thereby allowing one to deal with masses only. Under ideal conditions, the peaks are supposed to correspond to the masses of string fragments, or more precisely, masses of prefixes and suffixes of the string. Due to measurement errors, spurious peaks arise and one needs to apply specialized signal processing techniques to identify the correct peaks.}
\label{fig:spectrometry}
\end{figure}
Ideally, the masses of all prefixes and suffixes are recovered reliably, allowing one to read the message content by taking the differences of the increasing fragment masses and mapping them to the masses of the $0$ or $1$ symbol. Polymer synthesis is cost- and time-efficient and MS/MS sequencers are significantly faster than those designed for other macromolecules, such as DNA.
Nevertheless, despite the fact that the masses of the polymers can be tuned to allow for more accurate mass discrimination, polymer-based storage systems still suffer from large read error-rates. This is due to the fact that MS/MS sequencing methods tend to produce peaks, representing the masses of the fragments that are buried in analogue noise due to atom disassociation during the fragmentation process and other sources of errors.

In an earlier line of work, the authors of~\cite{acharya2014string} introduced the problem of \emph{binary string reconstruction from its substring composition multiset} to address the issue of MS/MS readout analysis. The substring composition multiset of a binary string is obtained by writing out substrings of the string of all possible lengths and then representing each substring by its composition. As an example, the string $101$ contains three substrings of length one - $1$, $0$, and $1$, two substrings of length 2 - $10$ and $01$, and one substring of length three - $101$. The composition multiset of the substrings of length one equals $\{ 0,1,1 \}$, the composition multiset of substrings of length two equals $\{0^11^1,0^11^1\}$ and the composition multiset of substrings of length three equals $\{0^11^2\}$. Note that composition multisets ignore information about the actual order of the bits in the substrings and may hence be seen as only capturing the information about the ``mass'' or ``weight'' of the unordered substrings. Furthermore, the multiset information cannot distinguish between a string and it's reversal, as well as some other nontrivial interleaved string settings. The problem addressed in~\cite{acharya2014string} was to determine for which string lengths one can guarantee unique reconstruction from an error-free composition multiset, up to string reversal. The main results of~\cite[Theorem ~17,~18,~20]{acharya2014string} assert that binary strings of length  $\leq 7$, one less than a prime or one less than twice a prime are uniquely reconstructable up to reversal.

For our line of work, we will rely on the two modeling assumptions first described in~\cite{acharya2014string}:

\textit{Assumption 1.} One can infer the composition of a polymer substring from its mass. 
\textit{Assumption 2.} When a polymer is broken down for mass spectrometry analysis, we observe the masses of all its substrings with identical frequency. 

The masses of all binary substrings of an encoded polymer may be abstracted by the composition multiset of a string, provided that Assumption 1 holds. Assumption 2 slightly deviates from practical ion series measurements in so far that the latter only provides information about the masses of the prefixes and suffixes, while the proposed modification allows one to observe the masses of all substrings, but without a priori knowledge of their order. Observe that one can make use of platforms that provide mass information for all substrings but such systems require more than one string disassociation and are hence are harder to implement and more expensive.

Unlike the work in~\cite{acharya2014string} which has solely focused on the problem of determining under which conditions unique string reconstruction is possible, we view the problem of multiset composition analysis from a coding-theoretic perspective and ask the following questions:

\textbf{Q1.} \emph{Can one add asymptotically negligible redundancy to information strings in such a way that unique reconstruction is possible, independent of the length of the strings?} Since only strings of specific lengths are reconstructable up to reversals, we aim to devise an efficiently encodable and decodable scheme that encode all strings of length $k \geq 1$ into strings of a larger length $n \geq k$ that are uniquely reconstructable for \emph{all possible string lengths}. Furthermore, we do not allow for both a string and its reversal to be included in the codebook. One simple means for ensuring that a string is uniquely reconstructable up to reversal is to pad the string with $0$s to obtain the shortest length of the form $\min\{{p-1,2q-1\}}$, where $p$ and $q$ primes. For example, if $k > 89693$, it is known that there exists a prime $p$ such that  
$k-1 < p-1 < \left( 1+\frac{1}{\ln^3\, k}\right)\,k-1.$ The result only holds for very large $k$ that are beyond the reach of polymer chemistry. Bertrand's postulate~\cite{hardy1929introduction} applies to shorter lengths $k>3$ but only guarantees that $k-1 < p-1 < 2k-4.$ This implies a possible coding rate loss of up to $1/2$. Note that eliminating reversals of strings reduces the codebook size by less than a half. 

\textbf{Q2.} \emph{Can one add asymptotically negligible redundancy to information strings in such a way that unique reconstruction is possible even in the presence of errors, independent on the length of the strings?} We focus on mass error models under which the composition (mass) of one substring is erroneously interpreted as a different composition (mass). In the asymmetric error model, no two errors can simultaneously affect the masses of two substrings of length $i$ and $k-i+1$, while in the symmetric error model such pairs are allowed. Clearly, the two models are the same when only one mass error is present. Furthermore, asymmetric errors are easily detectable even without added redundancy, while symmetric errors may not be automatically detectable. Symmetric errors tend to be correlated as they arise during the same fragmentation process, while asymmetric errors may be independent as they arise during two different fragmentation processes. It is therefore of interest to analyze both cases.

We answer both questions affirmatively by describing coding schemes that allow for both unique reconstruction and correction of multiple symmetric and asymmetric mass errors. For the case of asymmetric errors, encoding is performed by interleaving symmetric strings with shifted Catalan-Bertrand paths while decoding is accomplished through a modification of the backtracking decoding algorithm described in~\cite{acharya2014string}. For symmetric errors, the proposed encoding and decoding procedures use the polynomial factorization approach of~\cite{acharya2014string} and add redundancy in a fashion similar to that included in Reed-Solomon codes. 

Both lines of work extend the existing literature in string reconstruction~\cite{levenshtein2001efficient, dudik2003reconstruction , batu2004reconstructing, viswanathan2008improved} and coded string reconstruction~\cite{kiah2016codes,gabrys2018unique,cheraghchi2019coded,gabrys2017asymmetric}.

The paper is organized as follows. Section~\ref{sec:ps} introduces the problem, the relevant terminology and notation. The topic of reconstruction codes, or code design for unique reconstruction, is addressed in Section~\ref{sec:recons}. Asymmetric error-correction codes with unique reconstruction properties are addressed in Section~\ref{sec:asymmetric}, while symmetric error-correction code constructions are discussed in Section~\ref{sec:symmetric}. The paper concludes with a discussion of open problems in Section~\ref{sec:open}.

\section{Problem Statement} \label{sec:ps}
Let $\textbf{s}=s_1 s_2 \ldots s_k$ be a binary string of length $k \geq 2$. A substring of $\textbf{s}$ starting at $i$ and ending at $j$, where $1 \leq i \textcolor{black}{\leq} j \leq k,$ is denoted by $\textbf{s}_{i}^{j}$, and is said to have \emph{composition} $0^{z}1^{w}$, where $0 \leq z,w \leq j-i+1$ stand for the number of $0$s and $1$s in the substring, respectively. \textcolor{black}{Let $c(\textbf{s}_{i}^{j})$ denote the composition of $\textbf{s}_{i}^{j}$, $i\leq j$.} A composition only conveys information about the weight of the substring, but not the particular order of the bits. Furthermore, let $C_{l}(\textbf{s})$  stand for the multiset of compositions of substrings of $\textbf{s}$ of length $l$, $1\leq l \leq k$; clearly, this multiset contains $k-l+1$ compositions. For example, if $\textbf{s}=100101$, then the substrings of length two are $10,00,01,10,01$, so that $C_2(\textbf{s})=\{{0^11^1,0^2,0^11^1,0^11^1,0^11^1\}}$.

The multiset $C(\textbf{s})=\cup_{l=1}^{k} C_{l}(\textbf{s})$ is termed the \emph{composition multiset}. It is straightforward to see that the composition multisets of a string $\textbf{s}$ and its reversal, $\textbf{s}^r=s_k s_{k-1} \ldots s_1,$ are identical and hence these two strings are indistinguishable based on $C(\cdot)$.
We define the \emph{cummulative weight} of a composition multiset $C_{l}(\textbf{s}),$ with compositions of the form $0^{z}1^{w}$, where $z+w=l$, as $w_{l}(\textbf{s})=\sum_{0^{z}1^{w} \in C_{l}(\textbf{s})}\, w.$ Observe that $w_{1}(\textbf{s})=w_{k}(\textbf{s})$, as both equal the weight of the string $\textbf{s}$. 
More generally, one has $w_{l}(\textbf{s})=w_{k-l+1}(\textbf{s}), \text{ for all } 1 \leq l \leq k.$ \textcolor{black}{This assertion can be proved by counting the objects of interest in two different ways. One may arrange all substrings of length $\ell$ row-wise. In this case, the columns represent strings of length $k-\ell+1$. The weight counts of the rows have to be the same as those of the columns, so that $w_\ell = w_{k+1-\ell}$.}

In our subsequent derivations, we also make use of the following notation. For a string $\textbf{s}=s_1 s_2 \ldots s_k$, we let $\sigma_i=\text{wt}(s_is_{k-i+1})$ for $i \leq \lfloor \frac{k}{2} \rfloor,$ \textcolor{black}{and for odd $k,$} $\sigma_{ \lceil \frac{k}{2} \rceil}=\text{wt}(s_{ \lceil \frac{k}{2} \rceil})$, where $\text{wt}$ stands for the weight of the string. For our running example $\textbf{s}=100101,$ $\sigma_1=2,$ while $\sigma_2=0$. We use $\Sigma^{ \lceil \frac{k}{2} \rceil}$ to denote the sequence $( \sigma_i)_{i \in [ \lceil \frac{k}{2} \rceil]},$ where $[a]=\{{1,\ldots,a\}}$.

Whenever clear from the context, we omit the argument $\textbf{s}$.

The problems of interest are as follows. The first problem pertains to reconstruction codes: A collection of binary strings of fixed length is called a \textbf{reconstruction code} if all the strings in the code can be reconstructed uniquely based on their multiset compositions. We seek reconstruction codes of small redundancy and consequently, large rate.

As part of the second problem, we consider \textbf{error-correcting reconstruction codes.} In this context, one is given a valid composition multiset of a string $\textbf{s}$, $C(\textbf{s})$. Within the multiset $C(\textbf{s})$, some compositions may be arbitrarily corrupted \textcolor{black}{ to a composition of the same length}. We refer to such errors as \textbf{composition errors}. For example, when $\textbf{s}=100101$, the multiset 
$C_2(\textbf{s})=\{{0^11^1,0^2,0^11^1,0^11^1,0^11^1\}}$ may be corrupted to $\tilde{C}_2(\textbf{s})=\{{\mathbf{0^2},0^2,0^11^1,0^11^1,0^11^1\}}$, in which case we have a single composition error. 

\textcolor{black}{
For the case where multiple composition errors may occur so that the symmetric difference between $C(\textbf{s})$ and $\tilde{C}(\textbf{s})$, denoted $C(\textbf{s}) \bigtriangleup \tilde{C}(\textbf{s})$ may contain more than one element, we will call the errors as asymmetric if the following condition holds: For each $i \in \{1,2, \ldots, \lfloor \frac{n}{2} \rfloor \}$, if $\left| C_i(\textbf{s}) \bigtriangleup \tilde{C}_i(\textbf{s}) \right| \neq 0$, then $\left| C_{n-i+1}(\textbf{s}) \bigtriangleup \tilde{C}_{n-i+1}(\textbf{s}) \right| = 0$. In words, this means that the composition sets $C_i(\textbf{s})$ and ${C}_{n-i+1}(\textbf{s})$ cannot both be in error. For the case of symmetric errors, this condition need not hold (so that there are no restrictions on the structure of the composition errors), and asymmetric composition errors are a special case of symmetric composition errors. For the case where a single composition error occurs between $C(\textbf{s})$ and $\tilde{C}(\textbf{s})$, the single composition error is necessarily asymmetric (and therefore also symmetric).} 

\textcolor{black}{Continuing from our previous example, the multisets $C_2(\textbf{s})$ and $C_5(\textbf{s})$ may be corrupted to $\tilde{C}_2(\textbf{s})=\{{\mathbf{0^2},0^2,0^11^1,0^11^1,0^11^1\}}$ and $\tilde{C}_5(\textbf{s})=\{{\mathbf{0^11^4},0^31^2\}}$, in which case we say that we encountered an example of two \emph{symmetric composition errors}, given that the sum of the substrings lengths, $2$ and $5$, sum up to $k+1=7$. Note that this example does not represent two asymmetric composition errors because an error has occurred in a composition of length $i=2$ and also in composition of length $n-i+1=6-2+1=5$.}

Our main results are summarized below. 

Theorem~\ref{thm:recon} establishes the existence of efficiently decodable reconstruction codes that have asymptotic rate one \textcolor{black}{(proved in Section~\ref{sec:recons})}, while Theorem~\ref{thm:single} establishes similar results for the case of reconstruction codes capable of correcting one composition error \textcolor{black}{(proved in Section~\ref{sec:asymmetricS})}.

\begin{theorem} \label{thm:recon}
There exist efficiently encodable and decodable reconstruction codes with information string-length $k$ and redundancy at most $ \frac{1}{2}\,\log \,(k) + \textcolor{black}{5}$. 
\end{theorem} 

\begin{theorem}  \label{thm:single}
There exist efficiently encodable and decodable reconstruction codes with information string-length $k$ capable of correcting a single composition error and redundancy at most $\frac{1}{2} \log\,(k) +\textcolor{black}{13}$. 
\end{theorem} 

Theorems~\ref{thm:asym},~\ref{th:mainth} and~\ref{th:new} extend the results of Theorem~\ref{thm:single} for the case of multiple composition errors, including both the asymmetric and symmetric case \textcolor{black}{(proved in Sections~\ref{sec:mulerr},~\ref{sec:symmetric}, and~\ref{sec:symmetric} respectively) }. The result in Theorem~\ref{thm:asym} demonstrates the existence of explicit asymmetric error-correcting reconstruction codes of asymptotic rate one that can be efficiently reconstructed for constant $t$. The result in Theorem~\ref{th:mainth} applies to symmetric errors. The best known redundancy is achieved using the construction supporting Theorem~\ref{th:new}.

\begin{theorem} \label{thm:asym}
There exist efficiently encodable and decodable reconstruction codes with information string-length $k$ capable of correcting a constant number of $t$ asymmetric composition errors and redundancy $\cO \left( t \log k \right)$. The decoding algorithm has complexity $\cO(n^3\,2^t)$.
\end{theorem}

\begin{theorem}\label{th:mainth} There exist efficient symmetric $t$-error correcting reconstruction codes with information string-length $k$, redundancy $\cO(t^2 \log k)$ and decoding complexity $\cO(n^3)$.
\end{theorem}

\begin{theorem}\label{th:new} 
There exist symmetric $t$-error correcting reconstruction codes with information string-length $k$, redundancy $\cO(\log k + t)$ and decoding complexity $\cO(n^{3+3t})$.
\end{theorem} 

\subsection{Technical Background} \label{pb}
 \vspace{-0.04in}
\textcolor{black}{
For a string of length $k$, recall that $\sigma_i = \text{wt}(s_i, s_{k+1-i})$, and that given $C_1$ one can compute $w_1 = \sum^{\lceil \frac{k}{2} \rceil}_{j=1} \sigma_j.$ 
When $i=2$, the bits at positions $1$ and $k$ contribute once to $w_2$, whereas the bits $2, \dots, k-1$ all contribute 
twice to $w_2$. \\
Using $C_2,$ we can obtain $\sigma_1 + 2 \sum^{ \lceil \frac{k}{2} \rceil}_{j=2} \sigma_j = w_2.$ 
Generalizing this result for all $C_i, i\leq \lceil \frac{k}{2} \rceil$ is straightforward, and gives the following equalities: 
\begin{equation} \label{eq:sigmas}
\frac{1}{i} \sigma_1 + \frac{2}{i} \sigma_2 + \dots + \frac{i-1}{i} \sigma_{i-1} +  \sigma_i + \sigma_{i+1} + \dots +  \sigma_{\lceil \frac{k}{2} \rceil} = \frac{1}{i} w_{i}. 
\end{equation}
The above system of $\lceil \frac{k}{2} \rceil$ linear equations with $\lceil \frac{k}{2} \rceil$ unknowns can be solved efficiently. Thus, for all error-free composition sets, one can find $\Sigma^{\lceil \frac{k}{2} \rceil}$. }

Some of our code designs rely on the Backtracking algorithm~\cite{acharya2014string}, first used in the context of the Turnpike problem. We provide an example illustrating the operation of the algorithm. \textcolor{black}{The composition multiset $C(\textbf{s})$ of a string is given as the input to the algorithm, and its output is the set of all strings that have the same composition as $C(\textbf{s})$.}
\begin{example} Let $\textbf{s} = 1010001010$. The sequence $\Sigma^5 = (\sigma_1 = 1,\sigma_2 = 1,\sigma_3=1,\sigma_4=1,\sigma_5=0)$ can be uniquely determined from the composition multiset. This follows from $w_1 = \sum_{i=1}^{5} \sigma_{i}$ and $i\,w_1 - w_i = \sum^{i-1}_{j=1} (i-j) \sigma_j$, for $i=1,\ldots, \lceil \frac{k}{2} \rceil$. Solving the system of equations produces $\Sigma^{5}$.  

The Backtracking algorithm starts by determining the first and the last bit of the string and then proceeds to place the remaining bits in an inward fashion. Since $\sigma_1 = \emph{wt}(s_1 s_{10})$ is known, and since a string and its reversal have the same composition multiset, the first and the last bits are placed arbitrarily. In our example, without loss of generality, the Backtracking algorithm sets $s_1=1$ and $s_{10}=0$. 

Let $\ell_r$ be the length of the reconstructed prefix/suffix pair.
Backtracking produces a multiset of all compositions that are jointly determined by the reconstructed prefix and suffix of length $\ell_r =1,$ $\textbf{s}_1^1 = 1$, $\textbf{s}^{10}_{10} = 0$ and $\Sigma^{5}$. Denote this multiset by $\mathcal{T}_{\ell_r=1}$. 

Note that $\sigma_5 = 0 $ implies that the composition of $\textbf{s}^6_5$ is $ 0^2$. Similarly, $\sigma_4 = 1$ and $\sigma_5 = 0$ imply that the composition of $\textbf{s}_4^7 $ is $ 0^31$. Thus, using the information in $\Sigma^5$ alone one can reconstruct the following compositions: $0^61^4, 0^51^3, 0^41^2, 0^31^1, 0^2$.  
Note that compositions of substrings of the form $\textbf{s}_i^j$ can be reconstructed provided that $i,j$ satisfy: (1) $i \textcolor{black}{ \leq } j \leq \ell_r$ or (2) $k +1 - \ell_r \leq i \textcolor{black}{ \leq } j$ or (3) \textcolor{black}{$i \leq \ell_r +1$ and $j \geq k- \ell_r$}. Thus, the composition $0^51^4$ of $\textbf{s}_1^{9}$ and the composition $0^61^3$ of $\textbf{s}_2^{10}$ can both be reconstructed as well. Consequently, $\mathcal{T}_1 = \{0^61^4, 0^51^4, 0^61^3, 0^51^3, 0^41^2, 0^31^1, 0^2, 0, 1 \}$.  

In the next step, the Backtracking algorithm tries to determine the bits $s_2$ and $s_9$. First, recall that $\sigma_2 = 1$ is known. The algorithm determines the compositions of the two longest substrings in the multiset $C \setminus \mathcal{T}_1$ to be $\{ 0^51^3, 0^51^3\}$. Observe that these compositions must be those of the substrings $\textbf{s}_1^8$ and $\textbf{s}_3^{10}$ (although inconsequential for this example, it is still important to note that in general one does not know which one of the two largest compositions in $C \setminus \mathcal{T}_1$ correspond to the prefix). Hence, the compositions of the prefix-suffix pair $\{\textbf{s}_1^2, \textbf{s}_9^{10}\}$ equal $\{ 01,01\}$.   

Since the weight of the reconstructed prefix is not equal to the weight of the reconstructed suffix, \textit{i.e.}, $\emph{wt}(\textbf{s}_1^1) = 1 \not = 0 = \emph{wt}(\textbf{s}_{10}^{10})$, the Backtracking algorithm outputs $s_2 = 0, s_9 = 1$. This follows due to the fact that given that the reconstructed prefix-suffix pair have a weight mismatch, setting $(s_2 = 0, s_9 = 1)$,  or setting $(s_2 = 1, s_9 = 0 )$ leads to different prefix-suffix compositions. As a result, $\{1^2 , 0^2 \} \not = \{01,01 \}$. The algorithm completes this iteration by updating $\mathcal{T}$ to $\mathcal{T}_{\ell_r=2} = \{ 0^61^4, 0^51^4, 0^61^3,$ $0^51^3, 0^51^3 ,0^51^3, 0^41^3,$ $0^51^2, 0^41^2, 0^31^1, 0^2, 01,01,0,1, \textcolor{black}{0,1}\}$.

In the next iteration, following the same steps described above, the compositions of the prefix-suffix pair of length $3$ are found to be $\{ 01^2, 0^21 \}$. However, since $\emph{wt}(\textbf{s}_1^2) = \emph{wt}(\textbf{s}_9^{10})$, the Backtracking algorithm cannot determine the bits $s_3, s_8$. Thus, whenever $\emph{wt} (\textbf{s}_1^{\ell_r}) = \emph{wt} (\textbf{s}_{k+1- \ell_r}^k) $ and \textcolor{black}{$\sigma_{\ell+1} =1,$ } the algorithm \emph{guesses} the bits $s_{\ell_r +1}, s_{n-\ell_r}$. However, if $\emph{wt} (\textbf{s}_1^{\ell_r}) = \emph{wt} (\textbf{s}_{k+1- \ell_r}^k) $ and \textcolor{black}{$\sigma_{\ell+1} \in \{ 0,2\},$} then the reconstruction of bits $s_{\ell_r +1}$ and $s_{n-\ell_r}$ is straightforward. For example, guessing that $s_3=0$, and $s_7=1$ leads to an error. The error is detected by encountering a multiset $\mathcal{T}_{\ell_r}$ that is incompatible with the composition multiset $C$ of the given string. Upon detection of an error, the algorithm backtracks to the \textcolor{black}{last} position where it guessed the bit assignment, changes its guess and restarts the algorithm from that iteration. In our example, this leads to $s_3=1$ and \textcolor{black}{ $s_8=0$}, and one hence obtains the reconstructed string $1010001010$. 
\end{example}

The complexity of the Backtracking algorithm is summarized in the following theorem.
\begin{theorem*} \emph{\cite[Theorem ~32]{acharya2014string}} 
Let $$ \ell_s \overset{\text{def}}{=} | \{ i \leq \lfloor \frac{n}{2} \rfloor: \emph{wt}(\textbf{s}_1^i) = \emph{wt}(\textbf{s}_{n+1-i}^{n}) \text{ and } s_{i+1} \not  = s_{n-i} \} |, $$
$$ E_s \overset{\text{def}}{=} \{ \textbf{v}: C(\textbf{v})=C(\textbf{s})\},\; \ell_s^{*} \overset{\text{def}}{=} \max_{u \in E_s} \ell_{u}.
$$ 
For a given input $C(\textbf{s})$ and $\ell_s$, the Backtracking algorithm outputs a set of strings that contains $\textbf{s}$ in time $\mathcal{O}(2^{\ell_s}n^{2} \log\,(n))$. Furthermore, $E_s$ can be recovered in time $\mathcal{O}(2^{\ell^*_s}n^{2} \log\,(n))$.
\end{theorem*}
If a string has a length that does not allow for unique reconstruction up to reversal, the algorithm returns a set of strings and in the process backtracks multiple times. Backtracking is possible even when the string is uniquely reconstructable, but a condition that ensures that the algorithm does not backtrack is that no prefix has a matching suffix of the same length and same weight. If the algorithm does not backtrack, the string has to be unique. This observation is crucial for our subsequent constructions and it motivates the use of Catalan-Bertrand paths discussed in what follows.

\begin{theorem}\text{\emph{ \textcolor{black}{(Whitworth [1878] , Bertrand [1887])}}} \label{thm:BW}
Among all strings comprising \textcolor{black}{ of } $a$ $0$s and $b$ $1$s, where $a \geq b$, there are ${a+b \choose a} - {a+b \choose a+1}$ strings in which every prefix has at least as many $0$s as $1$s. Note that when $a=b=h$, 
$${a+b \choose a} - {a+b \choose a+1} = \frac{1}{h+1} {2h \choose h} =C_h.$$ 
The number $C_h$ is known as the $h^{\text{th}}$ \emph{Catalan number}. \textcolor{black}{Note that the central binomial coefficient ${2h \choose h}$ also counts the number of strings of length $2h$ whose every prefix has at least as many $0$s as $1$s. Furthermore, note} that the scaled central binomial coefficient $\frac{1}{2}{2h \choose h}$ counts the number of strings of length $2h$ whose every prefix contains strictly more $0$s than $1$s. 
\end{theorem} 

\textcolor{black}{The second part of Theorem~\ref{thm:BW} is proved in Appendix~\ref{app:BW}.}

Strings that have the property that their every prefix contains \emph{strictly} more $0$s than $1$s are henceforth referred to as \emph{Catalan-Bertrand} strings.

We also find the following bounds on the central binomial coefficient useful in our subsequent derivations. 
\begin{proposition} \label{prop2}
\emph{The central binomial coefficient may be bounded~\cite{speyer2016upper} as:}
\begin{equation}
    \frac{2^{2h}}{\sqrt{\pi (h+1)}}  \leq {2h \choose h} \leq \frac{2^{2h}}{\sqrt{\pi h}} ,\; \forall \, h \geq 1.
\end{equation}
\end{proposition}

\section{Reconstruction Codes} \label{sec:recons}

We describe next a family of efficiently encodable and decodable reconstruction codes that map strings of any length $k$ into strings of length $n \leq k+  \frac{1}{2} \log\,(k) + \textcolor{black}{ 5}$. 

\textcolor{black}{Recall that for a given string of length $n$, the system of $\lceil \frac{n}{2} \rceil$ linear equations with $\lceil \frac{n}{2} \rceil$ unknowns given by~\eqref{eq:sigmas} can be solved efficiently. Thus, for all error-free composition sets, one can find $\Sigma^{\lceil \frac{n}{2} \rceil}$. Therefore, the problem of interest is to determine $\textbf{s}$ given $\Sigma^{\lceil \frac{n}{2} \rceil}$ and $C(\textbf{s})$. 
Note that when $\text{wt}(\textbf{s}_1^i) \not = \text{wt}(\textbf{s}^{n}_{n +1 -i})$,~\cite[Lemma~31]{acharya2014string} asserts that $C(\textbf{s}), \textbf{s}_1^i,$ and $\textbf{s}_{n-i+1}^n$ determine the ordered pair $(s_{i+1},s_{n-i})$. }

The previous lemma ~\cite[Lemma~31]{acharya2014string} will be used to guide our construction of a reconstructible code based on Catalan-Bertrand strings. We proceed as follows. Let $I \subseteq [n]$. The string formed by concatenating bits at positions in $I$ in-order is denoted by $\textbf{s}_{I}$. We define a reconstruction code $\mathcal{S}_R(n)$ of even length $n$ as follows:
\begin{align}
    \mathcal{S}_R(n) = &\{ \textbf{s} \in \{ 0,1\}^{n}, s_1 =0, s_n =1, \text{ and } \label{set11} \\
 & \exists \; I \subseteq \{ \textcolor{black}{1}, 2, \dots, n-1, \textcolor{black}{n}\} \text{ such that} \nonumber \\ 
 & \qquad \qquad \qquad \quad \quad \quad \quad \text{ for all } i \in I, s_i \not = s_{n+1-i}, \nonumber\\
 & \qquad \qquad \qquad \quad \quad \quad \quad  \text{           } \text{for all }  i \not \in I, s_i = s_{n+1-i}, \nonumber \\
                        &\textbf{s}_{[\frac{n}{2}] \cap I}  \text{ is a Catalan-Bertrand string}.         \nonumber                \} 
                        \end{align}

For odd $n$, we define the codebook as $$\mathcal{S}_R(n)=\cup_{\textbf{s} \in  \mathcal{S}_R(n-1)} \{\textbf{s}_1^{\frac{n-1}{2}}\; 0 \; \textbf{s}_{\frac{n+1}{2}}^{n-1}, \; \textbf{s}_1^{\frac{n-1}{2}} \; 1 \, \textbf{s}_{\frac{n+1}{2}}^{n-1}\}.$$ 

The following proposition is an immediate consequence of the above construction.
\begin{lemma} \label{lem1}
Consider a string $\textbf{s}\in \mathcal{S}_R(n)$. For all prefix-suffix pairs of length $ 1\leq  j \leq \frac{n}{2}$, one has $\emph{wt}(\textbf{s}_1^j) \not = \emph{wt}(\textbf{s}^{n}_{n +1 -j})$.
\end{lemma}
The proof of Theorem~1 follows from the fact that $\mathcal{S}_R(n)$ is a reconstruction code, which may be easily established from the guarantees for the Backtracking algorithm and~Lemma \ref{lem1}.

\textcolor{black}{For $n$ even,} the size of $\mathcal{S}_R(n)$ may be bounded as: \textcolor{black}{
\begin{align*}
    |\mathcal{S}_R(n)| = \sum_{i=0}^{\frac{n}{2} -1} {\frac{n}{2} - 1 \choose i} 2^{\frac{n}{2} -1 -i} {i \choose  \lfloor \frac{i}{2} \rfloor } 
     \geq \frac{ 2^{n-3} }{\sqrt{ \pi n}} \, . 
\end{align*} }
The first equality follows from the description of the codebook, while the second inequality follows from Proposition \ref{prop2} and the binomial theorem. \textcolor{black}{For odd $n,$ $ |\mathcal{S}_R(n)| = 2  |\mathcal{S}_R(n-1)| \geq \frac{ 2^{n-3} }{\sqrt{ \pi n}}.$ Further details are provided in Apprendix~\ref{app:derivation}.} As $2^k \leq |\mathcal{S}_R(n)|$, simple algebraic manipulation reveals that the redundancy of the reconstruction code for information lengths $k$ is at most 
$1/2  \log\,(k) +5$ \textcolor{black}{for all $k \geq 8$}. 

\textcolor{black}{Given an information string of length $k$, the encoding algorithm returns a reconstructable string of length $n$. The encoding algorithm that accompanies our reconstruction codebook (a bijective map between the set of all information strings of length $k$ and a subset of the reconstructable strings of length $n$) can be implemented using simple lexicographical rankings of Catalan-Bertrand strings and symmetric strings. This encoding technique requires $\mathcal{O}(n^2)$ time (see Appendix~\ref{app:construction} for details). However, as described in~\cite{durocher2012cool}, there exist other ordering-based constructions for Catalan strings that may be used to further increase the efficiency of the encoding algorithm. The Backtracking algorithm reconstructs the coded string in $\mathcal{O}(n^3)$ time. The coded string is then mapped to the information string via the inverse encoding map, which takes an additive $\mathcal{O}(n^2)$ time. Thus, the overall reconstruction time remains $\mathcal{O}(n^3)$.  This concludes the proof of Theorem~1.}

\section{Error-Correcting Reconstruction Codes: The Asymmetric Setting} \label{sec:asymmetric}

For clarity of exposition, we will start with a discussion of single error-correcting reconstruction codes, as they illustrate the use of Catalan-Bertrand paths and are conceptually easy to extend for the case of multiple composition errors. Our reconstruction codes with composition error-correcting capabilities are derived using the interleaving procedure described in the previous section, and they require adding an additional logarithmic number of redundant bits to recover the sequence $\Sigma^{\lceil \frac{n}{2} \rceil}$.

\subsection{Single Error-Correcting Reconstruction Codes}  \label{sec:asymmetricS}
 
Let $\mathcal{S}_{R}(n-2)$ be the code of odd length $n-2$ \textcolor{black}{, $\lceil \frac{n}{2} \rceil $ divisible by three, as} described in the previous section. Then, a single (symmetric or asymmetric) composition error-correcting code of length $n$, $\mathcal{S}^{(1)}_{C}(n),$ can be constructed by adding two bits to each string in $\mathcal{S}_R(n-2)$ and subsequently fixing the value of one additional bit. These three redundant bits allow us to uniquely recover the sequence $\Sigma^{\lceil \frac{n}{2}\rceil }$ in the presence of a single composition error. \textcolor{black}{As will be seen from our subsequent derivations, given $\Sigma^{\lceil \frac{n}{2} \rceil}$ and the erroneous composition set of $\textbf{s}$, one can reconstruct $\textbf{s}$.} 

To prove Theorem~2, let $\tilde{C}(\textbf{s})$ denote the set obtained by introducing a single error in the composition set $C(\textbf{s})$ of a string $\textbf{s}$ of length $n$. Recall that $w_j$ stands for the cumulative weight of compositions of length $j$ in $C$ and that $w_j = w_{n-j+1}$. Let $\tilde{w}_j$ denote the cumulative weight of compositions in $\tilde{C}_j$. It is straightforward to prove the following proposition.

\begin{proposition}\label{prop:wsback} Let \textcolor{black}{$j \in \left[ \lceil \frac{n}{2} \rceil \right]$}. Then, 
$$w_j  \textcolor{black}{=}  j w_1 -   \sum_{i=1}^{j-1} i \, \sigma_{j-i},$$
\textcolor{black}{which implies}
\textcolor{black}{$$ j w_1 -   \sum_{i=2}^{j-1} i \, \sigma_{j-i}  \leq w_j \leq  j w_1 -   \sum_{i=2}^{j-1} i \, \sigma_{j-i} + 2.$$}
\end{proposition}
\begin{proof}
\textcolor{black}{
Note that $w_j = \sum_{i=1}^{j-1} i \sigma_i + \sum_{i=j}^{\lceil \frac{n}{2} \rceil} j \sigma_i$. Since $w_1 = \sum_{i=1}^{\lceil \frac{n}{2} \rceil} \sigma_i$, we have
\begin{align*}
    w_j =& \sigma_1 + 2 \sigma_2 + \cdots + (j-1) \sigma_{j-1} + j \sum_{i=j}^{\lceil \frac{n}{2} \rceil} \sigma_j \\
    =& j \left( \sum_{i=1}^{\lceil \frac{n}{2} \rceil} \sigma_i \right) - \sigma_{j-1} - 2 \sigma_{j-2} - \cdots - (j-1) \sigma_1. \\
    =& j w_1 - \sum_{i=1}^{j-1} i \sigma_{j-i}.
\end{align*}}
\end{proof}
This result immediately implies the next proposition.
\begin{proposition}\label{prop:cor} Let \textcolor{black}{$j \in \left[ \lceil \frac{n}{2} \rceil \right]$} and suppose that we are given $w_1, \sigma_1, \ldots, \textcolor{black}{\sigma_{j-2}}$. Then, the value $w_j \bmod 3$ uniquely determines $w_j$.
\end{proposition}

We also need the following three propositions. 

\begin{proposition}\label{prop:w1} Given $\emph{wt}(\textbf{s}) \bmod 2,$ $\tilde{w}_n$ and $\tilde{w}_1$, one can recover $w_1$. 
\end{proposition}
\begin{proof} If $\tilde{w}_n = \tilde{w}_1$, then clearly $w_1=\tilde{w}_n = \tilde{w}_1$. Hence, suppose that $\tilde{w}_n \neq \tilde{w}_1$ and observe that $|\tilde{w}_1 - w_1| \leq 1$. 
The last inequality follows since at most one composition error is allowed. If $\tilde{w}_1 \bmod 2 = \text{wt}(\textbf{s}) \bmod 2$, then $w_1 =  \tilde{w}_1$; 
otherwise, $w_1 = \tilde{w}_n$.
\end{proof}

\begin{proposition}\label{prop:change} Suppose that $n$ is odd and that either $ \lceil \frac{n}{2} \rceil + 1$ or $ \lceil \frac{n}{2} \rceil$ is divisible by $3$. Assume that $\textbf{s} = s_1 \, \ldots \, s_{ \lceil \frac{n}{2} \rceil} \, \ldots s_n,$ and let $\textbf{s}' = s_1 \, \ldots \, 1-s_{\lceil \frac{n}{2} \rceil} \, \ldots \, s_n$. Then,
\begin{align*}
\sum_{i=1}^{ \lceil \frac{n}{2} \rceil} w_i(\textbf{s}) \equiv \sum_{i=1}^{\lceil \frac{n}{2} \rceil} w_i(\textbf{s}') \bmod 3.
\end{align*}
\end{proposition}
\begin{proof} Suppose that $s_{\lceil \frac{n}{2} \rceil}=1$. Then, the bit $s_{ \lceil \frac{n}{2} \rceil}$ contributes 
$ \lceil \frac{n}{2} \rceil$ to $w_{ \lceil \frac{n}{2} \rceil}$ and $ \lceil \frac{n}{2}  \rceil-1$ to $w_{ \lceil \frac{n}{2}\rceil -1}$.

In summary, 
if $s_{\lceil \frac{n}{2} \rceil}=1$, then 
$$ \sum_{i=1}^{ \lceil \frac{n}{2} \rceil} w_i(\textbf{s}) = \sum_{i=1}^{ \lceil \frac{n}{2} \rceil} w_i(\textbf{s}') + \frac{ \lceil \frac{n}{2} \rceil  \, ( \lceil \frac{n}{2} \rceil +1)}{2}.$$
The result follows if either $ \lceil \frac{n}{2} \rceil + 1$ or $ \lceil \frac{n}{2} \rceil $ is divisible by $3$.
\end{proof}

\begin{proposition}\label{prop:change2} For odd $n,$ if $s_1 \, \ldots \, s_{ \lceil \frac{n}{2} \rceil} \, \ldots \, s_n \in \mathcal{S}_R(n)$, 
then $s_1 \, \ldots \, 1-s_{ \lceil \frac{n}{2} \rceil} \, \ldots \, s_n \in \mathcal{S}_R(n)$.
\end{proposition}
\textcolor{black}{The proof of the proposition is straightforward, as it follows from the definition of the reconstruction set $\mathcal{S}_R(n)$.}

\textcolor{black}{For odd $n$ such that $\lceil \frac{n}{2} \rceil \equiv 0 \mod 3$ our code is defined as follows:}
\begin{align*}
\mathcal{S}^{(1)}_{C}(n) = \Big \{ &\textbf{s}=s_1 s_2 s_3 \ldots s_{ \lceil \frac{n}{2} \rceil} \ldots  s_{n-2} s_{n-1} s_{n} \in \{0,1\}^n :  \\
&\ \ \ \ \ s_1  \textbf{s}_{3}^{n-2} s_n =  s_1 s_3 s_4 s_5 \ldots s_{n-4} s_{n-3}s_{n-2} s_n \in \mathcal{S}_R(n-2), \\
&\ \ \ \ \ \text{wt}(\textbf{s}) \bmod 2 = 0,   \\
&\ \ \ \ \ \sum_{i=1}^{ \textcolor{black}{\lceil \frac{n}{2} \rceil} } w_i(\textbf{s}) \equiv 0 \bmod 3, \, \text{ where } s_2 \leq s_{n-1} \Big \}. 
\end{align*}

The size of the code $\mathcal{S}^{(1)}_{C}(n)$ is $\frac{|\mathcal{S}_R(n-2)|}{2}$, which follows \textcolor{black}{from the second constraint that $s_1 s_3 \ldots s_{n-2} s_n \in \mathcal{S}_R(n-2)$, along with Proposition~\ref{prop:change2}}. To construct a string in $\mathcal{S}^{(1)}_{C}(n)$, we first fix $s_2$ and $s_{n-1}$ so that $\sum_{i=1}^{ \lceil \frac{n}{2} \rceil} w_i(\textbf{s}) \equiv 0 \bmod 3$. \textcolor{black}{Then, we choose $s_{ \lceil \frac{n}{2} \rceil}$ to satisfy} $\text{wt}(\textbf{s}) \equiv 0 \bmod 2 $. From Propositions~\ref{prop:change} and \ref{prop:change2}, the resulting string belongs to $\mathcal{S}^{(1)}_{C}(n)$. \textcolor{black}{The next proposition shows that for certain codelengths, there exists values for $s_2$ and $s_{n-1}$ that always allow for the constraints to be satisfied.}

\textcolor{black}{\begin{proposition} When $\lceil \frac{n}{2} \rceil$ is divisible by $3$, then for any $\bx = x_1 \ldots x_{n-2} \in \{0,1\}^{n-2}$, there exists $s_2 s_{n-1} \in \{0,1\}^2$ so that
\begin{align*}
 \sum_{i=1}^{ \lceil \frac{n}{2} \rceil } w_i\Big( x_1 s_2 x_2 x_3 \ldots x_{n-3} s_{n-1} x_{n-2} \Big) \equiv 0 \bmod 3,
\end{align*}
where $s_2 \leq s_{n-1}$.
\end{proposition}}
\textcolor{black}{\begin{IEEEproof} Let $\bs =  x_1 s_2 x_2 x_3 \ldots x_{n-3} s_{n-1} x_{n-2}$. Clearly, the elements $s_2$ and $s_{n-1}$ appear in exactly one composition from $C_1(\bs)$ (recall that $C_i(\bs)$ denotes the set of compositions of $\bs$ of length $i$). Furthermore, $s_2$ and $s_{n-1}$ each appear twice in every set $C_i(\bs),$ where $\lceil \frac{n}{2} \rceil \geq i \geq 2$. Therefore the symbol $s_2$ appears in
$$ 2\left( \lceil \frac{n}{2} \rceil -1 \right) + 1$$
compositions from $C_1(\bs) \cup C_2(\bs) \cup \cdots \cup C_{\lceil \frac{n}{2} \rceil}(\bs)$, and, by symmetry, the symbol $s_{n-1}$ appears $2\left( \lceil \frac{n}{2} \rceil -1 \right ) + 1$ times as well. \\
Suppose $\sum_{i=1}^{\lceil n/2 \rceil} w_i(s) \equiv a \bmod 3 $ when $(s_2, s_{n-1}) = (0,0)$. Then, more generally if $(s_2, s_{n-1}) = (c_1, c_2) \in \{0,1\}^2$ where $(c_1, c_2)$ are not necessarily equal to $(0,0)$ we have,
\begin{align*}
\sum_{i=1}^{ \lceil \frac{n}{2} \rceil } w_i( \bs ) &\equiv a +c_1 \Big( 2 \big( \lceil \frac{n}{2} \rceil - 1 \big)+ 1 \Big) + c_2 \Big( 2 \big( \lceil \frac{n}{2} \rceil - 1 \big)+ 1 \Big)   \bmod 3 \\
& \equiv a - c_1 - c_2 \bmod 3.
\end{align*}
Since for the case $(c_1, c_2)=(0,0)$, $\sum_{i=1}^{ \lceil \frac{n}{2} \rceil } w_i( \bs ) \equiv a \bmod 3$, for $(c_1, c_2) = (0,1)$, $\sum_{i=1}^{ \lceil \frac{n}{2} \rceil } w_i( \bs ) \equiv a - 1 \bmod 3$, and for $(c_1, c_2) = (1,1),$ $\sum_{i=1}^{ \lceil \frac{n}{2} \rceil } w_i( \bs ) \equiv a - 2 \bmod 3$. This completes the proof.
\end{IEEEproof}}

For the next lemma, recall that $\tilde{C}(\textbf{s})$ is the result of a single composition error in $C(\textbf{s})$.

\begin{lemma} Suppose that $\textbf{s} \in \mathcal{S}^{(1)}_{C}(n)$ \textcolor{black}{where $\lceil \frac{n}{2} \rceil$ is divisible by $3$}. Then, given $\tilde{C}(\textbf{s})$, one can recover $\Sigma^{ \textcolor{black}{\lceil \frac{n}{2} \rceil } }$. \label{finding_sigma}
\end{lemma}
\begin{proof} In order to prove the claim, we show that given $\tilde{C}(\textbf{s}),$ one can recover $w_1, w_2, \ldots, w_n,$ which we know uniquely determine $\Sigma^{ n/2 }$ according to (\ref{eq:sigmas}). Let $j$ be such that $\tilde{w}_j \neq \tilde{w}_{n+1-j}$. Since at most one 
single composition error is allowed, there exists at most one such $j$. It is straightforward to see that due to symmetry, 
either $\tilde{w}_j \neq w_j=w_{n+1-j}$ or $\tilde{w}_{n+1-j} \neq w_j=w_{n+1-j}$. Since $\text{wt}(\textbf{s}) \bmod 2 = 0$ by construction, 
it follows that we can determine $w_1$ based on Proposition~\ref{prop:w1}. \textcolor{black}{In addition, using the first identity from Proposition~\ref{prop:wsback}, it follows that we can recover $\sigma_1, \sigma_2, \ldots, \sigma_{j-2}$.}  \textcolor{black}{Also, using the constraint $\sum_{i=1}^{ \lceil \frac{n}{2} \rceil } w_i( \bs ) \equiv 0 \bmod 3$}, we can recover 
\textcolor{black}{$w_j \bmod 3$. Then, according to Proposition~\ref{prop:cor}, we can recover $w_j$ along with $w_1, \ldots, w_n$.} 
One case left to consider is when $\tilde{w}_i = \tilde{w}_{n+1-i},$ for all \textcolor{black}{$i \in [\lfloor \frac{n}{2} \rfloor ]$}. In this case, $\tilde{w}_{ \textcolor{black}{\lceil \frac{n}{2} \rceil} } \neq w_{ \textcolor{black}{\lceil \frac{n}{2} \rceil} }$. Applying Proposition~\ref{prop:cor} allows us to determine $w_{ \textcolor{black}{\lceil \frac{n}{2} \rceil} }$ for this case as well. This completes the proof.
\end{proof}

Next, recall that $\mathcal{T}_i$ stands for the set of compositions of all substrings \textcolor{black}{$\textbf{s}_j^l$ for which $j<l\leq i,$ or $n+1-i \leq j<l,$ or $j \leq i +1\textit{ and } n-i \leq l$, or $l = n+1-j$.} 

Let the two strings $\textbf{s}$ and $\textbf{v}$ be such that $\textbf{s}_1^j = \textbf{v}_1^j$ and $\textbf{s}_{n+1-j}^n = \textbf{v}_{n+1-j}^n$ and either $s_{j+1} \neq v_{j+1}$ or $s_{n-j} \neq v_{n-j}$. Then we say that the \emph{longest prefix-suffix pair} shared by the two strings has length $j$.

\textcolor{black}{Before we proceed to prove that $\mathcal{S}^{(1)}_C(n)$ is a single error-correcting code, we provide two illustrative examples - one for the case where the error occurs in a composition of length (size) $j \leq \lfloor \frac{n}{2} \rfloor,$ and another for the case where the error occurs in a composition of length (size) $j \geq \lceil \frac{n}{2} \rceil.$ 
\begin{example}
Let $n=11$ and consider $\textbf{s} = 00001111111 \in \mathcal{S}_C^{(1)}(n).$ Let the composition multiset with one composition error be $\tilde{C}(\textbf{s}) = ( C(\textbf{s})  \cup \{  1^4 \} ) \setminus \{ 0^4 \}.$ Given $\tilde{C}(\textbf{s}),$ by Lemma 2, we can infer $\Sigma^6 = (1,1,1,1,2,1).$ The Backtracking algorithm readily reconstructs up to $s_1s_2s_3 = 000$, and $s_{9}s_{10}s_{11} = 111.$ For further details on the Backtracking algorithm, refer to Example 1 and \cite{acharya2014string}. Next, observe that $w_4 \not = w_{8},$ and that the two largest compositions in $\tilde{C}(\textbf{s}) \setminus \mathcal{T}_3 = \{ 0^41^3, 1^7 \}$ are compatible with the reconstructed prefix, suffix and the constraints imposed by $\Sigma^6.$ Thus, the Backtracking algorithms proceeds by reconstructing $s_4s_{8} = 01,$ and computing $\mathcal{T}_4.$ Note that for this string, due to the constraints imposed by $\sigma_5, $ and $\sigma_6, $ the string is immediately reconstructed as $00001111111$. The Backtracking algorithm finds that $0^4 \in \mathcal{T}_4,$ but $0^4 \not \in \tilde{C}(\textbf{s}).$ However, this one single incompatibility is expected in the given error setting. In general, for the case of a single composition error, if the error occurs in a composition that corresponds to a substring of length $\leq \lfloor \frac{n}{2} \rfloor,$ it does not affect the Backtracking algorithm.
\end{example}
\begin{example}
Let $n=11$ and once again consider $\textbf{s} = 00001111111  \in \mathcal{S}_C^{(1)}(n).$ Let the composition multiset with one composition error be $\tilde{C}(\textbf{s}) = ( C(\textbf{s}) \cup  \{ 01^6 \} ) \setminus \{ 1^7 \}.$ Given $\tilde{C}(\textbf{s}),$ by Lemma 2, we can infer $\Sigma^6 = (1,1,1,1,2,1).$ The Backtracking algorithm readily reconstructs up to $s_1s_2s_3 = 000$, and $s_{9}s_{10}s_{11} = 111.$  Note that $w_4 \not = w_{8},$ and that the two largest compositions in $\tilde{C}(\textbf{s}) \setminus \mathcal{T}_3$ equal $\{ 0^41^3, 01^6\},$ implying $s_4s_8 = 00.$ Clearly, one of these two compositions is erroneous as $\sigma_4 =1$. Hence, the two possibilities for the bits $s_4s_{8} $ are $10$ or $01.$ If the Backtracking algorithm assigns $s_1s_2s_3s_4 = 0001,$ and $s_{8}s_{9}s_{10}s_{11} = 0111,$ then note that while $\mathcal{T}_4$ contains only one $1^4,$ $\tilde{C}(\textbf{s})$ contains four $1^4.$ In particular, the number of distinct elements in the sets $\tilde{C}(\textbf{s})$ and $\mathcal{T}_4$ due to incorrect bit assignments is strictly greater than one. Thus, if the Backtracking algorithm erroneously reconstructs the bits $s_4s_{8},$ it backtracks and assigns $s_4s_{8}=01$ instead. As mentioned in Example 2, due to the constraints imposed by $\sigma_5, $ and $\sigma_6,$ the string is then correctly reconstructed as $00001111111$.   
\end{example}
}

\begin{lemma} \label{lem:3}
Let $\textbf{s} \in \mathcal{S}^{(1)}_{C}(n)$.
Given $\tilde{C}(\textbf{s})$, one can uniquely reconstruct the string $\textbf{s}$. 
\end{lemma}

\begin{proof}
Let $j$ denote the index of the composition multiset $C_j$ that contains an error. \textcolor{black}{As shown in the example, single composition errors that occur in a composition of a substring of length $j \leq  \lfloor \frac{n}{2} \rfloor$ do not affect the reconstruction process, since the Backtracking algorithm only makes use of information provided by compositions of substrings of length $\geq \lceil \frac{n}{2} \rceil $. As the Backtracking algorithm progresses, the erroneous composition is identified through a comparison of the erroneous observed composition multiset and the iteratively constructed set $\mathcal{T}_\ell,$ as explained in the above examples. Errors that happen for $j \leq \lfloor\frac{n}{2} \rfloor$ are easily identified and automatically corrected by the Backtracking algorithm.} From Lemma~\ref{finding_sigma}, $\Sigma^{\lceil \frac{n}{2} \rceil}$ may be determined in an error-free manner. Using the obtained $\Sigma^{\lceil \frac{n}{2} \rceil}$, we run the Backtracking algorithm and in the process, we possibly run into incompatible compositions for $j \geq  \lceil \frac{n}{2} \rceil $. \textcolor{black}{Note that when $j = \lceil \frac{n}{2} \rceil$, the Backtracking algorithm has reconstructed the entire string. Given an already reconstructed prefix-suffix pair of length $i$, $\textbf{s}_{1}^{i}, \textbf{s}_{n-i+1}^{n}$, we make use of the two largest cardinality compositions in $\tilde{C}(\textbf{s}) \setminus \mathcal{T}_i$ to reconstruct the bits $s_{i+1}$ and $s_{n-i}.$ If $\sigma_{i+1} \in \{0,2\},$ $s_{i+1}$ and $s_{n-i}$ can be determined immediately. Also, given $\Sigma^{\lceil \frac{n}{2} \rceil},$ if the Backtracking algorithm can correctly determine which of the above two compositions corresponds to a prefix and a suffix, then $s_{i+1}$ and $s_{n-i}$ can be uniquely identified. Otherwise, the Backtracking algorithm halts and performs a guess. Thus, in summary, an error occurring in a composition in $C_j$ affects the reconstruction of the bits indexed by $i+1$ and $n-i,$ where $i$ is such that $j+i+1 = n.$}  

\textcolor{black}{Consider the case where the} incompatibility manifests itself through $\mathcal{T}_i \not \subset \tilde{C}$, where $j = n-i-1$. Here, we identify the element that is in $\mathcal{T}_i$ but not in $\tilde{C}_j$, and add its weight to $\tilde{w}_j$ and compare it with $\tilde{w}_{n+1-j}$; this allows us to identify the erroneous composition. \textcolor{black}{The assumption is that a composition corresponding to a substring of length $j$ is erroneous. Clearly, $\tilde{C}_j = (C_j \setminus \{{c_{i_1}\}}) \cup \{{c_{i_2}\}},$ for some compositions indexed by $i_1$ and $i_2$, where $c_{i_1}$ corresponds to the original, correct composition, while $c_{i_2}$ corresponds to the erroneous composition. Since $\mathcal{T}_i$ contains some composition $c_{i_3}$ of a substring of length $j$ that is not present in $\tilde{C}_j$, it must be that $c_{i_3} = c_{i_1}.$ Thus, we have $\text{wt}(c_{i_2}) = \tilde{w}_j + \text{wt}(c_{i_1}) - \tilde{w}_{n+1-j}$ and the erroneous composition can be identified and corrected.} Next, suppose \textcolor{black}{on the contrary} that $\mathcal{T}_i \subset \tilde{C}$. In this case, consider the two largest compositions in $\tilde{C} \setminus \mathcal{T}_i$. The two largest compositions in $\tilde{C} \setminus \mathcal{T}_i$ are the compositions of a prefix-suffix pair of length $j$.

Since we have reconstructed the prefix and suffix of length $i,$ and we know that $\sigma_{i+1} = 1$, \textcolor{black}{the prefix-suffix pair is either $(\textbf{s}_1^i 0, 1\textbf{s}_{n+1-i}^{n})$ or $(\textbf{s}_1^i 1, 0\textbf{s}_{n+1-i}^{n})$.}
 To show that only one of the constructed prefix-suffix pairs is valid/compatible, it suffices to show the following: For any two distinct strings $\textbf{s}, \textbf{v} \in \mathcal{S}^{(1)}_{C}(n)$ that have the same $\Sigma^{\lceil \frac{n}{2} \rceil}$, \textcolor{black}{and are such that the longest prefix-suffix pairs shared by them is of length $i,$ one has $|C(\textbf{s}) \setminus C(\textbf{v})| \geq 3$. \textcolor{black}{Note that it now follows from Lemma 2 that for all strings $\textbf{s} \in \mathcal{S}_C^{(1)}(n)$, $\Sigma^{\lceil \frac{n}{2} \rceil}$ can be determined. 
    Thus, if two strings $\textbf{u}, \textbf{v} \in \mathcal{S}_C^{(1)}(n)$ share the same $\Sigma^{\lceil \frac{n}{2} \rceil}$ sequence, and $\textbf{u}_1^{i+1} = \textbf{s}_1^i 0, \textbf{u}_{n-i}^{n} = 1\textbf{s}_{n+1-i}^{n}$, and $\textbf{v}_1^{i+1} = \textbf{s}_1^i 1, \textbf{v}_{n-i}^{n} = 0\textbf{s}_{n+1-i}^{n}$, $\tilde{C}(\textbf{u}) = \tilde{C}(\textbf{v})$ only if $|C(\textbf{u}) \setminus C(\textbf{v})| \leq 2.$ Thus, $|C(\textbf{u}) \setminus C(\textbf{v})| \geq 3$ implies that $\tilde{C}(\textbf{u}) \not = \tilde{C}(\textbf{v}).$ Observe that since $\textbf{v} \in \mathcal{S}^{(1)}_{C}(n)$, the number of $0$s in $c(\textbf{s}_1^{i})$ is at least by two larger than the number of $0$s in $c(\textbf{s}_{n+1-i}^{n}).$} }

Let us assume that on the contrary, there are two strings $\textbf{s}, \textbf{v}$ such that \textcolor{black}{ $|C(\textbf{s}) \setminus C(\textbf{v})|\leq 2$}, and that they differ only in their respective $C_j$ sets.

Since the prefixes and suffixes of the strings of length $i=n-j-1$ are identical, we let $s_1, \dots, s_i$ and $s_{n+1-i}, \dots, s_n$ denote the first and last $i$ bits of both strings. Let $c(\textbf{s})$ denote the composition of the string $\textbf{s}$. Furthermore, let $c(\textbf{s}_{i}^{j})$ denote the composition of $\textbf{s}_{i}^{j}$, $i\leq j$ and let $c = c(\textbf{s}_{i+3}^{n-i-2})$. 

When $n=2(i+1)+1$, the strings differ in two compositions in $C_{n-1-i}$ due to the assumption that the longest prefix-suffix pair shared by the two strings $\textbf{s}$ and $\textbf{v}$ is of length $i$. \textcolor{black}{Since $\textbf{s}$ and $\textbf{v}$ share the same $\Sigma^{\lceil \frac{n}{2} \rceil}$, let $\textbf{s}_{ \lceil \frac{n}{2} \rceil} = \textbf{v}_{ \lceil \frac{n}{2} \rceil} = b.$ Therefore, $\textbf{s} = \textbf{s}_1^i0b1\textbf{s}_{n+1-i}^n, \textbf{v}= \textbf{s}_1^i1b0\textbf{s}_{n+1-i}^n.$ Observe that for the cases $b=0$ and $b=1,$ $|C_{j-1}(\textbf{s}) \setminus C_{j-1}(\textbf{v})| \geq 1$. Thus, $|C(\textbf{s}) \setminus C(\textbf{v})|\geq 3.$ }

When $n \geq 2(i+1) + 3$ and $\sigma_{i+2}=1$, we let ${s_+}$ stand for the $(i+2)^{\text{th}}$ bit in the string $\textbf{s}$, and $v_+$ stand for the $(i+2)^{\text{th}}$ bit of string $\textbf{v}$. \textcolor{black}{Figure~\ref{fig:confusable_strings} illustrates this setting.} When $\sigma_{i+2} \in  \{0, 2 \}$, we let $b$ denote the $(i+2)^{\text{th}}$ bits of the two strings, which are identical.
Next, we determine conditions under which $C_{j-1}(\textbf{s})= C_{j-1}(\textbf{v})$.
\textcolor{black}{We know that if the compositions of the two strings differ by three or more, the two strings cannot be confused under the single composition error model. Due to the constraints imposed by the very construction of the string, we know that $|C_j (\textbf{s}) \setminus C_j(\textbf{v})| =2$. Thus, for the two strings not to be confusable under the given error model, $|C_\ell (\textbf{s}) \setminus C_\ell(\textbf{v})| > 0$ for some $\ell \in [n] \setminus \{ j \}.$ We show that a specific $\ell$ satisfying the previous inequality equals $j-1$, i.e., $|C_{j-1} (\textbf{s}) \setminus C_{j-1}(\textbf{v})| > 0.$}
Note that the compositions of substrings of length $n-i-2$ that contain the bits $i+1,\dots,n-i$ are identical for the two strings. 
\begin{figure}[h]
\centering
	\includegraphics[scale=0.27]{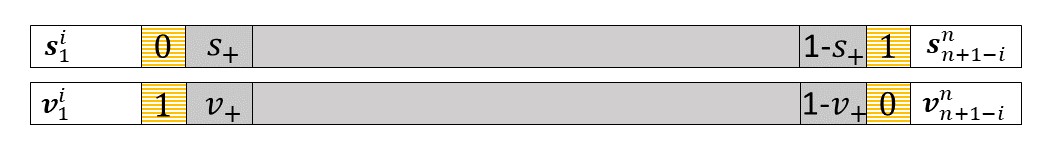} 
\caption{Two strings $\textbf{s}$ and $\textbf{v}$ that satisfy the assumptions used in the proof.}
\label{fig:confusable_strings}
\vspace{-0.08in}
\end{figure}\\
\emph{Case 1}: $\sigma_{i+2} = 1$. With a slight abuse of notation, we choose to write compositions as sets containing both bits and other compositions. On the left-hand-side of the equation below, the compositions correspond to the substrings of $\textbf{s}$ of length $n-i-2$ that \emph{may} differ for the two strings. 
The right-hand-side of the equation corresponds to the same entities in $\textbf{v}$. If the equation holds, then the multisets $C_{j-1}(\textbf{s})$ and $C_{j-1}(\textbf{v})$ are equal.
\begin{align*}
& \begin{rcases}
\begin{dcases}
\{{c(\textbf{s}_1^i), 0, s_+, c\}}, \\
\{{c(\textbf{s}_2^i), 0,  s_+, c ,1- s_+\}}, \\
\{{c(\textbf{s}_{j+2}^n),1, 1- s_+, c\}},  \\
\{{c(\textbf{s}_{j+2}^{n-1}), 1, 1- s_+, c,  s_+\}}
\end{dcases} 
\end{rcases} \\
& \qquad \qquad \qquad \qquad = 
\begin{rcases}
    \begin{dcases}
\{{c(\textbf{s}_1^i),1, v_+, c\}}, \\
\{{c(\textbf{s}_2^i), 1, v_+, c, 1-v_+\}}, \\
\{{c(\textbf{s}_{j+2}^n), 0, 1-v_+, c\}},  \\
\{{c(\textbf{s}_{j+2}^{n-1}), 0, 1-v_+, c, v_+\}}
	\end{dcases} 
\end{rcases} 
\end{align*}
The exhaustive case-by-case arguments that show that the above 
set equality is never true.

\emph{Case 2}: $\sigma_{i+2} \in \{0,2 \}$  Similar reasoning leads to a set equality condition 
in which $s_+$ and $v_+$ are replaced by $b$. Once again, it can be shown by an exhaustive case-by-case analysis that the set equality never holds, independently on the choice of $b$. 
This implies that the composition sets $C_{j-1}(\textbf{s})$ and $C_{j-1}(\textbf{v})$ differ, which in turn implies that the composition multisets of the two strings are at distance $\geq 3$. 
\end{proof}

\textcolor{black}{Recall that when $\lceil \frac{n}{2} \rceil \mod 3 = 0 ,$ the size of the code $\mathcal{S}^{(1)}_{C}(n)$ is $\frac{|\mathcal{S}_R(n-2)|}{2}.$ In addition to the redundancy required to construct the reconstruction code, we require one bit to ensure $n-3$ is even, three bits to fix $s_2, s_{n-1}$ and $s_{\lceil \frac{n}{2} \rceil }$, and four bits to ensure that $\lceil \frac{n}{2} \rceil $ is divisible by three. Thus, $\mathcal{S}^{(1)}_{C}(n)$ requires $\frac{1}{2} \log k + 13$ redundant bits. }

The backtracking string reconstruction process based on an erroneous composition set is straightforward: It takes $\mathcal{O}(n^2)$ time to compute the $\mathcal{T}_k$ multiset, and backtracking performs $\mathcal{O}(n)$ steps. Thus, the decoding algorithm can computes the original string in $\mathcal{O}(n^3)$ time. 

\subsection{Multiple Error-Correcting Reconstruction Codes: The Asymmetric Case} \label{sec:mulerr}

We consider an error model in which each of the multisets $C_i \cup C_{n+1-i},\, i \in [n]$ is allowed to contain at most one composition error and the total number of errors is at most $t$. The codes described in what follows add asymptotically negligible redundancy to the information strings to correct a fixed number of $t$ asymmetric errors. To construct the codes, we generalize the approach used in the previous section for correcting a single error. 

We start with the description of a \emph{$t$-shifted} reconstruction code of even length $m$, denoted by $\mathcal{S}^{(t)}_{R}(m)$ and defined below.

\begin{align}
    \mathcal{S}^{(t)}_{R}(m) = &\{ \textbf{s} \in \{ 0,1\}^{m}, \textbf{s}^t_1 = \mathbf{0}, \textbf{s}_{m-t+1}^m = \mathbf{1}, \text{ and }\label{set1} \\
 &\exists\, I \subseteq \{ t+1, \dots, m-t\} \text{ such that} \nonumber \\ 
 & \qquad \qquad \quad \quad \quad \quad \forall \, i \in I, s_i \not = s_{m+1-i}, \nonumber\\
 & \qquad \qquad \quad \quad \quad \quad  \text{           } \text{ and } \forall \, i \not \in I, s_i = s_{m+1-i}, \nonumber \\
                        &\textbf{s}_{[m/2] \cap I}  \text{ is } \text{a Catalan-Bertrand string.}         \nonumber                \} 
\end{align} 

We refer to strings of the form $\textbf{s}^t_1 \textbf{s}_{[m/2] \cap I} $ as \emph{$t$-shifted} Catalan-Bertrand strings. 
For a $\textbf{s} \in \mathcal{S}^{(t)}_{R}(m)$, every prefix of length $i$ where $\frac{m}{2} \geq i \geq t+1$, has at least $t+1$ more $0$s than its corresponding suffix of the same length. 

\begin{lemma} \label{lem:awesome}
Let \emph{$\textbf{s}, \textbf{v} \in \mathcal{S}^{(t)}_{R}(m)$} share the same $\Sigma^{\frac{m}{2}}$ sequence and satisfy \emph{$|C_j(\textbf{s}) \setminus C_j(\textbf{v}) | \leq 2$} for all $j \in [m]$. If the longest prefix-suffix pair shared by \emph{$\textbf{s}$} and \emph{$\textbf{v}$} is of length $i$, then their corresponding composition multisets $C_{m-i-1}, C_{m-i-2}, \dots, C_{m-i-t}, C_{m-i-t-1}$ each differ in at least $2$ compositions.  
\end{lemma}
We defer the proof to Appendix~\ref{app:lem4}. 
\begin{corollary} \label{cor:4}
Let \emph{$\textbf{s} \in \mathcal{S}^{(t)}_{R}(m)$}, and let \emph{$\tilde{C}(\textbf{s})$} be the composition multiset \emph{$C(\textbf{s})$} corrupted by at most $t$ asymmetric errors. Then, given the correct $\Sigma^{\frac{m}{2}}$ sequence, the string \emph{$\textbf{s}$} can be uniquely reconstructed from \emph{$\tilde{C}(\textbf{s})$}.
\end{corollary}
\begin{proof}
The result immediately follows from Lemma~\ref{lem:awesome}.
\end{proof}

Henceforth, we use $\mathcal{S}^{(t)}_{CA}(n)$ to denote an asymmetric $t$-error-correcting reconstruction code. Strings $\textbf{s} \in \mathcal{S}^{(t)}_{CA}(n)$ are constructed by adding $n-m$ redundancy bits to a string $\textbf{s}' \in \mathcal{S}^{(t)}_{R}(m)$ of even length in such a way that the $\Sigma^{\lceil \frac{n}{2} \rceil}$ sequence can be recovered even in the presence of $t$ asymmetric errors.

\begin{claim}
Let \emph{$\textbf{s}$} be an arbitrary string of even length $n$ and let \emph{$\tilde{C}(\textbf{s})$} denote the composition multiset \emph{$C(\textbf{s})$} corrupted by $t$ asymmetric errors. Then, at least $\frac{n}{2} - 3t$ elements in $(\sigma_1, \sigma_2, \dots, \sigma_{\frac{n}{2}})$ can be determined based on \emph{$\tilde{C}(\textbf{s})$}.
\end{claim}
\begin{proof}
The claim is a consequence of a simple analysis of the set of linear equations in~\eqref{eq:sigmas}. 
Clearly, $w_i$ is unknown whenever $C_i \cup C_{n+1-i}$ contains an error. Therefore, if we have $t$ errors we only have 
$\frac{n}{2}-t$ linear equations that involve $\frac{n}{2}$ variables. From this system of $\frac{n}{2}-t$ linear equations we form a new system of linear equations by subtracting equation~\eqref{eq:sigmas} with index $i$ from the equation~\eqref{eq:sigmas} with index $i+1$. Note that for all values of $i$ such that $w_{i-1}, w_i$ and $w_{i+1}$ are known, the value of $\sigma_{i}$ can be found from the new system of equations. Thus, the derived system of equations allows one to infer at least $\frac{n}{2}-3t$ elements of the $\Sigma^{\frac{n}{2}}$ sequence. Note that all the expressions above assume that $n$ is even. For odd $n$, $\lceil \frac{n}{2} \rceil$ should be used instead.
\end{proof}
We illustrate the above claim with an example. If $w_3, w_4$ and $w_5$ are known then using the linear equations corresponding to $i=3$ and $i=4$, one can infer $\sum_{k=4}^{\frac{n}{2}} \sigma_k$ and using the linear equations corresponding to $i=4$ and $i=5$, one can infer $\sum_{k=5}^{\frac{n}{2}} \sigma_k$. Thus, one can determine $\sigma_4 = \sum_{k=4}^{\frac{n}{2}} \sigma_k - \sum_{k=5}^{\frac{n}{2}} \sigma_k$.

Thus, to recover the entire $\Sigma^{ \frac{n}{2} }$ sequence, it suffices to take the $\Sigma^{\frac{n}{2} }$ string from a systematic Reed-Solomon code over the alphabet $\{ 0,1,2\}$ that can correct up to $3t$ erasures. 

Thus, the codestrings $\textbf{s} \in \mathcal{S}^{(t)}_{CA}(n)$ are constructed via the following procedure:
\begin{itemize}
\item Pick a string $\textbf{s}^{'} = \textbf{s}_1^{' \frac{m}{2} } \textbf{s}_{ \frac{m}{2} +1}^{'m} \in \mathcal{S}^{(t)}_{R}(m)$.
\item Using a systematic Reed-Solomon code over the alphabet $\{ 0,1,2\}$ that can correct up to $3t$ erasures, the $\Sigma^{ \frac{m}{2} }$ sequence is mapped to $\Sigma^{ \frac{n}{2} }$. Note that the sequence $(\sigma_{ \frac{m}{2}  +1}, \dots, \sigma_{\frac{n}{2}})$ is appended to $\Sigma^{\frac{m}{2}}$.
\item A string $\textbf{b}$ of length $n-m$ is created using the sequence $(\sigma_{\frac{m}{2} +1}, \dots, \sigma_{\frac{n}{2}})$ as follows. For all $k \in \left[ \frac{n-m}{2}\right]$:
\begin{align*}
b_k b_{n-m+1 -k} = 
\begin{cases}
00, \text{ if } \sigma_{\frac{m}{2} + k} = 0; \\
01, \text{ if } \sigma_{\frac{m}{2} + k} = 1; \\
11, \text{ if } \sigma_{\frac{m}{2} + k} = 2.  
\end{cases}
\end{align*}
\item A codestring $\textbf{s} \in \mathcal{S}^{(t)}_{CA}(n)$ is obtained by concatenating the strings $\textbf{s}^{'}$ and $\textbf{b}$, namely $\textbf{s} = \textbf{s}_1^{'\frac{m}{2}} \, \textbf{b}_1^{n-m} \, \textbf{s}_{\frac{m}{2}+1}^{'m}$. 
\end{itemize}

Given \emph{$\tilde{C}(\textbf{s})$}, the composition multiset \emph{$C(\textbf{s})$} corrupted by $t$ asymmetric errors, the string $\textbf{s}$ can be uniquely reconstructed via the the following four-step procedure:
\begin{itemize}
\item Construct the linear system of equations governed by~\eqref{eq:sigmas} using the erroneous composition multiset. 
\item Solve for the $\sigma_i$ values that can be inferred from the linear system. 
\item Infer the correct $\Sigma^{\frac{n}{2}}$ sequence using an efficient polynomial evaluation decoder. 
\item Reconstruct the string $\textbf{s}$ using the Backtracking algorithm. 
\end{itemize}
The procedure described above requires $\frac{1}{2} \log n + 6$ redundant bits to ensure the Catalan-Bertrand \textcolor{black}{ string structure of even length}, $2t$ redundant bits for the t-shifted structure and $3t \log n$ redundant bits to correct erasures in the $\Sigma^{\frac{n}{2}}$ sequence. Thus, the number of redundant bits $r$ required is $\left( \frac{1}{2} + 3t \right) \log n + 2t + 6 $. Furthermore, $r$ does not exceed $\left( \frac{1}{2} + 3t \right) \log k + 2t + \textcolor{black}{7} + \left( \frac{1}{2} + 3t \right) \frac{1}{\kappa} $, where $\kappa$ is supremum over all $\kappa > 0$ such that $n \geq (1+\kappa) \left(\left( \frac{1}{2} + 3t \right) \log n + 2t + 7 \right)$. 

\textcolor{black}{Recall that the Backtracking algorithm takes $\mathcal{O}(n^3)$ time to reconstruct the string (it takes $\mathcal{O}(n^2)$ time to find the longest compositions in the set $\mathcal{T}_\ell \setminus C,$ and reconstruct the bits $s_{\ell+1} s_{n-\ell};$ and, there are $\frac{n}{2}$ such pairs to be reconstructed). With a slight abuse of notation, we say that index $i$ corresponds to an asymmetric error if a single composition error occurred in $C_i \cup C_{n+1-i}.$}
\textcolor{black}{Now assume that the indices $i, i+1, \dots , j$ ($j \geq i$) correspond to composition lengths that contain asymmetric errors, and that, $C_{i-1}, C_{n+2-i}, C_{j+1}, C_{n-j}$ are error-free. Note that the proof of Lemma~\ref{lem:awesome} established that the Backtracking algorithm can reconstruct the correct substrings $\textbf{s}_i^j \textbf{s}_{n+1-j}^{n+1-i}$ before proceeding to reconstruct the bits $s_{j+1}, s_{n-j}.$ Thus, every contiguous burst of errors of length $t'$ causes an additional $\mathcal{O}(n^2 2^{t'})$ reconstruction time delay. Thus, the worst case reconstruction time is $\mathcal{O}(n^2 2^{t}).$ } 

Combining this result with that of Corollary~\ref{cor:4} establishes Theorem~\ref{thm:asym}. 

\section{Multiple Error-Correcting Reconstruction Codes: The Symmetric Case}\label{sec:symmetric}

We now turn our attention to designing reconstruction codes capable of correcting symmetric composition errors. The proposed method leverages a polynomial formulation of the composition reconstruction problem first described in~\cite{acharya2014string}. The main result is a constructive proof for the existence of codes with $\O(t^2 \log k)$ bits of redundancy capable of correcting $t$ symmetric composition errors. 

To this end, we first review the results of~\cite{acharya2014string} describing how to formulate the string reconstruction problem in terms of bivariate polynomial factorization. 

For a string $\textbf{s} \in \{0,1\}^n$, let $P_{\textbf{s}}(x,y)$ be a bivariate polynomial of degree $n$ with coefficients in $\{0,1\}$ such that $P_{\textbf{s}}(x,y)$ contains exactly one term with total degree $i \in \{0,1,\ldots, n\}$. If $\textbf{s}= s_1 \ldots s_n$ and if $\Big (P_{\textbf{s}}(x,y) \Big)_i$ denotes the unique term of total degree $i$, then $\Big (P_{\textbf{s}}(x,y) \Big)_0 = 1$, and 
\begin{align*}
\Big (P_{\textbf{s}}(x,y) \Big)_i = \begin{cases}
y \, \Big(P_{\textbf{s}}(x,y) \Big)_{i-1}, \text{ if $s_i = 0$}, \\
x \, \Big(P_{\textbf{s}}(x,y) \Big)_{i-1}, \text{ if $s_i = 1$.}
\end{cases}
\end{align*}
In words, we use $y$ to denote the bit $0$ and $x$ to denote the bit $1$ and then summarize the composition of all prefixes of the string $\textbf{s}$ in polynomial form. As a simple example, for $\textbf{s} = 0100$ we have $P_{\textbf{s}}(x,y) = 1 + y + xy + xy^2 + xy^3$. To see why this is true, we start with the free coefficient $1$, then add $y$ to indicate that the prefix of length one of the string equals $0$, add $xy$ to indicate that the prefix of length two contains one $0$ and one $1$, add $xy^2$ to indicate that the prefix of length three contains two $0$s and one $1$ and so on. 

We also introduce another bivariate polynomial $S_{\textbf{s}}(x,y)$ to describe the composition multiset $C(\textbf{s})$ in a manner similar to $P_{\textbf{s}}(x,y)$. In particular, we now associate each composition with a monomial in which the symbol $y$ represents the bit $0$ and the symbol $x$ with the bit $1$. As an example, for $\textbf{s} = 0100$ we have 
\begin{align*}
C(\textbf{s})=\Big \{ 0 , 1, 0, 0, 01, 01, 0^2, 0^21, 0^21, 0^31 \Big \},
\end{align*}
and 
$$S_{\textbf{s}}(x,y) = x + 3y + 2xy + y^2 + 2xy^2 + xy^3,$$
where the first two terms in $S_{\textbf{s}}(x,y)$ indicate that the composition multiset contains one substring $1$ and three substrings $0$; the next three terms indicate that the string contains two substrings with one $1$ and one $0$ and one substring with two $0$s. The remaining terms are interpreted similarly. 

The key identity from~\cite{acharya2014string} is of the form
\begin{align}\label{eq:SP}
P_{\textbf{s}} (x,y) \, P_{\textbf{s}}\left(\frac{1}{x}, \frac{1}{y}\right) =  (n+1) + S_{\textbf{s}}(x,y) + S_{\textbf{s}}\left(\frac{1}{x}, \frac{1}{y}\right).
\end{align}
Given a bivariate polynomial $f(x,y)$, we use $f^{*}(x,y)$ to denote its reciprocal polynomial, defined as
$$ f^{*}(x,y) = x^{\text{deg}_x(f)} y^{\text{deg}_y(f)} f\left(\frac{1}{x}, \frac{1}{y}\right),$$
where $\text{deg}_x(f)$ denotes the $x$-degree of $f(x,y)$ and $\text{deg}_y(f)$ denotes its $y$-degree. For simplicity, we hence write $d_x = \text{deg}_x(P_{\textbf{s}})$ and $d_y = \text{deg}_y(P_{\textbf{s}})$. Using the notion of the reciprocal polynomial we can rewrite the expression in (\ref{eq:SP}) as:
\begin{align}
P_{\textbf{s}} (x,y) \, P^{*}_{\textbf{s}}(x, y) = x^{d_x} y^{d_y} \, \left( n + 1 + S_{\textbf{s}}(x,y) \right) + S^{*}_{\textbf{s}}(x,y).
\end{align}

Note that if $\tilde{C}(\textbf{s})$ is the composition multiset resulting from $t$ symmetric composition errors in $C(\textbf{s})$ and $\tilde{S}_{\textbf{s}}(x,y)$ is the polynomial representation of $\tilde{C}(\textbf{s})$ while $S_{\textbf{s}}(x,y)$ is the polynomial representation of $C(\textbf{s})$, then
\begin{align*}
\tilde{S}_{\textbf{s}}(x,y) = S_{\textbf{s}}(x,y) + E(x,y),
\end{align*}
where $E(x,y)$ has at most $2t$ nonzero coefficients. \textcolor{black}{Note that the coefficients of $E(x,y)$ lie in $\{-t, -t+1, \dots, -1, 0, 1, \dots , t-1, t \}$. A composition error corresponds to removing a multinomial $e_t$ from $S_{\textbf{s}}(x,y)$ and adding a different multinomial $e_f$. Thus, $-e_t$, and $+e_f$ are addends in $E(x,y)$. Since up to $t$ errors are possible, the coefficients of every multinomial in $E(x,y)$ are integers in $\{-t, -t+1, \dots, -1, 0, 1, \dots , t-1, t \}$. If every multinomial removed from or added to $S_{\textbf{s}}(x,y)$ is unique, then there are $2t$ terms in $E(x,y)$. Otherwise, the number of multinomials is less than $2t$.} Our first result relates $\tilde{S}_{\textbf{s}}(x,y)$ and $P_{\textbf{s}}(x,y)$.

\begin{claim}\label{cl:equiv} Suppose that $\emph{wt}(\textbf{s}) \textcolor{black}{\bmod (2t+1)} = c_w$ for some $c_w \in \{0,1,\ldots, 2t\}$. Then, given $\tilde{S}_{\textbf{s}}(x,y)$ and $c_w$ one can generate
\begin{align*}
P_{\textbf{s}}(x,y) \, P^{*}_{\textbf{s}}(x,y)  + \tilde{E}(x,y),
\end{align*}
where the polynomial $\tilde{E}(x,y)$ has at most $4t$ terms.
\end{claim}
\begin{proof} First, recall that $\tilde{S}_{\textbf{s}}(x,y) = S_{\textbf{s}}(x,y) + E(x,y)$ where $E(x,y)$ has at most $2t$ nonzero coefficients. Given $c_w$, we can easily determine the exact degrees $d_x$ and $d_y$ of the polynomial encoding of $\textbf{s}$: \textcolor{black}{In the error-free case, the sum of all compositions of length $1$ (i.e., the sum of the bits of the string) equals $\text{wt}(\textbf{s}) = d_x$.
When the composition multiset is erroneous, we can only observe $\tilde{d}_x, $ which takes a value in the set $\{d_x-t, d_x-t+1, \dots, d_x, d_x+1, d_x+2, \dots, d_x+t-1, d_x+t \}.$ Equivalently, we know that 
$$d_x \in \{ \tilde{d}_x-t, \tilde{d}_x-t+1, \dots, \tilde{d}_x, \tilde{d}_x+1, \tilde{d}_x+2, \dots, \tilde{d}_x+t-1, \tilde{d}_x+t \}. $$
Since $d_x \equiv c_w \mod (2t +1),$ exactly one value in the set $\{ \tilde{d}_x-t, \tilde{d}_x-t+1, \dots, \tilde{d}_x, \tilde{d}_x+1, \tilde{d}_x+2, \dots, \tilde{d}_x+t-1, \tilde{d}_x+t \}$ will satisfy this condition. Hence, $d_w$ can be inferred exactly, and since $d_y = n - d_x$, the same conclusion holds for $d_y$.}

Next, we form $P_{\textbf{s}}(x,y) \, P^{*}_{\textbf{s}}(x,y)$ as follows:
\begin{align*}
&x^{d_x} y^{d_y} \, \left( n + 1 + \tilde{S}_{\textbf{s}}(x,y) + \tilde{S}_{\textbf{s}} \left(\frac{1}{x}, \frac{1}{y}\right) \right) \\
&= x^{d_x} y^{d_y} (n+1) + x^{d_x} y^{d_y} \, \times \\
   &\left( S_{\textbf{s}}(x,y) + E(x,y) 
  +  S_{\textbf{s}}\left(\frac{1}{x}, \frac{1}{y}\right) + E\left(\frac{1}{x},\frac{1}{y}\right) \right)   \\
&= P_{\textbf{s}}(x,y) \, P^{*}_{\textbf{s}}(x,y) + x^{d_x} y^{d_y} \left( E(x,y) + E \left( \frac{1}{x}, \frac{1}{y} \right) \right) \\
&= P_{\textbf{s}}(x,y) \, P^{*}_{\textbf{s}}(x,y) + \tilde{E}(x,y),
\end{align*}
where $\tilde{E}(x,y) =  x^{d_x} y^{d_y}\left(E(x,y) + E\left( \frac{1}{x}, \frac{1}{y} \right) \right)$ has at most $4t$ nonzero coefficients, which proves the desired result.
\end{proof}

Let $\mathbb{F}_q$ be a finite field of order $q$, where $q$ is an odd prime. Let $\alpha \in \mathbb{F}_q$ be a primitive element of the field. For a polynomial $f(x) \in \mathbb{F}_q[x]$, let $\R(f)$ denote the set of its roots. We find the following result useful for our subsequent derivations.

\begin{theorem}\label{th:RS} (\cite[Ch.~5]{roth2006introduction}) Assume that $E(x) \in \mathbb{F}_q[x]$ has \textcolor{black}{at most $t$} nonzero coefficients. Then, $E(x)$ can be uniquely determined in $\cO(n^2)$ time given $E(\alpha^t), E(\alpha^{t-1}), \ldots, E(\alpha^0), E(\alpha^{-1}), \ldots, E(\alpha^{-t})$.
\end{theorem}

\subsection{The Code Construction}
Our approach to constructing a symmetric $t$-error-correcting code of length $n$, denoted by $\mathcal{S}^{(t)}_{CS}(n)$, relies on the fact that $\tilde{E}(x,y)$ may be written as:
\begin{align}\label{eq:exy}
\tilde{E}(x,y) =& ( a_{i_{1},1} y^{j_{i_1,1}} + \cdots + a_{i_1,m_{i_1}} y^{j_{i_1,m_{i_1}}}) x^{i_1} + \nonumber \\
&(a_{i_2,1} y^{j_{i_2,1}} + \cdots + a_{i_2,m_{i_2}} y^{j_{i_2,m_{i_2}}}) x^{i_2} +  \nonumber \\
& \qquad \qquad \quad \quad \, \vdots \\
&(a_{i_h,1} y^{j_{i_h,1}} + \cdots + a_{i_h,m_{i_h}} y^{j_{i_h,m_{i_h}}}) x^{i_h}, \nonumber
\end{align}
where each $a_{i,j} \in \{-1,1\}$, $h \leq 4t$ and the total number of nonzero terms is $\leq 4t$. Since $\tilde{E}(x,y)$ is restricted to have at most $4t$ nonzero terms, each of the polynomials $(a_{i_\ell,1} y^{j_{i_\ell,1}} + \cdots + a_{i_\ell,m_{i_\ell}} y^{j_{i_\ell,m_{i_\ell}}})$ can contain at most $4t$ nonzero terms. Consequently, one has $m_{i_\ell} \leq 4t$ for all $\ell \in \{1,2,\ldots,h\}$. 

Based on the previous observations we are ready to introduce our first code construction. We assume that $P_{\textbf{s}}(x,y)$ is a \textcolor{black}{ bivariate polynomial over the field $\mathbb{F}_q,$ where $q$ is the smallest prime $\geq 2n+1.$} Clearly, for a $P_{\textbf{s}}(x,y) \in \mathbb{I}[x,y]$ over the set integers $\mathbb{I}$, one can obtain $P_{\textbf{s}}(x,y) \in \mathbb{F}_q[x,y]$ by simply reducing $P_{\textbf{s}}(x,y)$ modulo $q$.

\begin{lemma}\label{lem:deca1} Let 
\begin{align*}
\C =\{\textbf{s} \in \{0,1\}^n \text{ s.t. } \emph{wt}(\textbf{s}) \bmod 2t+1 &= 0, \\
\{1,\alpha,\alpha^2, \ldots, \alpha^{4t} \} &\subseteq \R(P_{\textbf{s}}(x,1)),\\
\{1,\alpha,\alpha^2, \ldots, \alpha^{4t} \} &\subseteq \R(P_{\textbf{s}}(x,\alpha)),\\
                                            & \, \vdots \\
\{1,\alpha,\alpha^2, \ldots, \alpha^{4t} \} &\subseteq \R(P_{\textbf{s}}(x,\alpha^{4t}))\}.
\end{align*}
Then, $\C$ is a symmetric $t$-error-correcting code.
\end{lemma}
\begin{proof} We prove the claim by describing a decoding algorithm that for any given $\tilde{S}_{\textbf{s}}(x,y)$, which is the result of at most $t$ composition errors occurring in $S_{\textbf{s}}(x,y)$, uniquely recovers $S_{\textbf{s}}(x,y)$.

Since there are at most $t$ erroneous compositions in $\tilde{S}_{\textbf{s}}(x,y)$, one can determine $\text{wt}(\textbf{s})$ by summing up the length-one compositions (i.e., the bits) in $\tilde{S}_{\textbf{s}}(x,y)$ along with the fact that $\text{wt}(\textbf{s}) \bmod 2t+1 = 0$. Therefore, from Claim~\ref{cl:equiv}, we can construct the polynomial
\begin{align}\label{eq:l1deceq}
F(x,y) = P_{\textbf{s}}(x,y) \, P^{*}_{\textbf{s}}(x,y)  + \tilde{E}(x,y),
\end{align} 
where $\tilde{E}(x,y)$ has at most $4t$ nonzero coefficients. 

Suppose that $\beta, \beta' \in \mathbb{F}_q$. First, observe that if $P_{\textbf{s}}(\beta,\beta') \, P^{*}_{\textbf{s}}(\beta, \beta') = 0$, then $P_{\textbf{s}}(\frac{1}{\beta},\frac{1}{\beta'}) \, P^{*}_{\textbf{s}}(\frac{1}{\beta},\frac{1}{\beta'}) = 0$ which immediately follows from the definition of $P^{*}_{\textbf{s}}(x,y)$. \textcolor{black}{Thus, if $(\beta, \beta')$ is a root of $P_{\textbf{s}}(\cdot , \cdot) $, then so is $(\beta^{-1}, \beta'^{-1})$.} Since $\{1,\alpha,\alpha^2, \ldots, \alpha^{4t} \} \subseteq \R(P_{\textbf{s}}(\alpha^{\ell_1},y))$ for all $\ell_1 \in \{ 0,1,\ldots,4t \}$, and similarly $\{1,\alpha,\alpha^2, \ldots, \alpha^{4t} \} \subseteq \R(P_{\textbf{s}}(x,\alpha^{\ell_2}))$ for all $\ell_2 \in \{ 0,1,\ldots,4t \},$ it follows that $F(\alpha^{\ell_1},\alpha^{\ell_2})=\tilde{E}(\alpha^{\ell_1},\alpha^{\ell_2})$. Hence, we have:
\begin{align*}
\tilde{E}(\alpha^{\ell_1},&\alpha^{\ell_2}) = \\
& \quad \, \, \big( a_{{i_1,1}} \alpha^{{\ell_2} \times j_{i_1,1}} + \cdots + a_{i_1,m_{i_1}} \alpha^{{\ell_2} \times j_{i_1,m_{i_1}}} \big ) \alpha^{\ell_1 \times i_1}\\
&+ \big ( a_{{i_2,1}} \alpha^{{\ell_2} \times j_{i_2,1}} + \cdots + a_{{i_2,m_{i_2}}} \alpha^{{\ell_2} \times j_{i_2,m_{i_2}}} \big) \alpha^{\ell_1 \times i_2}  \\
&  \qquad \qquad \qquad \qquad \, \quad \vdots  \\ 
&+ \big( a_{{i_h,1}} \alpha^{{\ell_2} \times j_{i_h,1}} + \cdots + a_{{i_h,m_{i_h}}} \alpha^{{\ell_2} \times j_{i_h,m_{i_h}}} \big) \alpha^{\ell_1 \times i_h},
\end{align*}
for ${\ell_1},{\ell_2} \in  \{ 0,1,\ldots,4t,-1,-2,\ldots,-4t \}$. From Theorem~\ref{th:RS}, for any fixed $\ell_2$ we know the evaluations $\tilde{E}(\alpha^{\ell_1}, \alpha^{\ell_2})$ for $\ell_1 \in \{0,1,\ldots,4t,-1,-2, \ldots,-4t\}$, so that we can recover the polynomials
\begin{align}\label{eq:l1efy}
\tilde{E}(x,\alpha^{\ell_2}) &= \big( a_{{i_1,1}} \alpha^{{\ell_2} \times j_{i_1,1}} + \cdots + a_{i_1,m_{i_1}} \alpha^{{\ell_2} \times j_{i_1,m_{i_1}}} \big ) x^{i_1}  \nonumber \\
&+ \big ( a_{{i_2,1}} \alpha^{{\ell_2} \times j_{i_2,1}} + \cdots + a_{{i_2,m_{i_2}}} \alpha^{{\ell_2} \times j_{i_2,m_{i_2}}} \big) x^{i_2}  \nonumber  \\
&\qquad \qquad \qquad \qquad \, \quad \vdots \nonumber \\
&+ \big( a_{{i_h,1}} \alpha^{{\ell_2} \times j_{i_h,1}} + \cdots + a_{j_{i_h,m_{i_h}}} \alpha^{{\ell_2} \times j_{i_h,m_{i_h}}} \big) x^{i_h},
\end{align} 
using a decoder for a cyclic Reed-Solomon code of complexity $\cO(n^2)$. 

Let 
$$M_{i_\ell}(y) =  a_{{i_\ell,1}} y^{j_{i_\ell,1}} + \cdots + a_{i_\ell,m_{i_{\ell}}} y^{j_{i_\ell,m_{i_\ell}}}$$
be the polynomial multiplier of $x^{i_\ell}$ in $\tilde{E}(x,y)$. From the previous discussion, we know that the maximum number of nonzero terms in $M_{i_\ell}(y)$ is $4t$. Using (\ref{eq:l1efy}), we can determine $M_{i_\ell}(\alpha^{\ell_2})$ for $\ell_2 \in \{0,1,2,\ldots,4t, -1, -2, \ldots, -4t\}$. Due to Theorem~\ref{th:RS}, this implies that we can recover $M_{i_\ell}(y)$ for $\ell\in \{1,2,\ldots, h \}$ once again using a decoder for a Reed-Solomon code. Since $$\tilde{E}(x,y) = M_{i_1}(y) x^{i_1} + M_{i_2}(y) x^{i_2} + \cdots + M_{i_h}(y) x^{i_h},$$ \textcolor{black}{ we can determine $E(x,y)$ by noting the following: 1) Given $\text{wt} (\textbf{s}) \mod (2t +1), $ we can recover $\text{wt} (\textbf{s})$ from the erroneous composition multiset, from which $d_x$ and $d_y = n- d_x$ can be determined. 2) Since $d_x , d_y$ are known, and $\tilde{E}(x,y) =  x^{d_x} y^{d_y}\left(E(x,y) + E\left( \frac{1}{x}, \frac{1}{y} \right) \right),$  $E(x,y)$ can be determined. } Subsequently we can reconstruct $S_{\textbf{s}}(x,y)$ given $\tilde{S}_{\textbf{s}}(x,y)$.
\end{proof}

The following corollary is an immediate consequence of Lemma~\ref{lem:deca1}.

\begin{corollary}\label{cor:sideinfo} Let 
\begin{align*}
\C=\{\textbf{s} \in \{0,1\}^n \text{ s.t. } P_{\textbf{s}}(\alpha^{\ell_1}, \alpha^{\ell_2}) &= a_{\ell_1,\ell_2},\\
\emph{wt}(\textbf{s}) &\equiv a \bmod 2t+1\},
\end{align*}
for all $\ell_1,\ell_2 \in \{0,1,\ldots, 4t\}$, $a \in \{0,1,\ldots, 2t\}$, and where $(a_{\ell_1,\ell_2})_{\ell_1=0, \ell_2=0}^{4t}$ is an arbitrary vector from $\mathbb{F}_q^{ (4t+1)^2}$. Then, $\C$ can correct $t$ symmetric composition errors. 
\end{corollary}

\subsection{A Systematic Encoder $\E_{t,n}$}

We construct next a systematic encoder $\E_{t,n}$ for the previously proposed codes. 

Let $r$ be the number of redundant bits in the proposed code construction. We will show in Theorem~\ref{th:mainth} that for all $n$, one requires a redundancy that does not exceed
\begin{align*}
 4 \Big[ &(4t+1)^2 (\log (2n+1)+1) + \log (2t+1) \\
&+ t \log \left(  (4t+1)^2 (\log (2n+1)+1) + \log (2t+1) \right) \Big] \\
&+ \frac{1}{2} \log(n)+ \textcolor{black}{5}.
\end{align*}
One can show that $r$ does not exceed $156 t^2 \log 8n$.
Thus, $r= \cO(t^2 \log n)$. Furthermore, $r$ does not exceed $156 t^2 \log 8k + 156 t^2 \left( \frac{1}{\kappa} \right)$, where $\kappa$ is supremum over all $\kappa > 0$ such that $n \geq 156 (1+\kappa) (t^2 \log 8n +1)$. \textcolor{black}{To express the redundancy in terms of the information length $k$, we upper bound $c t^2 \log n,$ where \textcolor{black}{$c$ is a constant,} as follows. First, we write 
$$c t^2 \log n = c t^2 \left( \log k + \log \left(\frac{n-k+k}{k}\right) \right).$$ 
Then, we upper bound the term $\log \left(\frac{n-k+k}{k}\right)$ using the Taylor series for $\log(1+x)$ and the linear term involved to arrive at $\log \left(\frac{n-k+k}{k}\right) < \frac{n-k}{k}$. For $k_0$ large enough and for all $k \geq k_0$, $\frac{n-k}{k} ct^2$ can be upper bound by a constant independent of $n$ and $k$ under the given parameter assumptions.}

The encoder $\E_{t,n}$ takes as input the string $\textbf{u} \in \{0,1\}^{n-\hat{r}}$, where $\hat{r}>0$ is a redundancy to be precisely specified later, and it produces a string $\textbf{s}$. \textcolor{black}{The evaluations of the polynomial $P_{\textbf{s}}(x,y)$ is stored in 
$$ \left( w_1, w_2, \ldots, w_{\frac{\hat{r}}{2}} \right) \bmod 2,$$
where we recall that $w_i$ stands for the cumulative weight of compositions of length $i$ in $C(\textbf{s})$.}
 
Let $\E_t : \{0,1\}^{m} \to \{0,1\}^{m + t \log m}$ be a systematic encoder for a code with minimum Hamming distance $2t+1$ that inputs a string of length $m$ and outputs a string of length $m+t\log m$. We will use this encoder with \textcolor{black}{$m= (4t+1)^2 (1+\log (2n+1)) + \log (2t+1)$}. Clearly, such a code exists since binary BCH codes of odd minimum distance have the desired set of parameters. 
\vspace{3pt}
\hrule \vspace{0.5pt} \hrule
\vspace{3pt}
\textbf{Encoder} $\E_{t,n} : \{0,1\}^{n- \hat{r}} \to \{0,1\}^n$.
\vspace{3pt}
\hrule \vspace{0.5pt} \hrule
\vspace{3pt}
\textbf{Input} String $\textbf{u} \in \{0,1\}^{n- \hat{r}}$. 

\textbf{Output} Symmetric $t$-error-correcting codestring $\textbf{s} \in \{0,1\}^{n}$.
\vspace{3pt}
\hrule 
\vspace{3pt}
\begin{enumerate}
\item Let $\alpha \in \mathbb{F}_q$ be a primitive element and let $q$ be an odd prime $\geq 2n+1$. For $\ell_1,\ell_2 \in \{0,1,\ldots, 4t\}$, set $a_{\ell_1,\ell_2} = P_{\textbf{u}}(\alpha^{\ell_1}, \alpha^{\ell_2})$, $\ba = (a_{\ell_1,\ell_2})_{\ell_1=0, \ell_2=0}^{4t}$. \\
Let $a=\text{wt}(\textbf{u}) \bmod 2t+1$. 
\item Let $\bf{\bar{s}} = \E_t($$a$, $\ba) \in \{0,1\}^{\frac{\hat{r}}{4}}$.
\item For $j \in \{1,2, \ldots, \frac{\hat{r}}{2} \}$, define $\textbf{z}=(z_1 \ldots z_{\frac{\hat{r}}{2}})$ as 
\begin{align*}
z_{j} = \begin{cases}
 \sum_{i=1}^{j-1} z_{i} \bmod 2, &\text{ if $j$ is odd and } \bar{s}_{\frac{j+1}{2}}=0,\\
 \sum_{i=1}^{j-1} z_{i} + 1 \bmod 2, &\text{ if $j$ is odd and } \bar{s}_{\frac{j+1}{2}} = 1, \\
 0, &\text{ if $j$ is even.}
\end{cases}
\end{align*}
\item Set $\textbf{s} =  \textbf{0} \, \textbf{u} \, \textbf{z}  \in \{0,1\}^n$, where $\textbf{0}$ is an all-zero string of length $\frac{\hat{r}}{2}$.
\end{enumerate}
\vspace{3pt}
\hrule \vspace{0.5pt} \hrule
\vspace{3pt}
The $t$-error-correcting code $\mathcal{S}^{(t)}_{CS}(n)$ is generated by the following two-step procedure:
\begin{itemize}
\item An information string of length $k$ is first encoded using the reconstruction code $\S_{R}$, resulting in the string $\textbf{u} \in \S_{R}(n - \hat{r})$.
\item The string $\textbf{u}$ is passed through the encoder $\E_{t,n}$, resulting in the codestring $\textbf{s}  = \E_{t,n} (\textbf{u}) \in \mathcal{S}^{(t)}_{CS}(n)$. 
\end{itemize}
Based on the above analysis, we set $\hat{r}$ to be the smallest integer $\geq r- \left( \frac{1}{2} \log(n) + \textcolor{black}{5} \right)$ that is divisible by $4$. 

The redundancy of the code may be calculated as follows: 
\begin{enumerate}
\item Since $q \geq 2n+1$, every $\alpha_{\ell_1,\ell_2}$, $\ell_1, \ell_2 \in \{0,1, \dots 4t \}$ requires at most $1 + \log (2n + 1)$ (due to the fact that given any positive integer $x$, there exits a prime number between $x$ and $2x$).
\item Note that $a$ requires $\log 2t +1$ bits of redundancy. Thus, $\frac{\hat{r}}{4}$ is at most
\begin{align}
 & (4t+1)^2 (1+\log (2n+1)) + \log (2t+1)  \notag \\
&+ t \log ((4t+1)^2 (1+\log (2n+1)) + \log (2t+1)). \notag
\end{align}
\item As already observed, the reconstructable string $\textbf{u}$ requires at most $ \frac{1}{2}\log n + \textcolor{black}{5}$ bits of redundancy.
\end{enumerate}
The redundancy of the encoder $\E_{t,n}$ is $\O (t^2 \log n)$ bits. 

We find the following claims useful in our subsequent derivations.
\begin{claim}\label{cl:zr} At Step 3) of the encoding procedure, for odd $j \in [\frac{\hat{r}}{2}]$, one has
\begin{align}
\bar{s}_{\frac{j+1}{2}} = \sum_{i=1}^j z_i \bmod 2.
\end{align}
\end{claim}
This claim obviously follows from the definition of the string $\textbf{z}$.

Recall next that for a string $\textbf{s} \in \{0,1\}^n$, its $\Sigma^{\lceil \frac{n}{2}\rceil}$ sequence $( \sigma_1, \sigma_2, \ldots, \sigma_{\lceil \frac{n}{2} \rceil}) \in \{0,1,2\}^{\lceil \frac{n}{2} \rceil}$ equals $\sigma_i = s_i + s_{n+1-i}$. As a result of Step 4) of encoding with $\E_{t,n}$, we have the next result.

\begin{claim}\label{cl:rs} For $j \in [\frac{\hat{r}}{2}]$,
\begin{align*}
\textcolor{black}{z_j = \sigma_{\frac{\hat{r}}{2} +1 -j}.}
\end{align*}
\end{claim}

The next claim connects the quantities $w_i$ and $\bar{\textbf{s}}$, defined in Step 2 of the encoding procedure.
\begin{claim}\label{cl:rw} For $j \in \frac{\hat{r}}{4}$, it holds
$$ w_{2j} \equiv \bar{s}_{j} \bmod 2. $$
\end{claim}
\begin{proof} The result is a consequence of the observation that 
\begin{align*}
w_{2j} \equiv  & 2j w_1 - (2j-1) \sigma_1 - (2j-2) \sigma_2- \cdots - \sigma_{2j-1} \bmod 2 \\
\equiv  &\sigma_1 + \sigma_3 + \cdots + \sigma_{2j-1} \bmod 2,
\end{align*}
where the first line follows from Equation (\ref{eq:sigmas}). From Claims~\ref{cl:zr} and \ref{cl:rs}, and the previous observation, and the fact that we set $z_j = 0$ for even values of $j$ in Step 3) of the encoding procedure, we have
\begin{align*}
w_{2j} \equiv \sum_{i=1}^{2j-1} \sigma_j \equiv \sum_{i=1}^{2j-1} z_j \equiv \bar{s}_{j} \bmod 2.
\end{align*}\end{proof}

The next result will be used to prove the main finding regarding symmetric error-correction codes, as stated in Theorem~\ref{th:mainth}.

\begin{lemma}\label{lem:mclemma} The collection of strings  
\begin{align*}
\C = \Big \{ \textbf{s} : \textbf{s} = \E_{t,n}(\textbf{u}), \textbf{u} \in \{0,1\}^{n - \hat{r}} \Big \}
\end{align*}
constitutes a symmetric $t$-error-correcting code.
\end{lemma}
\begin{proof} In order to prove the result, we will describe how to recover $S_{\textbf{s}}(x,y)$ given $\tilde{S}_{\textbf{s}}(x,y),$ where $\tilde{S}_{\textbf{s}}(x,y)$ is the result of at most $t$ composition errors in $S_{\textbf{s}}(x,y)$ for a codestring generated according to $\E_{t,n}(\textbf{u}) = \textbf{s}$. 

We begin by forming the string 
\begin{align*}
\tilde{\bw} = \Big ( \tilde{w}_{2}, \tilde{w}_{4}, \ldots, \tilde{w}_{\frac{\hat{r}}{2}}  \Big).
\end{align*}
One can obtain $\tilde{\bw}$ from $\tilde{S}_{\textbf{s}}(x,y)$ by summing up the 1s in all compositions of length two to get $\tilde{w}_2$, summing up the 1s in all compositions of length four to get $\tilde{w}_4$, and so on. For simplicity, let $\bw = \Big ( w_{2}, w_{4}, \ldots, w_{\frac{\hat{r}}{2}}  \Big)$ for the string $\textbf{s}$. 

Since there are at most $t$ composition errors in $\tilde{S}_{\textbf{s}}(x,y)$, it follows that 
\begin{align*}
d_H \Big( {\bw} \bmod 2, \tilde{{\bw}} \bmod 2 \Big) \leq t.
\end{align*}
From Claim~\ref{cl:rw}, since $\bw \bmod 2$ belongs to a code with minimum Hamming distance $2t+1$, we can recover $\bw \bmod 2$ from $\tilde{\bw} \bmod 2$. Then, given $\bw \bmod 2,$ we can recover $\bar{\textbf{s}}$ from Step 2) of the encoding procedure, and from $\bar{\textbf{s}}$ \textcolor{black}{we can determine $a = \text{wt}(\textbf{u}) \mod (2t +1).$} Using $\bar{\textbf{s}}$, it is also straightforward to determine $\textbf{z}$ from Step 3) of the encoding procedure. \textcolor{black}{Thus, $\text{wt}(\textbf{z})$ is determined accurately as well. One can then easily determine the exact (yet potentially erroneous) weight of $\textbf{u},$ since $\text{wt}(\textbf{u}) = \text{wt}(\textbf{s}) - \text{wt}(\textbf{z}).$ Given $\tilde{\text{wt}} (\textbf{s})$, as determined from the sum of all compositions of substrings of length one, since we know 1) $| \left[ \tilde{\text{wt}} (\textbf{s}) - \text{wt} (\textbf{z}) \right] - \text{wt} (\textbf{u})| \leq t$, and 2) $a = \text{wt}(\textbf{u}) \mod (2t +1)$ we can infer $\text{wt} (\textbf{u})$ exactly.}
Subsequently, we can recover
\begin{align*}
\text{wt}(\textbf{s}) = \textcolor{black}{ \text{wt}(\textbf{u}) + \text{wt}(\textbf{z})},
\end{align*}
and from $\text{wt}(\textbf{s})$, we can determine $d_x$ and $d_y$, the $x$ and $y$ degrees of the polynomial $P_{\textbf{s}}(x,y)$. 

Next, we turn our attention to recovering the evaluations of the polynomial $P_{\textbf{s}}(\alpha^{\ell_1},\alpha^{\ell_2})$ for $\ell_1,\ell_2 \in \{0,1,\ldots, 4t\}$. These, along with $\text{wt}(\textbf{s})$, suffice according to Lemma~\ref{lem:deca1} to recover $\textbf{s}$.
From $
\bar{\textbf{s}} $, we can determine $P_{\textbf{u}}(\alpha^{\ell_1}, \alpha^{\ell_2})$ according to Steps 1) and 2) of the encoding procedure. 

Let $d_{x,\textbf{u}} = \deg_x(P_{\textbf{u}}(x,y))$ and $d_{y,\textbf{u}} = \deg_y (P_{\textbf{u}}(x,y))$. 

First, note that
\begin{align*}
P_{\textbf{s}}(x,y) &=P_{\textbf{0}}(x,y)+ y^{\frac{\hat{r}}{2}} (P_{\textbf{u}}(x,y)-1) \\
&+ x^{d_{x,\textbf{u}}} y^{\frac{\hat{r}}{2}+d_{y,\textbf{u}}} \, (P_{\textbf{z}}(x,y)-1).
\end{align*}
Therefore, since $\textbf{z}$ is already known, we have
\begin{align*}
P_{\textbf{s}}(\alpha^{\ell_1},\alpha^{\ell_2}) &= P_{\textbf{0}}(\alpha^{\ell_1},\alpha^{\ell_2}) + \alpha^{\ell_2 \times \frac{\hat{r}}{2}} (P_{\textbf{u}}(\alpha^{\ell_1},\alpha^{\ell_2})-1)\\
&+ \alpha^{\ell_1 \times d_{x,\textbf{u}}} \alpha^{\ell_2 \times (\frac{\hat{r}}{2}+d_{y,\textbf{u}} )} \, (P_{\textbf{z}}(\alpha^{\ell_1},\alpha^{\ell_2})-1),
\end{align*}
The proof of the claim now follows from Corollary~\ref{cor:sideinfo}.
\end{proof}

We are left with the task of reconstructing the string $\textbf{s}$ from its correct composition multiset $C(\textbf{s})$. Recall that if all pairs of prefixes and suffixes of the same length are such that their weights differ, the string can be reconstructed efficiently by the Backtracking algorithm. Also, recall that the string $\textbf{s}$ is obtained by concatenating three strings, \textit{i.e.}, $\textbf{s} = \textbf{0} \, \textbf{u} \, \textbf{z}$. The prefix of length $\frac{\hat{r}}{2}$ is fixed to be all zeros and can therefore be reconstructed immediately. Lemma~\ref{lem:mclemma} allows one to recover the suffix $\textbf{z}$. Since $\textbf{u} \in \S_{R}(n - \hat{r})$, every prefix of length \textcolor{black}{$\leq \lfloor \frac{n}{2} \rfloor$} has strictly more $0$s than its corresponding suffix of the same length. Thus, the Backtracking algorithm can efficiently reconstruct the correct string $\textbf{s}$. This establishes the result of Theorem~\ref{th:mainth}.

We conclude our exposition by describing another family of uniquely reconstructable codes that can correct up to $t$ composition errors in $C(\textbf{s})$. These codes rely on the use of \emph{Catalan paths}. Recall that Catalan paths of length $2h$ may be represented by binary strings that have the property that every prefix has at least as many $0$s as $1$s and the weight of the strings is $h$. 

Let $\P(2h) \subset \{0,1\}^{2h}$ denote the set of Catalan strings of even length $2h$. It is well-known that the codebook $\P(2h)$ has approximately $\frac{3}{2}\, \log h$ bits of redundancy, which follows directly from the expression for the Catalan number $C_h=\frac{1}{h+1}\binom{2h}{h}$.

The main differences between the polynomial construction and the Catalan-based designs are that the former has a \emph{larger order of redundancy} ($\cO(t^2\,\log\,n)$ compared to $\cO(\log n+t)$) but also has \emph{an efficient decoding algorithm}. At this point, no algorithm scaling efficiently with both $n$ and $t$ is known for the Catalan-based construction.   
 
The basic idea behind the construction is simple and it imposes two constraints on the underlying codestrings:
\begin{enumerate}
	\item \textbf{The Catalan string constraint}: This constraint requires that the codestrings be Catalan.
	\item \textbf{Parity symbols}: The codestrings need to include $4t+1$ 0s in the prefix and $4t+1$ 1s in the suffix.
\end{enumerate}
Intuitively, the fixed prefixes of $0$s and suffixes of $1$s, as well as the balancing property of Catalan strings ensure that for at least $4t+1$ choices of $\ell$, the compositions multisets $C_{\ell}(\textbf{s})$ and $C_{\ell}(\textbf{v})$ of two distinct codestrings $\textbf{s}$ and $\textbf{v}$ differ in at least one composition.

Throughout our subsequent exposition, due to the heavy use of subscripts and superscripts, we write $-i$ instead of $n-i+1$ for all indices used. 

Let
\begin{align}
\C(n,t) = & \Big\{ \textbf{s} \in \{0,1\}^n \, : \,   s_1\, \ldots \, s_{4t+1}= 0 \, 0 \, \ldots 0, \label{constr4} \\
& s_{-4t-1}\, \textcolor{black}{s_{-4t}}\, \ldots \,s_{-1} = 1\,1\ldots \,1, \nonumber \\
& s_{4t+2}\, s_{4t+3}\, \ldots\, s_{-4t-2} \in \P(n-2(4t+1)) \Big \} \nonumber,
\end{align}
\textcolor{black}{where $n$ is even.}

We show next that $\C(n,t)$ is a $t$ symmetric composition error-correcting code with $\cO(\log n + t)$ bits of redundancy. This redundancy is significantly improved compared to that of the previously described polynomial evaluation construction. 

Henceforth, $S_1 \bigtriangleup S_2 = (S_1 \setminus S_2) \cup (S_2 \setminus S_1)$ is used to denote the symmetric difference of two sets $S_1$ and $S_2$.

\begin{theorem}\label{th:main} The code $\C(n,t)$ can correct $t$ composition errors.
\end{theorem} 
\begin{proof} We prove the result by showing that any pair of distinct codestrings $\textbf{s}, \textbf{v} \in \C$ satisfies 
	$$ |C(\textbf{v}) \bigtriangleup C(\textbf{s})| \geq 4t+1,$$
	which implies the desired result. 
	
	Suppose that $i$ is the smallest integer such that either $s_i \neq v_i$ or $s_{-i} \neq v_{-i}$. Since the first and last $4t+1$ bits of each codestring are identical \textcolor{black}{and since every Catalan string begins with a $0$ and ends with a $1$}, we have $i \geq 4t+\textcolor{black}{3}.$
	
	Next, assume that \textcolor{black}{$s_{-i} \neq v_{-i} , s_{i} = v_{i}$. The cases $s_i \neq v_i, s_{-i} = v_{-i}$ and $s_{i} \neq v_{i}, s_{-i} \neq v_{-i}$} can be proven similarly by considering the reversals of the strings $\textbf{s}$ and $\textbf{v}$. 
	
	Consider the compositions of the following two substrings:
	\begin{align*}
	\textbf{s}_1^{-i-1} &= s_1 \, s_2 \,  \ldots s_{-i-1}, \\
	\textbf{v}_1^{-i-1} &= v_1 \, v_2 \, \ldots v_{-i-1}.
	\end{align*}
	We claim that $\text{wt}(\textbf{s}_1^{-i-1}) \neq \text{wt}(\textbf{v}_1^{-i-1})$, which implies $\text{c}(\textbf{s}_1^{-i-1}) \neq \text{c}(\textbf{v}_1^{-i-1})$. This follows due to the Catalan constraint, which ensures that $\text{wt}(\textbf{s}) = \text{wt}(\textbf{v})$, the assumptions that $s_{-i} \neq v_{-i}$, $\textbf{s}_{-i+1}^{-1} = \textbf{v}_{-i+1}^{-1}$, $\textbf{s}_1^{i-1} = \textbf{v}_1^{i-1}$, and from the choice of $i$. 
	
As a result, we have 
	\begin{align*}
	\text{wt}(\textbf{s}_i^{-i}) = \text{wt}(\textbf{v}_i^{-i}).
	\end{align*}
	Next, we establish that $c(\textbf{s}_1^{-i-1}) \in C(\textbf{v}) \bigtriangleup C(\textbf{s})$. For any $1 < j \leq i+1$, we have the following equality that holds for substrings of $\textbf{s}$ of length $n-i$ :
	\begin{align*}
	\text{wt}(\textbf{s}_j^{-(i-j+2)}) = \text{wt}(\textbf{s}_j^{i-1}) + \text{wt}(\textbf{s}_i^{-i}) + \text{wt}(\textbf{s}_{-i+1}^{-(i-j+2)}).
	\end{align*}
	\textcolor{black}{To prove this result, we consider the strings of length $n-i$ that are in the symmetric difference $C(\textbf{v}) \bigtriangleup C(\textbf{s})$. In particular, we consider the following three cases:}
	\begin{enumerate}
		\item $j \leq i-1$,
		\item $j = i$,
		\item $j=i+1$.
	\end{enumerate}
	Clearly, for the first case it holds that
	\begin{align*}
	\text{wt}(\textbf{s}_j^{-(i-j+2)})=
	\text{wt}(\textbf{v}_j^{-(i-j+2)}).
	\end{align*}
	For the second case, due to the constraints that $s_{4t+2}\, s_{4t+3}\, \ldots\, s_{-4t-2} \in \P(n-2(4t+1))$, $s_1\, \ldots\, s_{4t+1}= 00\,\ldots\,0$ and $s_{-4t-1}\, s_{-4t}\, \ldots\, s_{-1} = 11\,\ldots\,1$, it follows that $\textbf{s}_1^{-i-1}$ contains more $0$s than $1$s, but $\textbf{v}_i^{-2}$ contains more $1$s than $0$s. A similar argument may be used for the third case, \textcolor{black}{and it can be shown in this case that $\bv_{i+1}^{-1}$ also contains more $1$s than $0$s, which implies that $c(\textbf{s}_1^{-i-1}) \in C(\textbf{v}) \bigtriangleup C(\textbf{s}),$ as desired. In other words, we consider substrings of length $n-i$ (because $\textbf{s}_1^{-i-1}$ has length $n-i$), of the form $\textbf{s}_j^{-(i-j+2)}$. For the case where $1<j \leq i-1$, the substrings of length $n-i$ in $\textbf{v}$ and $\textbf{s}$ have the same compositions, since $\text{wt}(\textbf{s}_j^{-(i-j+2)}) = \text{wt}(\textbf{v}_j^{-(i-j+2)})$. Thus, these substrings do not affect the compositions in $C(\textbf{v}) \Delta C(\textbf{s})$. This covers the first case, Case 1). Hence it remains to show that $c(\textbf{s}_1^{-i-1}) \neq c(\textbf{v}_j^{-i-j+2})$ for Cases 2) and 3) (when $i=j$ and $i=j+1$). For the case $i=j$, we have $c(\textbf{s}_1^{-(i-1)}) \neq c(\textbf{v}_i^{-2})$, since $\textbf{s}_1^{-(i-1)}$ has more $0$s than $1$s, whereas $\textbf{v}_i^{-2}$ has more $1$s than $0$s. For the case $j=i+1$, $c(\textbf{v}_{i+1}^{-1})$} also has more $1$s than $0$s. This completes the claim that for \textcolor{black}{$l=1,$} $c(\textbf{s}_l^{-i-1})  = c(\textbf{s}_1^{-i-1}) \in C(\textbf{v}) \Delta C(\textbf{s})$. The case $l \geq 2$ can be analyzed similarly.
	
Based on the discussion above, it is straightforward to identify additional substrings whose compositions lie in the symmetric difference of $C(\textbf{s})$ and $C(\textbf{v})$. In particular, if we can show that for every $l \in \{2,3,4,\ldots,4t+1\}$ one of the following two claims is true:
	\begin{enumerate}
		\item $c(\textbf{s}_l^{-i-1}) \in C(\textbf{v}) \bigtriangleup C(\textbf{s})$, or
		\item $c(\textbf{v}_l^{-i-1}) \in C(\textbf{v}) \bigtriangleup C(\textbf{s}).$
	\end{enumerate}
	then $ |C(\textbf{s}) \bigtriangleup C(\textbf{v})| \geq 4t+1$.
	
	For $l \in \{2,3,4,\ldots,4t+1\}$, it is straightforward to see that 
	\begin{align*}
	\text{wt}( \textbf{s}_l^{-i-1}) \neq \text{wt}( \textbf{v}_l^{-i-1}).
	\end{align*}
	Without loss of generality, we may assume that $\text{wt}( \textbf{s}_l^{-i-1}) < \text{wt}( \textbf{v}_l^{-i-1})$. Then $c(\textbf{s}_l^{-i-1}) \in C(\textbf{v}) \bigtriangleup C(\textbf{s}) $. 
	Similarly as before, for any $l < j \leq i+l$, the following holds for substrings of $\textbf{s}$ of length $n-i-l+1$: 
	\begin{align*}
	\text{wt}(\textbf{s}_j^{-(i-j+l+1)}) = \text{wt}(\textbf{s}_j^{i-1}) + \text{wt}(\textbf{s}_i^{-i}) + \text{wt}(\textbf{s}_{-i+1}^{-(i-j+l+1)}).
	\end{align*}
	If $j \leq i-1$, we have 
	\begin{align*}
	\text{wt}(\textbf{s}_j^{-(i-j+l+1)}) =& \text{wt}(\textbf{s}_j^{i-1}) + \text{wt}(\textbf{s}_i^{-i}) + \text{wt}(\textbf{s}_{-i+1}^{-(i-j+l+1)})\\
	=& \text{wt}(\textbf{v}_j^{i-1}) + \text{wt}(\textbf{v}_i^{-i}) + \text{wt}(\textbf{v}_{-i+1}^{-(i-j+l+1)}) \\
	=& \text{wt}(\textbf{v}_j^{-(i-j+l+1)}).
	\end{align*}
	For the case $j \geq i \geq 4t+\textcolor{black}{3}$, note that $\textbf{s}_l^{-i-1}$ contains more zeros than ones but for $j > i-1$, the substring $\textbf{v}_j^{-(i-j+l+1)}$ contains at least as many $1$s as $0$s. Therefore, for any $j > i-1$,
	\begin{align*}
	c(\textbf{s}_l^{-i-1}) \neq c(\textbf{v}_j^{-(i-j+l+1)}).
	\end{align*}
	
	We are left with analyzing the compositions of substrings of length $n-i-l+1$ in $\textbf{v}$ to the left of $\textbf{v}_l^{-i-1}$. Since every codestring in $\C(n,t)$ starts with $4t+1$ \textcolor{black}{$0$s}, it follows that for any $j < l$
	\begin{align*}
	\text{wt}(\textbf{v}_j^{-(i-j+l+1)}) \leq \text{wt}(\textbf{v}_{j-1}^{-(i-(j-1)+l+1)}).
	\end{align*}
	Furthermore, since $\text{wt}( \textbf{s}_l^{-i-1}) < \text{wt}( \textbf{v}_l^{-i-1})$, it follows that for any $j<l$,
	\begin{align*}
	\text{wt}( \textbf{s}_l^{-i-1}) < \text{wt}(\textbf{v}_j^{-(i-j+l+1)}).
	\end{align*}
	Thus, $c(\textbf{s}_l^{-i-1}) \in C(\textbf{v}) \bigtriangleup C(\textbf{s})$. This completes the proof.
\end{proof}

The result of Theorem~\ref{th:mainth} may be used to prove Theorem~\ref{th:new} since the number of redundant bits, $\cO(\log k + t),$ is a direct consequence of the code construction described in~\eqref{constr4}.

The reconstruction time for the described codes for a constant number of errors $t$ is polynomial in $n$. To see this, consider the ${{n+1 \choose 2} \choose t}$ possible choices for errors in distinct compositions. Each composition can be corrupted in at most $n$ different ways (for the composition corresponding to the whole string this number equals $n$). 
Thus, given an erroneous composition multiset $\tilde{C}(\textbf{s})$, there are at most ${{n+1 \choose 2} \choose t} n^t$ 
candidate true composition multisets $\{ \tilde{C}^1(\textbf{s}), \tilde{C}^2(\textbf{s}), \dots \tilde{C}^m(\textbf{s}) \}$, 
where $m = \cO(n^{3t})$.  Thus, by reconstructing the strings as given by the compositions $\{ \tilde{C}^1(\textbf{s}), \tilde{C}^2(\textbf{s}), \dots \tilde{C}^m(\textbf{s}) \}$ using the Backtracking algorithm, we can recover the string $\textbf{s}$ in $\cO(n^{3+3t}) $ time.

\section{Open Problems} \label{sec:open}

Many combinatorial and coding-theoretic problems related to mass error-correcting codes remain open and are listed below.
\begin{itemize}
\item In Sections~\ref{sec:recons},~\ref{sec:asymmetric} and~\ref{sec:symmetric} we showed that  the number of redundant bits sufficient for unique and efficient reconstruction without errors and in the presence of a constant number of $t$ errors equals $\cO(\log k)$ and $\cO(t^2 \log k)$, respectively. Lower bounds on the number of redundant bits are still unknown. 
\item The decoding algorithm used in the proof of Theorem~\ref{thm:asym} is efficient only if the number of errors $t$ is a constant. We do not know of any string reconstruction algorithms that are efficient both in $t$ and $n$. 
\item In our analysis, we made two simplifying assumptions described in the Introduction and previously used in~\cite{acharya2014string}. However, in reality one does not have access to the masses of all substrings but rather to corrupted masses of prefixes and suffixes of mixtures of strings. Mixing polymer strings also allows for faster readouts of information via MS/MS spectrometers. Therefore, a natural question is how to perform reconstruction of \emph{multiple} strings based on the union of their composition multisets or prefix-suffix sets. 
\item We addressed the string reconstruction problem when the errors are either asymmetric or symmetric. However, MS/MS errors are often bursty and context-dependent. Thus, studying more general error models is another problem of interest.   
\item Several problems outlined in~\cite{acharya2014string} at this time also remain open. We restate two of those problems for completeness: 1) Improve the upper and lower bounds on the number of \emph{confusable} strings; 2) Determine explicit polynomial-time algorithm for string reconstruction problems, the existence of which was established in~\cite{lenstra1985factoring, grigoryev1984factoring, kaltofen1985polynomial, kaltofen1990computing}. 
\end{itemize}

\section*{Acknowledgment} The authors gratefully acknowledge funding from the DARPA Molecular Informatics program, the NSF+SRC SemiSynBio program under award number 1807526 and the NSF grant number 1618366. 

\bibliography{biblio} 
\bibliographystyle{ieeetr}


\appendix
\section*{Appendices}

\subsection{Proof of the second part of Theorem~\ref{thm:BW}} \label{app:BW}
\begin{theorem*} 
\textcolor{black}{The central binomial coefficient ${2m \choose m }$ counts the following types of binary strings of length $2m$. \\
(A) Those whose every prefix has at least as many $0$s as $1$s. \\
(B) Those whose every prefix has strictly more $0$s than $1$s, or vice-versa.}
\end{theorem*} 
\begin{proof}
\textcolor{black}{ The number of binary strings of Type (A) of length $\ell = a+ b,$ with $\ell$ possibly odd, such that the number of $0$s is greater or equal to the number of $1$s, i.e., $a \geq b$ is given by ${ \ell \choose a} - {\ell \choose a+1}.$ The number of strings for which every prefix has at least as many $0$s as $1$s is given by $\sum_{a \geq \lceil \frac{\ell}{2} \rceil} { \ell \choose a} - {\ell \choose a+1},$ which is a telescoping sum that equals $ { \ell \choose \lceil \frac{\ell}{2} \rceil}.$ \\
To prove (B), let us consider strings of length $2m$ whose every prefix has strictly more $0$s than $1$s. 
In this case, the first bit of any string $\textbf{s}$ is always $0.$ Thus, the remaining length-$(2m-1)$ binary string $\textbf{s}_2^{2m}$ is such that for every prefix, the number of $0$s is at least as large as the number of $1$s in that same prefix. 
Thus, the number of strings of length $2m$ whose every prefix has strictly more $0$s than $1$s is ${ 2m -1 \choose m}.$ As a result, the total number of binary strings of length $2m$ whose every prefix has strictly more $0$s than $1$s or vice-versa is equal to $2 { 2m -1 \choose m} = {2m \choose m}.$  }
\end{proof}

\subsection{Derivation of the lower bound of $|\mathcal{S}_R(n)|$} \label{app:derivation}

\textcolor{black}{Let $n$ be even. Note that all strings $ \textbf{s} \in \mathcal{S}_R(n)$ satisfy $s_1 =0$ and $s_n =1.$ Let $|\left[ \frac{n}{2} \right] \cap I| = i+1.$ Thus, the indices corresponding to the Catalan-Bertrand string can be chosen in ${ \frac{n}{2}-1 \choose i }$ ways. Since $s_1 =0$, it must be that every prefix of $ \textbf{s}_{\left[ \frac{n}{2} \right] \cap I \setminus \{ 1\}}$ contains at least as many $0$s as $1$s. There are ${ i \choose \lfloor \frac{i}{2} \rfloor }$ such binary strings of length $i.$ Therefore, 
$$ |\mathcal{S}_R(n)| = \sum_{i=0}^{\frac{n}{2} -1} {\frac{n}{2} - 1 \choose i} 2^{\frac{n}{2} -1 -i} {i \choose  \lfloor \frac{i}{2} \rfloor}.$$ As a result, }
\textcolor{black}{
\begin{align}
    &\sum_{i=0}^{\frac{n}{2} -1} {\frac{n}{2} - 1 \choose i} 2^{\frac{n}{2} -1 -i} {i \choose \lfloor \frac{i}{2} \rfloor} \label{eqr1} \\
    &\geq   \sum_{i=2}^{\frac{n}{2} -1} {\frac{n}{2} - 1 \choose i} 2^{\frac{n}{2} -1 -i}  \frac{2^{i -1}}{\sqrt{\pi (i+1)}}   +   {\frac{n}{2} - 1 \choose 1} 2^{\frac{n}{2} -1 -1} + {\frac{n}{2} - 1 \choose 0} 2^{\frac{n}{2} -1 }      \label{eqr2} \\
    &\geq \frac{2^{\frac{n}{2} - 2}}{\sqrt{\pi n} } \sum_{i=0}^{\frac{n}{2}-1} { \frac{n}{2} -1 \choose i } \label{eqr3} \\
      & =  \frac{2^{\frac{n}{2} - 2}}{\sqrt{\pi n} } 2^{\frac{n}{2} - 1} = \frac{1}{\sqrt{\pi n}}2^{n-3}.  \label{eqr5}
\end{align}
Expression \eqref{eqr1} follows from the description of the codebook. Also, ${ 2\ell + 1 \choose \ell} \geq { 2\ell \choose \ell}$ clearly holds. As a result, inequality~\eqref{eqr2} follows from Proposition 1, for all $i \geq 2$. Inequality~\eqref{eqr3} holds since for all $0 \leq i \leq \frac{n}{2}$, $(i+1) \leq n$. The next two equalities in~\eqref{eqr5} follow from the fact that $\sum_{i=0}^{\ell} { \ell \choose i} = 2^\ell$, and some rearrangements of terms. }

\textcolor{black}{   For odd $n,$ $ |\mathcal{S}_R(n)| = 2  |\mathcal{S}_R(n-1)| \geq 2 \frac{ 2^{n-1-3} }{\sqrt{ \pi (n-1)}} \geq \frac{ 2^{n-3} }{\sqrt{ \pi n}}.$       }

\subsection{A bijective map between information strings and reconstructable strings} \label{app:construction}

\textcolor{black}{An optimal approach for performing encoding of information strings into Catalan string was first described in} \cite{durocher2012cool} \textcolor{black}{and it relies on using a ranking/unranking scheme of complexity $\mathcal{O}(n)$. However, we provide a much simpler method to order and retrieve the reconstructable strings in additive $\mathcal{O}(n^2)$ time, which is still absorbed in the leading complexity term of $\mathcal{O}(n^3)$ incurred by the Backtracking algorithm. \\
Recall that the reconstruction code is obtained by interleaving arbitrary, unconstrained strings with a Catalan-Bertrand strings and then mirroring the interleaved string around what will be the midpoint of the resulting codestring.}

\textcolor{black}{
1) The construction starts by partitioning the first $\lfloor \frac{n}{2} \rfloor$ indices into two sets, say $\mathcal{I}_0$ and $\mathcal{I}_1,$ the cardinalities of which are in $\{{0,\ldots,\lfloor \frac{n}{2} \rfloor\}}$.  Let $\mathcal{I}_0$ denote the set of indices that describe the locations of the string to be interleaved, and let $\mathcal{I}_1$ denote the indices that describe the locations of the Catalan-Bertrand string. \\
Next, order all possible partitions according to the cardinality of their corresponding $\mathcal{I}_0$ sets, in increasing order. For example, if $000111010$ and $001111010$ are the labels of two partitions of a string of length $9$, than $001111010$ appears in the rank-ordered list before $000111010$ (the first partition has $|\mathcal{I}_0|=4$, while the second partition has $|\mathcal{I}_0|=5>4$).\\
In the next step, order the partitions with the same value of $|\mathcal{I}_0|$. Given a partition described using the binary alphabet as above, one can convert the binary strings into integers and arrange them in increasing order which naturally induces a ranking of the partitions themselves. Finding the index of a partition in this ranking takes $\mathcal{O}(n)$ time. To see this, consider a partition $\Pi$ of $m=\lfloor \frac{n}{2} \rfloor$ indices, and let $|\mathcal{I}_0|=i.$} 
\textcolor{black}{ Thus, the rank of this partition is an integer in the interval $\left[\sum_{j=0}^{i-1}{ m \choose j }+1, \sum_{j=0}^{i}{m \choose j}\right].$ Assume that the set $\mathcal{I}_0$ contains the indices $( \ell_1, \ell_2, \dots , \ell_i)$ arranged in increasing order. The rank of the partition $\Pi$ is given by \\ 
$$ \sum_{j=0}^{i-1}{ m \choose j }+ \left[ { \ell_i -1 \choose i} + { \ell_{i-1} -1 \choose i-1} + { \ell_{i-2} -1 \choose i-2} + \dots + { \ell_{1} -1 \choose 1} + 1 \right].$$}

\textcolor{black}{Therefore, given the index of a partition, one can determine the actual partition in time $\mathcal{O}(n^2).$}

\textcolor{black}{
2) Next, given the indices in $\mathcal{I}_0,$ place unrestricted binary strings in the corresponding locations according to the lexicographical order.}

\textcolor{black}{
3) At indices in $\mathcal{I}_1,$ place bits of a Catalan-Bertrand string. Let us now assume that there exists a bijective map $F_{m}(\cdot)$ that for all natural numbers $m$ orders all Catalan-Bertrand strings of length $m$ efficiently. In particular, we assume that given an index \texttt{ind}, $F_{m}(\texttt{ind})$ returns the corresponding Catalan-Bertrand string in time $\mathcal{O}(n^2)$. Further, given a Catalan-Bertrand string $\textbf{s}$, $F^{-1}_{m}(\textbf{s})$ returns its index \texttt{ind} in $\mathcal{O}(n)$ time. We defer the description of the map to the end of this exposition. \\
Let $f_m(i)$ denote the number of Catalan-Bertrand strings with $m-i$ $0$s and $i$ $1$s. Then, $f_m = \sum_{i=0}^{\lfloor \frac{m}{2} \rfloor} f_m(i)$ is the number of all Catalan-Bertrand strings of length $m.$ Note that $f_m(i)$ has a closed form expression as given in Theorem 6, and $f_m$ equals $\frac{1}{2} {m \choose \lfloor \frac{m}{2} \rfloor}.$ \\
The ordering for the codestrings of the reconstruction code is obtained as follows:\\
a) Given two reconstructable codestrings $\textbf{s}_1$, and $\textbf{s}_2$, and their corresponding partitions $\Pi_{1}$ and $\Pi_{2}$ from 1), if $\Pi_1$  is ranked lower than $\Pi_2$, then $\textbf{s}_1$ is ranked lower than $\textbf{s}_2$.\\ 
b) Given two reconstructable codestrings $\textbf{s}_1$, and $\textbf{s}_2$ such that $\Pi_1=\Pi_2$, if the string of $\textbf{s}_1$ indexed by $\mathcal{I}_0$ is ranked lower than that of $\textbf{s}_2$ (as per 2)), then $\textbf{s}_1$ is ranked lower than $\textbf{s}_2.$\\
c) Given two reconstructable codestrings $\textbf{s}_1$ and $\textbf{s}_2$ such that $\Pi_1=\Pi_2$, and the strings of $\textbf{s}_1$ and $\textbf{s}_2$ indexed by $\mathcal{I}_0$ are the same, if the string indexed by $\mathcal{I}_2$ in $\textbf{s}_1$ is ranked lower than the string in $\textbf{s}_2$, then the string $\textbf{s}_1$ is ranked lower than $\textbf{s}_2.$ \\
In summary, the reconstructable codestrings are encoded and decoded as described below.\\
\textbf{Encoding:} \\
A $k-$bit binary string is converted into an index $\texttt{ind}.$ The time taken to find the corresponding partition and Catalan-Bertrand string is $\mathcal{O}(n).$ Combining this result with the result pertaining to the ranking map proves that the information string can be encoded in $\mathcal{O}(n^2)$ time.\\
\textbf{Decoding:}\\
Given a reconstructable codestring, its index can be computed in $\mathcal{O}(n)$ time. The $k-$bit binary expansion of the index uniquely determines the information string. Since the Backtracking algorithm takes $\mathcal{O}(n^3)$ time, the overall decoding time equals $\mathcal{O}(n^3).$}

\textcolor{black}{It remains to show that encoding and decoding of the Catalan-Bertrand strings can be performed in time $\mathcal{O}(n^2)$. Since the decoding process is easier to describe and leads to a straightforward approach for encoding, we start with the description of the decoding algorithm.\\
\textbf{Decoding Catalan-Bertrand strings:}
Let $\textbf{s} = s_1 s_2 \dots s_{m-1} s_{m}$ denote a Catalan-Bertrand string of length $m$ that contains $m-i$ $0$s and $i$ $1$s and recall that $f_m(i)$ denotes the number of such Catalan-Bertrand strings. We start by ranking the string $\textbf{s}$ against the set of all Catalan-Bertrand strings of length $m$ that contain $m-i$ $0$s and $i$ $1$s. The following simple algorithm determines the temporary index for $\textbf{s}$ in $\mathcal{O}(n)$ time. \\
\begin{align*}
    &\texttt{ind}_{temp} \gets f_m(i)   \\
    & l \gets i \\
    &\texttt{for }  j \texttt{ from } 0 \texttt{ to } m-1: \\
    & \qquad \qquad \texttt{ind}_{temp} = \texttt{ind}_{temp} - \mathbf{1}_{\{{s_{m-j} == 0\}}}\, f_{m-1-j}(l) \\
    & \qquad \qquad \texttt{if } s_{m-j} == 1: \\
    & \qquad \qquad\qquad \qquad l \gets l-1,
\end{align*}
where $\mathbf{1}$ denotes the indicator function. Note that $F_m(\cdot)$ then assigns the final index value $\sum_{\ell <i } f_m(i) + \texttt{ind}_{temp} $ to the given Catalan-Bertrand string $\textbf{s}$ of length $m$ in time $\mathcal{O}(n)$. \\
\textbf{Encoding Catalan-Bertrand strings:}
From the decoding procedure it is easy to deduce how to perform the encoding: Given $m$ and \texttt{ind}, we first find an $i$ such that $\textbf{s}$ has $m-i$ $0$s and $i$ $1$s. Then, iteratively, the bits $s_m$ through $s_1$ are computed using the correspondence between the bit value and the index range as described in the decoding process. Hence, encoding takes $\mathcal{O}(n^2)$ time.   
}

\subsection{Proof of Lemma~\ref{lem:awesome} } \label{app:lem4}

\textcolor{black}{Recall that we consider asymmetric errors, in which case a single error may occur either in $C_j$ or $C_{n+1 -j}$ but not both multisets. Furthermore, up to $t$ such errors are allowed. The presented code corrects such errors with at most $c_1 t \log k + c_2$ bits of redundancy, where $k$ is the length of the information string, and $c_1$, and $c_2$ are two positive constants. \\
The code construction involves two parts: 1) String $\textbf{s} \in \mathcal{S}_R(n_1)$ is padded with a prefix of $t$ 0s and a suffix of $t$ 1s to form an intermediate string $\textbf{s}'$ of length $n' + 2t$. \\
2) The $\Sigma^{\frac{n_1}{2}}$ is then encoded using a systematic $t-$ error correcting code and the redundant bits are placed in the middle of the string in manner such that the resultant string $\textbf{s}'' \in \mathcal{S}_R(n).$ \\
We show through a case-by-case analysis that the code is indeed a $t-$asymmetric error correcting code.}

Our analysis proceeds through multiple steps addressing different possible choices for the values of $\sigma_i, i=1,\ldots, \frac{n}{2},$ and the currently reconstructed bits (i.e., prefixes and suffixes of the codestring). The initial setting is depicted in Figure~\ref{fig:lema_figure1}. Each subsequent figure (Figures~\ref{fig:lema_figure2a},~\ref{fig:lema_figure2},~\ref{fig:lema_figure3},~\ref{fig:lema_figure3b},~\ref{fig:lema_figure3a},~\ref{fig:lema_figure4},~\ref{fig:lema_figure5} and~\ref{fig:lema_figure6}) explains how to extend two partially reconstructed strings from their prefix and suffix pairs so as to minimize the number of compositions they disagree in. For simplicity, such pairs are termed ``confusable'' and finding confusable pairs allows us to determine the minimum composition set differences between codestrings based on the Catalan-Bertrand construction. The final result establishes that the previous construction ensures a minimum composition set difference $ \geq 2(t +1)$.

First, we observe from Construction~\eqref{set1} that any pair of distinct strings $\textbf{s}, \textbf{v} \in \mathcal{S}_{R}^{(t)}(m)$ shares a prefix-suffix pair of length at least $t$ as all strings are padded by 0s and 1s on the left and right, respectively. 

Next, we characterize the conditions that allow one to identify strings that are ``closest'' to a codestring $\textbf{s}$. More precisely, we construct a set $ \mathcal{V}_{\textbf{s}}$ of strings such that for all $\textbf{v} \in \mathcal{V}_{\textbf{s}}$ one has: 1) $\textbf{v}$ and $\textbf{s}$ share the same $\Sigma^{\frac{m}{2}}$ sequence; 2) If the length of the longest shared prefix-suffix pair of 
$\textbf{v}$ and $\textbf{s}$ equals $i$, then for all $j \in \{ m-i-1, m-i-2, \dots, m-i-t-1 \}$ the inequality $|C_j(\textbf{s}) \setminus C_j(\textbf{v}) | \leq 2$ holds. These conditions summarize when a string may be confused with $\textbf{s}$ during the backtracking reconstruction procedure.

Recall that $c(\cdot)$ refers to the composition of its argument string. The substrings $\{ \textbf{s}_i^{i+j-1} \}, i=1,\ldots,m-j+1$ of $\textbf{s}$ of length $j$ share a common substring $\textbf{s}_{m+1-j}^{j},$ provided that $j > \frac{m}{2}$. For simplicity of notation, denote the composition of the common substring $\textbf{s}_{m+1-j}^{j}$ by $c_j$, \textit{i.e.}, let $c_j = c(\textbf{s}_{m+1-j}^{j})$.

We start with the following observation. If $\sigma_{i+1} \neq 1$, the two strings $\textbf{s}$ and $\textbf{v}$ necessarily share a prefix-suffix pair of length $i+1$, which contradicts the assumption that the longest prefix-suffix pair shared by the two strings is of length $i$. Thus, we have $\sigma_{i+1} =1$ and $|C_{m-i-1}(\textbf{s}) \setminus C_{m-i-1}(\textbf{v})|=2$, where the latter claim follows from the discussion pertaining to the single error-correction case: The compositions of length $m-i-1$ that are not shared by the two strings include $ \{ c(\textbf{s}_1^i), 0, c_{m-i-1} \}$, $\{c(\textbf{s}_{m+1-i}^m),1,c_{m-i-1} \}$ , $\{ c(\textbf{v}_1^i), 1, c_{m-i-1} \}$, $\{ c(\textbf{v}_{m+1-i}^m), 0, c_{m-i-1} \}$, and these differ by construction. 

Next, we describe how to simultaneously reconstruct a pair of prefix-suffix bits and update the set $\mathcal{V}_{\textbf{s}}$ when taking a step in the Backtracking algorithm. We show that under the conditions of the lemma, $|C_{m-i-1 -j}(\textbf{s}) \setminus C_{m-i-1-j}(\textbf{v})|= 2$ for all $\textbf{v} \in \mathcal{V}_{\textbf{s}}, 1 \leq j \leq t$. For notational simplicity, at every step of the reconstruction algorithm we use the index $``+"$ to denote the next bit in the prefix and $``-"$ to denote the next bit in the suffix to be reconstructed. As an example, for a reconstructed prefix-suffix pair of length $i+1$, $+$ corresponds to $i+2$ and $-$ corresponds to $m-i-1,$ \textit{i.e.}, $s_+ = s_{i+2}$ and $s_- = s_{m-i-1}$.

Let $\sigma_+ = \text{wt}(s_{+} s_{-}) = \text{wt}(v_{+} v_{-})$. We analyze the two cases $\sigma_+ =1$ and $\sigma_+ \in \{ 0,2\}$ separately, as depicted in Figure~\ref{fig:lema_figure1}. 

Consider the case that $\sigma_+ =1$. Note that for any substring $\textbf{s}_{\ell_1}^{\ell_2}$ \textcolor{black}{such that $\ell_1 \leq i+1, m-i \leq \ell_2$, the corresponding substring} $\textbf{v}_{\ell_1}^{\ell_2}$ of $\textbf{v}$ has the same composition. The compositions in $C_{m-i-2}(\textbf{s})$ and $C_{m-i-2}(\textbf{v})$ that may be confused are listed below on the left and right hand side of the equality, respectively:

\begin{figure*}[h!]
\centering
  \includegraphics[scale= 0.45]{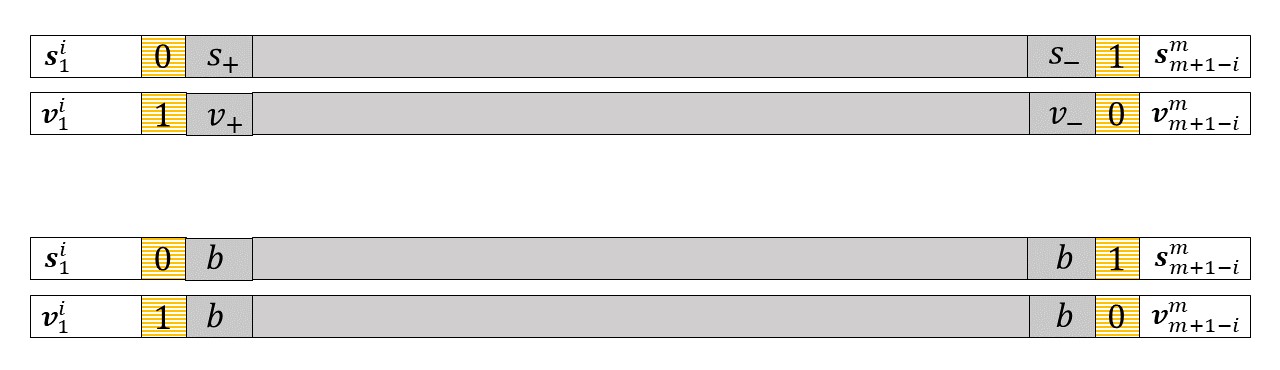}
  \caption{Illustration of two strings $\textbf{s}$ and $\textbf{v}$ that share the same $\Sigma^{\frac{m}{2}}$ sequence. Furthermore, the two strings also satisfy $\textbf{s}_1^i = \textbf{v}_1^i$, $\textbf{s}_{m+1-i}^m = \textbf{v}_{m+1-i}^m$ and $s_{i+1} \not = v_{i+1}$, \textit{i.e.}, the longest prefix-suffix pair that the strings share is of length $i$. The top pair of strings corresponds to the case $\sigma_{i+2} =1,$ while the bottom pair of strings corresponds to the case $\sigma_{i+2} \in \{ 0,2\}$.}
\label{fig:lema_figure1}
\vspace{-0.08in}
\end{figure*}

\begin{align*}
& \begin{rcases}
 \begin{dcases}
\{{c(\textbf{s}_1^i), 0, s_+, c_{m-i-2}\}}, \\
\{{c(\textbf{s}_2^i), 0, s_+, c_{m-i-2} ,1-s_+\}}, \\
\{{\textcolor{black}{c(\textbf{s}_{m-i+1}^m)},1, 1-s_+, c_{m-i-2}\}},  \\
\{{\textcolor{black}{c(\textbf{s}_{m-i+1}^{m-1})}, 1, 1-s_+, c_{m-i-2}, s_+\}}
\end{dcases} 
\end{rcases} =
\begin{rcases}
    \begin{dcases}
\{{c(\textbf{v}_1^i),1, v_+, c_{m-i-2}\}}, \\
\{{c(\textbf{v}_2^i), 1, v_+, c_{m-i-2}, 1-v_+\}}, \\
\{{c(\textbf{v}_{m-i+1}^m), 0, 1-v_+, c_{m-i-2}\}},  \\
\{{c(\textbf{v}_{m-i+1}^{m-1}), 0, 1-v_+, c_{m-i-2}, v_+\}}
	\end{dcases} 
\end{rcases}. 
\end{align*}

We want to determine under which conditions the terms on the two sides of the equality can be perfectly matched; in the process, we will show that $|c_{m-i-2}(\textbf{s}) \setminus c_{m-i-2}(\textbf{v})| \leq 2$. 

The above sets may be more succinctly written as:
\begin{align*}
& \begin{rcases}
    \begin{dcases}
\{{c(\textbf{s}_1^i), 0, s_+, c_{m-i-2}\}}, \\
\{{c(\textbf{s}_2^i), 0^2 1, c_{m-i-2} \}}, \\
\{{c(\textbf{s}_{m-i+1}^m),1, 1-s_+, c_{m-i-2}\}},  \\
\{{c(\textbf{s}_{m-i+1}^{m-1}), 01^2, c_{m-i-2}\}}
\end{dcases} 
\end{rcases}  =
\begin{rcases}
    \begin{dcases}
\{{c(\textbf{v}_1^i),1, v_+, c_{m-i-2}\}}, \\
\{{c(\textbf{v}_2^i), 01^2, c_{m-i-2} \}}, \\
\{{c(\textbf{v}_{m-i+1}^m), 0, 1-v_+, c_{m-i-2}\}},  \\
\{{c(\textbf{v}_{m-i+1}^{m-1}), 0^21, c_{m-i-2}\}}
	\end{dcases} 
\end{rcases}.  
\end{align*}

Regrouping the a priori known extension bits with the prefixes and suffixes simplifies the sets to be matched as
\begin{align*}
&\begin{rcases}
    \begin{dcases}
\{{c(\textbf{s}_1^i), 0, s_+, c_{m-i-2}\}}, \\
\{{c(\textbf{s}_1^i), 0 1, c_{m-i-2} \}}, \\
\{{c(\textbf{s}_{m-i+1}^m),1, 1-s_+, c_{m-i-2} \}},  \\
\{{c(\textbf{s}_{m-i+1}^{m}), 01, c_{m-i-2} \}}
\end{dcases} 
\end{rcases} =
\begin{rcases}
    \begin{dcases}
\{{c(\textbf{v}_1^i),1, v_+, c_{m-i-2} \}}, \\
\{{c(\textbf{v}_1^i), 1^2, c_{m-i-2} \}}, \\
\{{c(\textbf{v}_{m-i+1}^m), 0, 1-v_+, c_{m-i-2}\}},  \\
\{{c(\textbf{v}_{m-i+1}^{m}), 0^2, c_{m-i-2} \}}
	\end{dcases} 
\end{rcases}.  
\end{align*}

For example, $\{{c(\textbf{s}_2^i), 0^2 1, c_{m-i-2} \}}$ is rewritten as $\{{c(\textbf{s}_1^i), 0^1 1, c_{m-i-2} \}}$ by moving one $0$ to the prefix composition.

Next, we remove the compositions $c_{m-i-2}$ shared by the two sets. Then we identify which compositions cannot be matched as follows. First, \textcolor{black}{it follows from the construction} that the composition of a prefix of length $i>t$ includes at least $t+1$ $0$s.  \textcolor{black}{As a result, $c(\textbf{s}_1^i)$ is composed of at least $t+1$ more $0$s than $c(\textbf{s}_{m-i+1}^m)$. Similarly, $c(\textbf{s}_{m-i+1}^m)$ is composed of at least $t+1$ more $1$s than $c(\textbf{s}_1^i)$. Hence, a composition involving less than $i+t+1$ bits that contains a composition of a prefix of length $i > t$ is composed of more $0$s than a composition of the same length that contains a composition of a suffix of length $i $. Thus, compositions $\{ c(\textbf{s}_1^i), 0, s_+) \}, \{c(\textbf{s}_1^i), 0 1) \} $ are not the same as either of the compositions $ \{{c(\textbf{v}_{m-i+1}^m), 0, 1-v_+}\}, \{{c(\textbf{v}_{m-i+1}^{m}), 0^2}\}$, since $c(\textbf{s}_1^i)$ contains at least $t+1$ more $0$s than $c(\textbf{v}_{m-i+1}^{m})$.}
Therefore, we only need to consider the two reduced set equalities:  
\begin{equation*}
\begin{rcases}
    \begin{dcases}
\{{c(\textbf{s}_1^i), 0, s_+\}}, \\
\{{c(\textbf{s}_1^i), 0 1 \}}
\end{dcases} 
\end{rcases}
=
\begin{rcases}
    \begin{dcases}
\{{c(\textbf{v}_1^i),1, v_+\}}, \\
\{{c(\textbf{v}_1^i), 1^2 \}}
	\end{dcases} 
\end{rcases},  
\end{equation*}
and 
\begin{equation*}
\begin{rcases}
    \begin{dcases}
\{{c(\textbf{s}_{m-i+1}^m),1, 1-s_+\}},  \\
\{{c(\textbf{s}_{m-i+1}^{m}), 01\}}
\end{dcases} 
\end{rcases}
=
\begin{rcases}
    \begin{dcases}
\{{c(\textbf{v}_{m-i+1}^m), 0, 1-v_+\}},  \\
\{{c(\textbf{v}_{m-i+1}^{m}), 0^2\}}
	\end{dcases} 
\end{rcases}.  
\end{equation*}

Clearly, $\{{c(\textbf{v}_1^i), 1^2 \}}$ and $\{{c(\textbf{v}_{m-i+1}^{m}), 0^2\}}$ cannot be equal to any other composition in the two sets. The possible values for the set difference $|C_{m-i-2}(\textbf{s}) \setminus C_{m-i-2}(\textbf{v})|$ for four different assignments of values for $(s_+,v_+)$ are summarized in Table~\ref{tab:yellow1}. Based on the table, if $\sigma_{i+2}(\textbf{s}) =1,$ then all strings $\textbf{v} \in \mathcal{V}_{\textbf{s}}$ satisfy $(s_{+}, v_{+})=(s_{i+2}, v_{i+2}) \in \{ (0,0), (1,0)\}$.

Next, we consider the case $\sigma_+ \in \{ 0,2 \}$. As before, we focus on $C_{m-i-2}(\textbf{s})$ and $C_{m-i-2}(\textbf{v})$ in order to establish conditions under which $|C_{m-i-2}(\textbf{s}) \setminus C_{m-i-2}(\textbf{v})|$ is minimized. 

To this end, let $b = s_+ = v_+ = s_- = v_-$. \textcolor{black}{Following the previously outlined line of reasoning,} it suffices to find when the following set equalities hold: 
\begin{equation*}
\begin{rcases}
    \begin{dcases}
\{{c(\textbf{s}_1^i), 0, b \}}, \\
\{{c(\textbf{s}_1^i), 0 1 \}}
\end{dcases} 
\end{rcases}
=
\begin{rcases}
    \begin{dcases}
\{{c(\textbf{v}_1^i),1, b\}}, \\
\{{c(\textbf{v}_1^i), 1^2 \}}
	\end{dcases} 
\end{rcases}  
\end{equation*}
and 
\begin{equation*}
\begin{rcases}
    \begin{dcases}
\{{c(\textbf{s}_{m-i+1}^m),1, b\}},  \\
\{{c(\textbf{s}_{m-i+1}^{m}), 01\}}
\end{dcases} 
\end{rcases}
=
\begin{rcases}
    \begin{dcases}
\{{c(\textbf{v}_{m-i+1}^m), 0, b\}},  \\
\{{c(\textbf{v}_{m-i+1}^{m}), 0^2\}}
	\end{dcases} 
\end{rcases}.  
\end{equation*}
It can be easily seen that the compositions cannot be matched. The possible cardinalities of the set difference $|C_{m-i-2}(\textbf{s}) \setminus C_{m-i-2}(\textbf{v})|$ are summarized in Table~\ref{tab:yellow2}.

As a result of the above discussion, for any $\textbf{v} \in \mathcal{V}_{\textbf{s}}$ we necessarily have $(s_{i+2}, v_{i+2}) \in \{ (0,0), (1,0)\}$ and $\sigma_{i+2} = 1$.
This consequently determines the pair of bits $s_{m-i-1}$ and $v_{m-i-1}$. 

To determine $s_{i+3}$,$s_{m-i-2}$, $v_{i+3}$ and $v_{m-i-2}$ we need to once again analyze two cases, one for which we assume that $\sigma_{i+3} = 1$ and another, for which we assume that $\sigma_{i+3} \in \{ 0,2\}$. This analysis has to be performed in the context depicted in Figure~\ref{fig:lema_figure1}, and under the constraints imposed by Tables~\ref{tab:yellow1} and~\ref{tab:yellow2}.

We focus on the bits $s_{i+2+i'}$ and $v_{i+2+i'} $ for some $i'$ such that $t-1 \geq i' \geq 0,$ in the following inductive setting:
\begin{itemize}
\item Assume that starting from the index $i+2$, the values of $\sigma$ corresponding $i'$ consecutive positions all equal to $1$. More precisely, $\sigma_{i+2}^{i+1+i'} = (1, 1, \dots 1)$.
\item The bits $s_{i+1}$ and $v_{i+1}$ are followed by a run of $i'$ $0$s, \textit{i.e.}, $\textbf{s}_{i+2}^{i+1+i'} = \textbf{v}_{i+2}^{i+1+i'} = \textbf{0}$.
\end{itemize} 
This setting is depicted in Figure~\ref{fig:lema_figure2a}. We proceed to characterize the conditions under which $|C_{m-i-i'-2}(\textbf{s}) \setminus C_{m-i-i'-2}(\textbf{v})|$ is minimized. As done before, we consider the cases $\sigma_{i+2+i'} = 1$ and $\sigma_{i+3+i'} \in \{ 0,2\}$ separately. 

When $\sigma_+ \in \{0,2 \}$, we assume that $s_+ = s_{-} = v_+ =  v_{-} = b$. The set equality of interest reads as:
\begin{align*}
& \begin{rcases}
    \begin{dcases}
\{{c(\textbf{s}_1^i), 0, 0^{i'},s_+\}}, \\
\{{c(\textbf{s}_2^i), 0, 0^{i'}, 01 \}}, \\
\{{c(\textbf{s}_3^i), 0, 0^{i'}, 01^2 \}} \\
\{{c(\textbf{s}_4^i), 0, 0^{i'}, 01^3 \}} \\
\qquad \qquad \vdots \\
\{{c(\textbf{s}_{i'+2}^i), 0, 0^{i'}, 01^{i'+1} \}} \\
\{{c(\textbf{s}_{m-i+1}^m),1, 1^{i'}, 1-s_+\}},  \\
\{{c(\textbf{s}_{m-i+1}^{m-1}),1, 1^{i'}, 01\}},  \\
\{{c(\textbf{s}_{m-i+1}^{m-2}),1, 1^{i'}, 0^21\}},  \\
\{{c(\textbf{s}_{m-i+1}^{m-3}),1, 1^{i'}, 0^31\}},  \\
\qquad \qquad \vdots \\
\{{c(\textbf{s}_{m-i+1}^{m-i'-1}), 1, 1^{i'}, 0^{i'+1}1 \}}
\end{dcases} 
\end{rcases}  =
\begin{rcases}
    \begin{dcases}
\{{c(\textbf{v}_1^i), 1, 0^{i'},v_+\}}, \\
\{{c(\textbf{v}_2^i), 1, 0^{i'}, 01 \}}, \\
\{{c(\textbf{v}_3^i), 1, 0^{i'}, 01^2 \}} \\
\{{c(\textbf{v}_4^i), 1, 0^{i'}, 01^3 \}} \\
\qquad \qquad \vdots \\
\{{c(\textbf{v}_{i'+2}^i), 1, 0^{i'}, 01^{i'+1} \}} \\
\{{c(\textbf{v}_{m-i+1}^m),0, 1^{i'}, 1-v_+\}},  \\
\{{c(\textbf{v}_{m-i+1}^{m-1}),0, 1^{i'}, 01\}},  \\
\{{c(\textbf{v}_{m-i+1}^{m-2}),0, 1^{i'}, 0^21\}},  \\
\{{c(\textbf{v}_{m-i+1}^{m-3}),0, 1^{i'}, 0^31\}},  \\
\qquad \qquad \vdots \\
\{{c(\textbf{v}_{m-i+1}^{m-i'-1}), 0, 1^{i'}, 0^{i'+1}1 \}}
	\end{dcases} 
\end{rcases}.  
\end{align*}

\begin{figure*}[h!]
\centering
  \includegraphics[scale= 0.45]{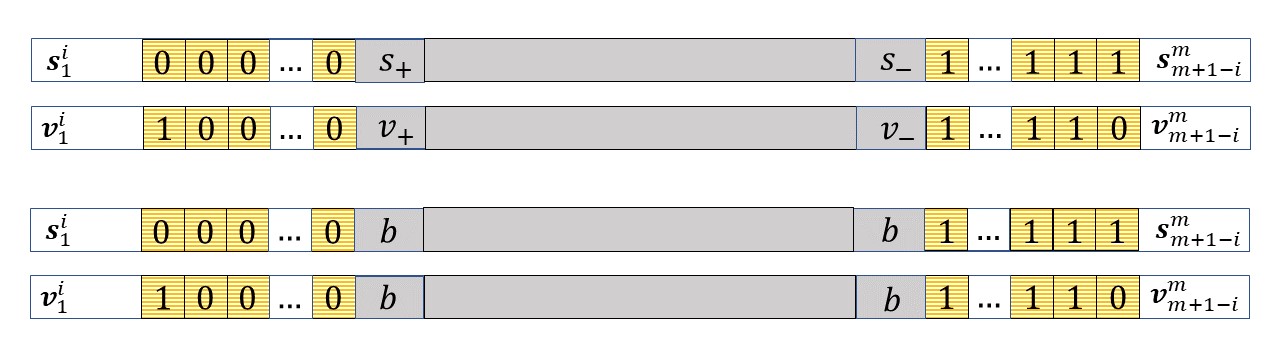}
  \caption{Illustration of the setup for determining the bits $s_+, s_-, v_+ \text{ and }v_- $ under the conditions that the bits $s_{i+1}$ and $v_{i+1}$ are followed by a run of $i'$ $0$s, and $\sigma_{i+2}^{i+1+i'} = (1, 1, \dots , 1)$. The second pair of strings illustrates the setting for which $\sigma_{i+2+i'} \in \{ 0,2\} $, and $s_+ = s_- = v_+ = v_- = b$.}
\label{fig:lema_figure2a}
\vspace{-0.08in}
\end{figure*}

Using the same line of reasoning as presented earlier, one can show that it suffices to focus on two reduced set equalities, namely
\begin{equation*}
\begin{rcases}
    \begin{dcases}
\{{c(\textbf{s}_1^i), 0^{i'+1},s_+\}}, \\
\{{c(\textbf{s}_1^i), 0^{i'+1}1 \}}
\end{dcases} 
\end{rcases}
=
\begin{rcases}
    \begin{dcases}
\{{c(\textbf{v}_1^i),  0^{i'}1,v_+\}}, \\
\{{c(\textbf{v}_{1}^i), 1^{i'+2} \}} 
\end{dcases} 
\end{rcases},
\end{equation*}
and
\begin{align*}
& \begin{rcases} 
    \begin{dcases}
    \{{c(\textbf{s}_{m-i+1}^{m}), 01^{i'+1}\}},\\
\{{c(\textbf{s}_{m-i+1}^m), 1^{i'+1}, 1-s_+\}} 
\end{dcases}
\end{rcases}  =
\begin{rcases}
    \begin{dcases}
    \{{c(\textbf{v}_{m-i+1}^{m}), 0^{i'+2} \}}, \\
\{{c(\textbf{v}_{m-i+1}^m), 01^{i'}, 1-v_+\}}
	\end{dcases} 
\end{rcases}.
\end{align*}

The possible values of $|C_{m-i-i'-2}(\textbf{s}) \setminus C_{m-i-i'-2}(\textbf{v})|$ are summarized in Table~\ref{tab:yellow3}.

We now turn our attention to the case $\sigma_{i+i'+2} \in \{ 0,2\}$. Again, let $b= s_+ = s_- = v_+ = v_-$. It suffices to consider the following sets:
\begin{equation*}
\begin{rcases}
    \begin{dcases}
\{{c(\textbf{s}_1^i), 0^{i'+1}, b \}}, \\
\{{c(\textbf{s}_1^i), 0^{i'}, b^2  \}}
\end{dcases} 
\end{rcases}
=
\begin{rcases}
    \begin{dcases}
\{{c(\textbf{v}_1^i), 0^{i'}1, b \}}, \\
\{{c(\textbf{v}_{2}^i), b^2, 1^{i'+1} \}}
\end{dcases} 
\end{rcases}
\end{equation*}
and
\begin{equation*}
\begin{rcases}
    \begin{dcases}
\{{c(\textbf{s}_{m-i+1}^m), 1^{i'+1}, b \}},  \\
\{{c(\textbf{s}_{m-i+1}^{m}), 1^{i'}, b^2 \}}
\end{dcases} 
\end{rcases}
=
\begin{rcases}
    \begin{dcases}
\{{c(\textbf{v}_{m-i+1}^m), 01^{i'}, b \}},  \\
\{{c(\textbf{v}_{m-i+1}^{m-1}), 0^{i'+1}, b^2 \}}
\end{dcases} 
\end{rcases}.
\end{equation*}
The possible values of $|C_{m-i-i'-2}(\textbf{s}) \setminus C_{m-i-i'-2}(\textbf{v})|$ are summarized in Table~\ref{tab:yellow4}.

From the above analysis we can conclude that exactly one of the following two conditions holds: 
\begin{enumerate}
\item The strings $\textbf{s}$ and  $\textbf{v}$ satisfy $\textbf{s}_{i+2}^{i+t+1} = \textbf{v}_{i+2}^{i+t+1} = \textbf{0}$ and $\sigma_{i+1}^{i+t+1} = (1,1, \dots,1)$. Their corresponding composition multisets $C_{m-i-1}, C_{m-i-2}, \dots, C_{m-i-t}, C_{m-i-t-1}$ each differ in exactly $2$ compositions.
\item The strings $\textbf{s}$ and $\textbf{v}$ satisfy $\textbf{s}_{i+2}^{i+1+i'} = \textbf{v}_{i+2}^{i+1+i'} = \textbf{0},$ $\sigma_{i+2}^{i+2+i'} = (1,1, \dots,1)$, and $(s_{i+i'+2},v_{i+i'+2}) = (1,0),$ where $t > i' \geq 0$. Their corresponding composition multisets $C_{m-i-1}, C_{m-i-2}, \dots, C_{m-i-i'-1}, C_{m-i-i'-2}$ each differ in exactly $2$ compositions.  
\end{enumerate} 
Figure~\ref{fig:lema_figure2} illustrates the observations. The longest substring such that $(s_{i+1}, v_{i+1}) =(0,1)$, $(s_{i+2} , v_{i+2})=(0,0)$ , $\dots$, $(s_{i+i'+1}, v_{i+i'+1})=(0,0)$ and $\sigma_{i+2}^{i+i'+2} = (1,\dots,1)$ is depicted by a horizontal block in Figure~\ref{fig:lema_figure2}. 
The bits $s_{i+i'+2}$, $s_{m-i-i'-1}$, $v_{i+i'+2}$, $v_{m-i-i'-1}$ that terminate the $00\dots 0$ (in \textbf{s}) and $10\dots 0$ (in \textbf{v}) substrings in the prefix and the $1\dots 1 1$ (in \textbf{s}) and $1\dots 10$ (in \textbf{v}) substrings in the suffix are represented by vertical shades in Figure~\ref{fig:lema_figure2}.  

\begin{figure*}[h!]
\centering
  \includegraphics[scale= 0.45]{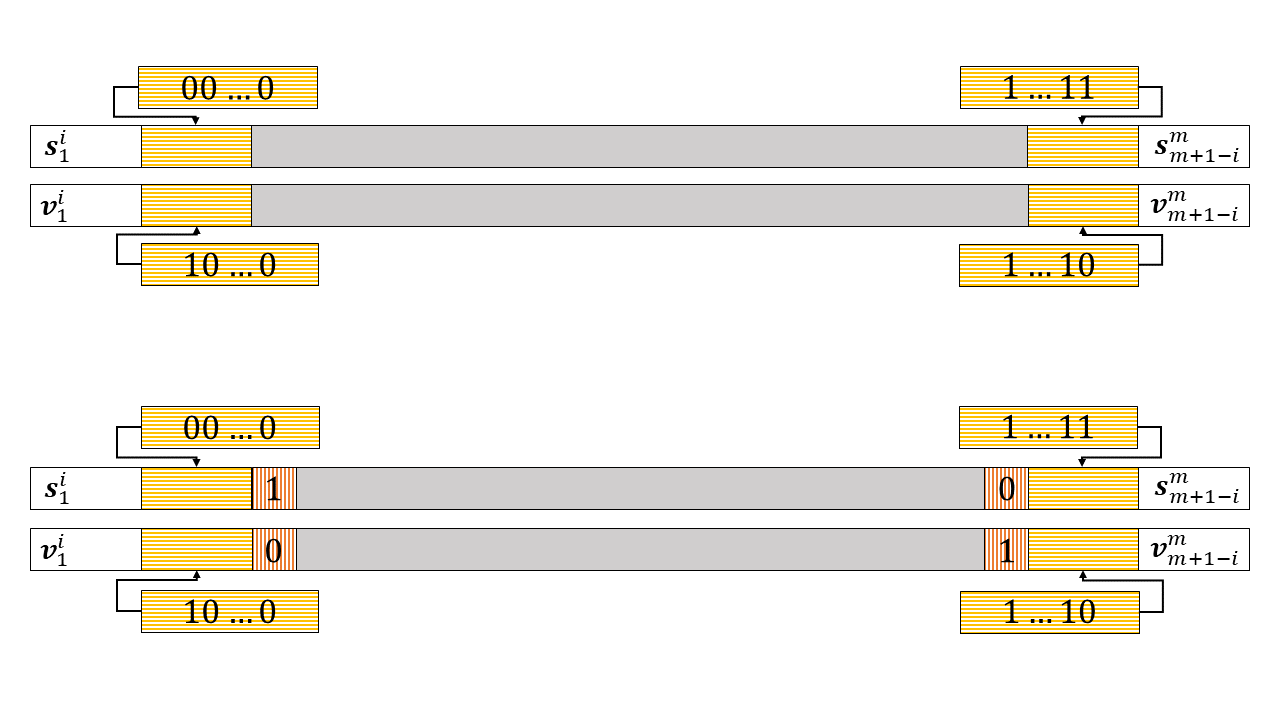}
  \caption{Illustration of the procedure for determining the set $\mathcal{V}_{\textbf{s}}$ based on several special cases. For the first case, we have $\textbf{s}_{i+2}^{i+t+1} = \textbf{v}_{i+2}^{i+t+1} = \textbf{0}$ and $\textbf{s}_{m-i-t}^{m-i-1} = \textbf{v}_{m-i-t}^{m-i-1} = \textbf{1}$. For the second case, there exist two identical substrings $\textbf{s}_{i+2}^{i+1+i'} = \textbf{v}_{i+2}^{i+1+i'} = \textbf{0}$ of length $t > i' \geq 0$ each and it holds that $(s_{i+i'+2}, v_{i+i'+2}) = (1,0)$.} 
\label{fig:lema_figure2}
\vspace{-0.08in}
\end{figure*}

Assume that the running reconstructions of the distinct strings $\textbf{s}$ and $\textbf{v}$ are as depicted in the second pair of blocks in Figure~\ref{fig:lema_figure2}. In the next step, illustrated in Figure~\ref{fig:lema_figure3}, we extend the prefixes and suffixes and identify the conditions under which $|C_{m-i-i'-3}(\textbf{s}) \setminus C_{m-i-i'-3}(\textbf{v})|$ is minimized. The results are summarized in 
Tables~\ref{tab:red1} and~\ref{tab:red2}.
\begin{figure*}[h!]
\centering
  \includegraphics[scale= 0.45]{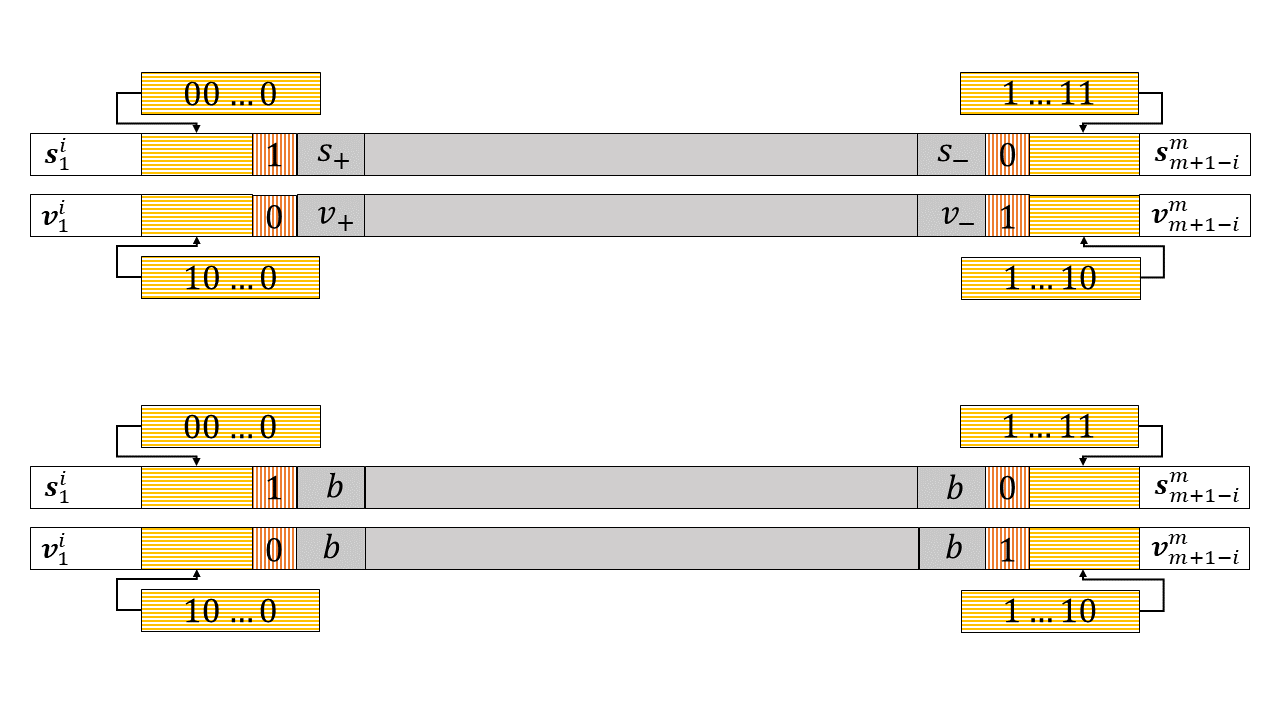}
  \caption{Illustration of the reconstruction step following the one depicted in Figure~\ref{fig:lema_figure2}. The first pair of strings corresponds to the case $\sigma_+ =1$, while the second pair of strings corresponds to the case $\sigma_+ \in \{ 0,2\} $ and $s_+ = s_- = v_+ = v_- = b$.}
  \label{fig:lema_figure3}
\vspace{-0.08in}
\end{figure*} 
In this step, we examine the bits $s_{i+i'+r+3}, v_{i+i'+r+3},$ $s_{m-i-i'-r-2}$ and $v_{m-i-i'-r-2}$.\\ 

Assume that $$\textbf{s}_{i+i'+3}^{i+i'+r+2} = \textbf{v}_{i+i'+3}^{i+i'+r+2} = b_1 \dots b_r,$$ where $r > 0$ and $r=0$ corresponds to a string of length $0$. We have $$\textbf{s}_{m-i-i'-r-1}^{m-i-i'-2} = \textbf{v}_{m-i-i'-r-1}^{m-i-i'-2} = \bar{b}_r \dots \bar{b}_1,$$ where
$$\bar{b}_i = 
\begin{cases}
b_i, \text{ if } \sigma_i \neq 1, \\
1- b_i, \text{ if } \sigma_i = 1, 
\end{cases} \text{ for all } 1 \leq i \leq r. 
$$ 
Such a structure is illustrated in Figure~\ref{fig:lema_figure3b}.

\begin{figure*}[h!]
\centering
  \includegraphics[scale= 0.45]{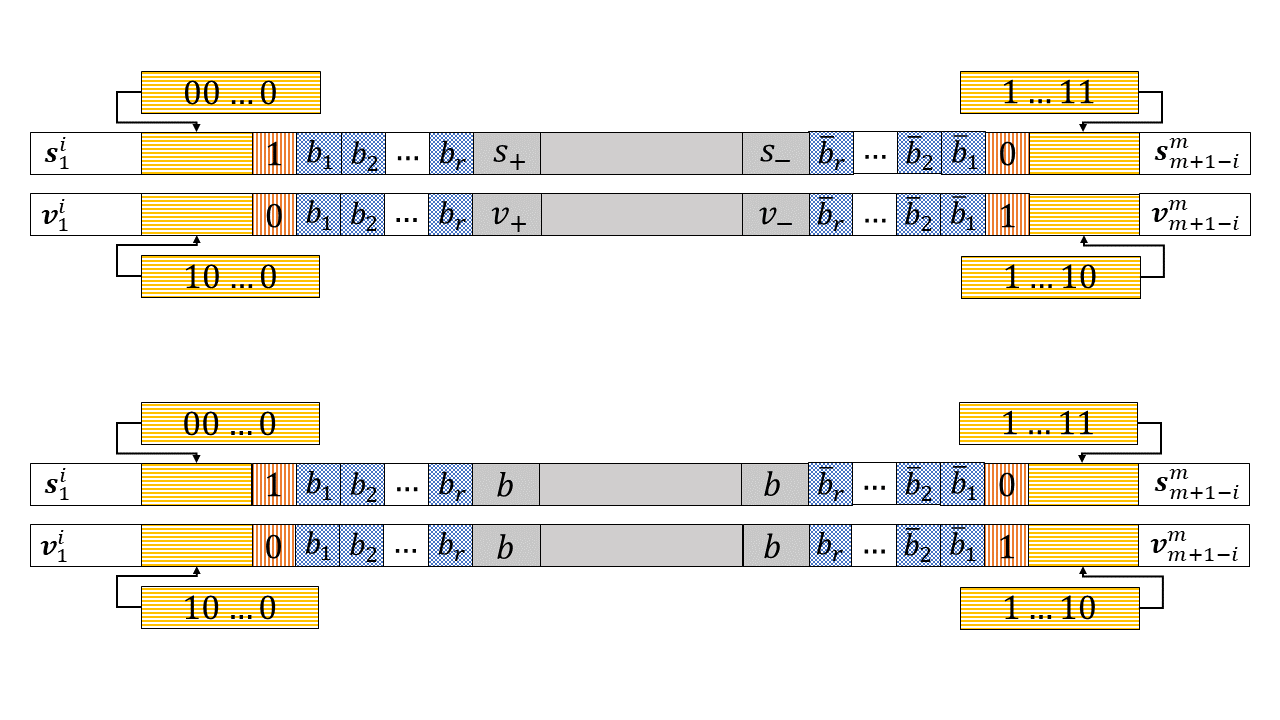}
  \caption{Two pairs of strings explaining how to extend the partially reconstructed strings illustrated in Figure~\ref{fig:lema_figure3}. The $r$ bits that follow the substring $00\dots 0 \, 1$ in $\textbf{s}$ are equal to the corresponding $r$ bits in $\textbf{v}$. For all $ r \geq i \geq 1$, $\bar{b}_i = b_i$ or $\bar{b}_i = 1- b_i$. The first pair corresponds to $\sigma_+ =1$, while the second pair corresponds to $\sigma_+ \in \{ 0,2\} $.}
\label{fig:lema_figure3b}
\vspace{-0.08in}
\end{figure*}

For the case $(s_+ , v_+) \not = (0,1)$, it is straightforward to see using arguments similar to the ones previously described that the possible set differences are as listed in Tables~\ref{tab:blue1} and~\ref{tab:blue2}.

For the case $(s_+ , v_+) = (0,1)$ depicted in Figure~\ref{fig:lema_figure3a}, the conditions that ensure that the composition multisets of $\textbf{s}$ and $\textbf{v}$ differ by at most $2$ introduce the restrictions $b_1, \dots , b_r = 1 \dots 1$ and $\bar{b}_1, \dots , \bar{b}_r = 0 \dots 0$. 

\begin{figure*}[h!]
\centering
  \includegraphics[scale= 0.45]{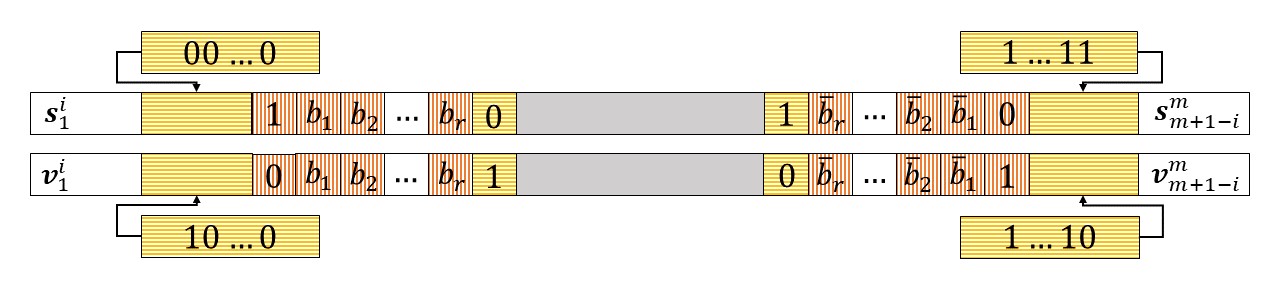}
  \caption{Conditions on the values of $b_i$ and $\bar{b}_i$ for all $i$ such that \textcolor{black}{$ r \geq i \geq 1$} that ensure that the partially reconstructed strings from the previous step can be compatibly extended when $\sigma_+ =1$ and $(s_+ = 0, v_+=1)$. }
\label{fig:lema_figure3a}
\vspace{-0.08in}
\end{figure*}

We now extend the description of the set $\mathcal{V}_{\textbf{s}}$ illustrated in Figure~\ref{fig:lema_figure2} as shown in Figure~\ref{fig:lema_figure4}.

Given a pair of distinct strings depicted in the second row of Figure~\ref{fig:lema_figure2}, one of the conditions must hold:
\begin{itemize}
\item The reconstructed prefix of \textbf{s} is followed by a substring $b_1 b_2 \dots b_r$ that is shared by the two strings and is such that the length of the substrings $00\dots 0 \, 1 \, b_1 b_2 \dots b_r$ (in \textbf{s}) and $10\dots 0 \, 0 \, b_1 b_2 \dots b_r$ (in \textbf{v}) in the prefixes equals $t+1$. In this case, each pair of composition multisets in $C_{m-i-1}, C_{m-i-2}, \dots, C_{m-i-t}, C_{m-i-t-1}$ differs in exactly $2$ compositions.
\item The reconstructed prefix in \textbf{s} is followed by the substring $1\dots 1 \, 0$ and the reconstructed prefix in \textbf{v} is followed by the substring $1\dots 1 \, 1$. The length of the substrings $00\dots 0 \, 11\dots 1 \, 0$ and $10\dots 0 \, 01\dots 1 \, 1$ is equal to some $0<j'<t$. In this case, each pair of composition multisets in $C_{m-i-1}, C_{m-i-2}, \dots, C_{m-i-j'+1}, C_{m-i-j'}$ also differs in exactly $2$ compositions.
\end{itemize} 

\begin{figure*}[h!]
\centering
  \includegraphics[scale= 0.45]{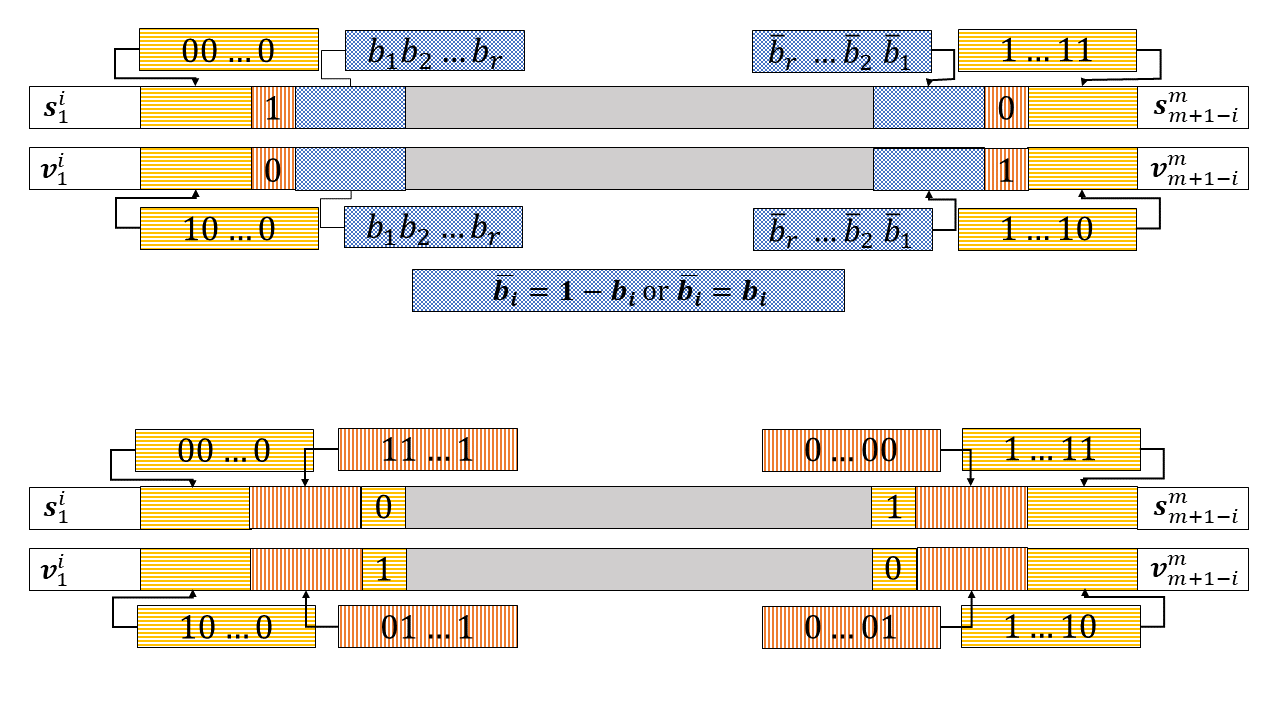}
  \caption{The procedure for constructing the set $\mathcal{V}_{\textbf{s}}$ for several special cases of $\sigma$ values. The first pair depicts the setting in which $\textbf{s}$ and $\textbf{v}$ share a substring $b_1 b_2 \dots b_r$ that follows the substring $00 \dots 0 \, 1$ in \textbf{s} and $10 \dots 0 \, 0$ in \textbf{v}. The length of the substrings $00\dots 0 \, 1 \, b_1 b_2 \dots b_r$ and $10\dots 0 \, 0 \, b_1 b_2 \dots b_r$ in the prefix of \textbf{s} and \textbf{v}, respectively, equals $t+1$. The second pair depicts a setting in which the substring $11\dots 1 \, 0$ follows the $00 \dots 0$ substring in \textbf{s} and the substring $01\dots 1 \, 1$ follows the $10 \dots 0$ substring in \textbf{s}. Here, the lengths of the substrings $00 \dots 0 \, 11\dots 1 \, 0$ in \textbf{s} and $10 \dots 0 \, 01\dots 1 \, 0$ in \textbf{v} may be less than $t+1$.}
  \label{fig:lema_figure4}
\vspace{-0.08in}
\end{figure*}

The bits that were most recently reconstructed in Figure~\ref{fig:lema_figure4} reestablish the initial problem we started with and the analysis henceforth parallels our previous discussion. The pertinent explanations are summarized in Figures~\ref{fig:lema_figure5} and~\ref{fig:lema_figure6}.

\begin{figure*}[h!]
\centering
  \includegraphics[scale= 0.45]{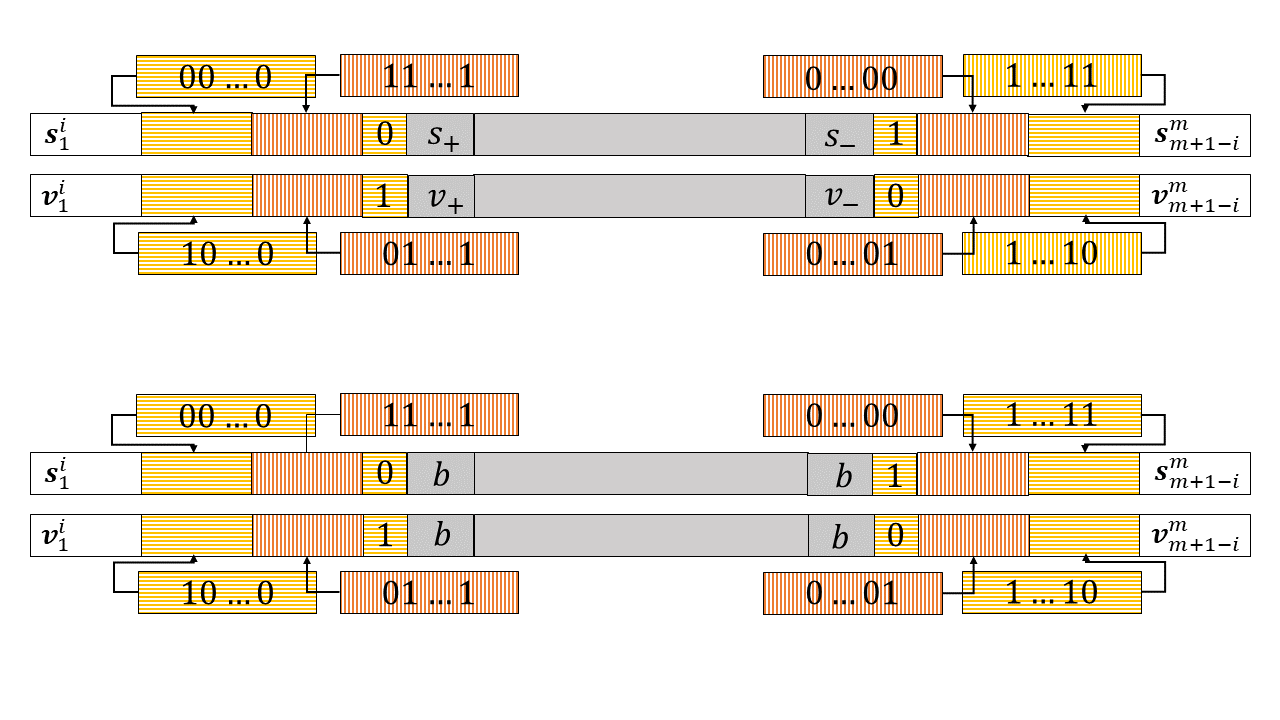}
  \caption{Two pairs of strings explaining how to extend the partially reconstructed strings illustrated in Figure~\ref{fig:lema_figure4}. The first pair corresponds to the case $\sigma_+ =1$, while the second pair corresponds to the case $\sigma_+ \in \{ 0,2\} $ and $s_+ = s_- = v_+ = v_- = b$.}
   \label{fig:lema_figure5}
\vspace{-0.08in}
\end{figure*}
\begin{figure*}[h!]
\centering
  \includegraphics[scale= 0.45]{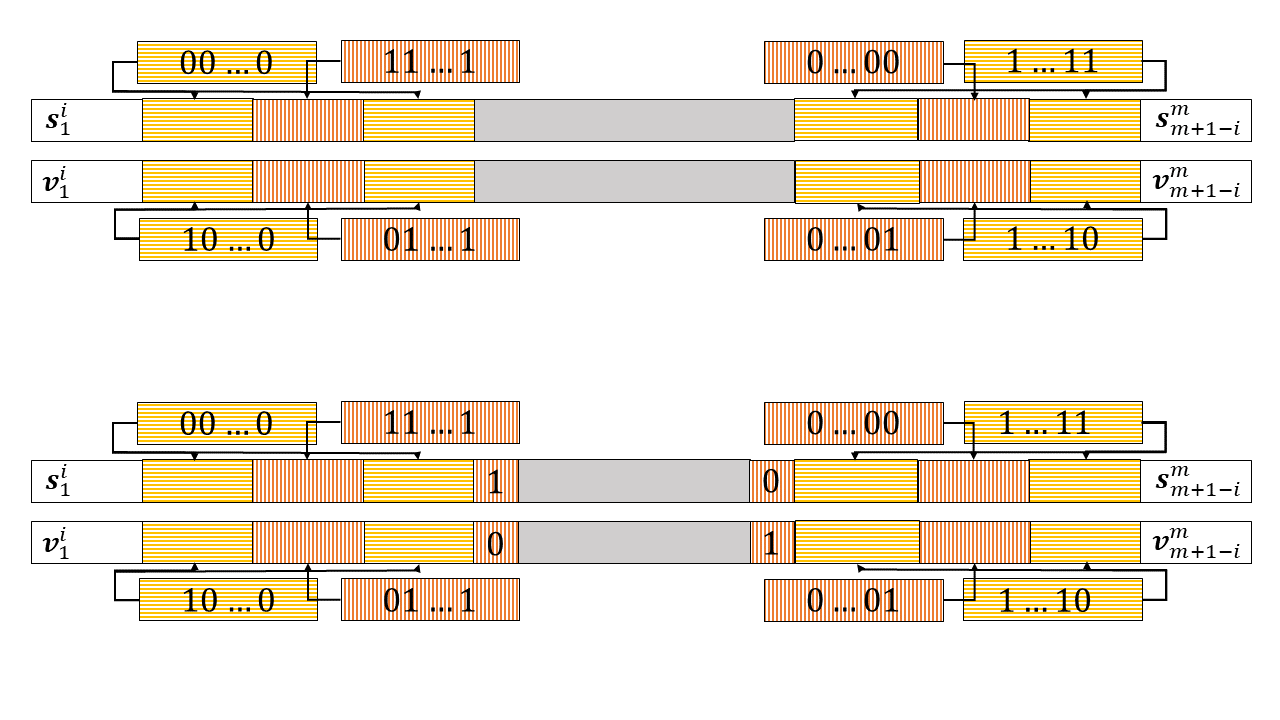}
  \caption{The partial structure of strings in $\mathcal{V}_{\textbf{s}}$ as inferred by the previous analysis and the setup shown in Figure~\ref{fig:lema_figure5}.}
   \label{fig:lema_figure6}
\vspace{-0.08in}
\end{figure*}

Combining the results of all the intermediary steps allows us to describe the set $\mathcal{V}_{\textbf{s}}$ as satisfying one of the two conditions:
\begin{itemize}
\item The string \textbf{s} and a string $\textbf{v} \in \mathcal{V}_{\textbf{s}}$ share a prefix-suffix pair that is followed by a certain number of alternating substrings $00\ldots 0$ and $11\ldots 1$ (in \textbf{s}) and alternating substrings $10\ldots 0$ and $01\ldots 1$ (in \textbf{v}) in the prefixes. The length of the alternating substrings may vary as described in the analysis, and the substrings are induced by $\sigma$ values equal to $1$. The last of the alternating substrings in the prefixes (equal to either  $11\ldots 1$ of  $01\ldots 1$) is followed by a shared substring. The number of bits in the previously described substrings equals $t+1$. The corresponding composition multisets $C_{m-i-1}, C_{m-i-2}, \dots, C_{m-i-t}, C_{m-i-t-1}$ of the string \textbf{s} and $\textbf{v} \in \mathcal{V}_{\textbf{s}}$ differ in exactly $2$ compositions.

\item The string \textbf{s} and a string $\textbf{v} \in \mathcal{V}_{\textbf{s}}$ share a prefix-suffix pair that is followed by a certain number of alternating substrings $00\ldots 0$ and $11\ldots 1$ (in \textbf{s}) and alternating substrings $10\ldots 0$ and $01\ldots 1$ (in \textbf{v}) in the prefixes. The length of the alternating substrings may vary as described in the analysis. The last of the alternating substrings in the prefixes (equal to either  $11\ldots 1$ or $01\ldots 1$) is followed by either the substring $00\ldots 0$ (in \textbf{s}) or $10\ldots 0$ (in \textbf{v}). The number of bits covered by these cases totals $t+1$ and all underlying values of $\sigma$ are equal to $1.$ The corresponding composition multisets $C_{m-i-1}, C_{m-i-2}, \dots, C_{m-i-t}, C_{m-i-t-1}$ of the string \textbf{s} and $\textbf{v} \in \mathcal{V}_{\textbf{s}}$ differ in exactly $2$ compositions.
\end{itemize}
Figure~\ref{fig:lema_figure6} summarizes the structure of the set $\mathcal{V}_{\textbf{s}}$ and concludes our proof.

\begin{figure*}[h!]
\centering
  \includegraphics[scale= 0.45]{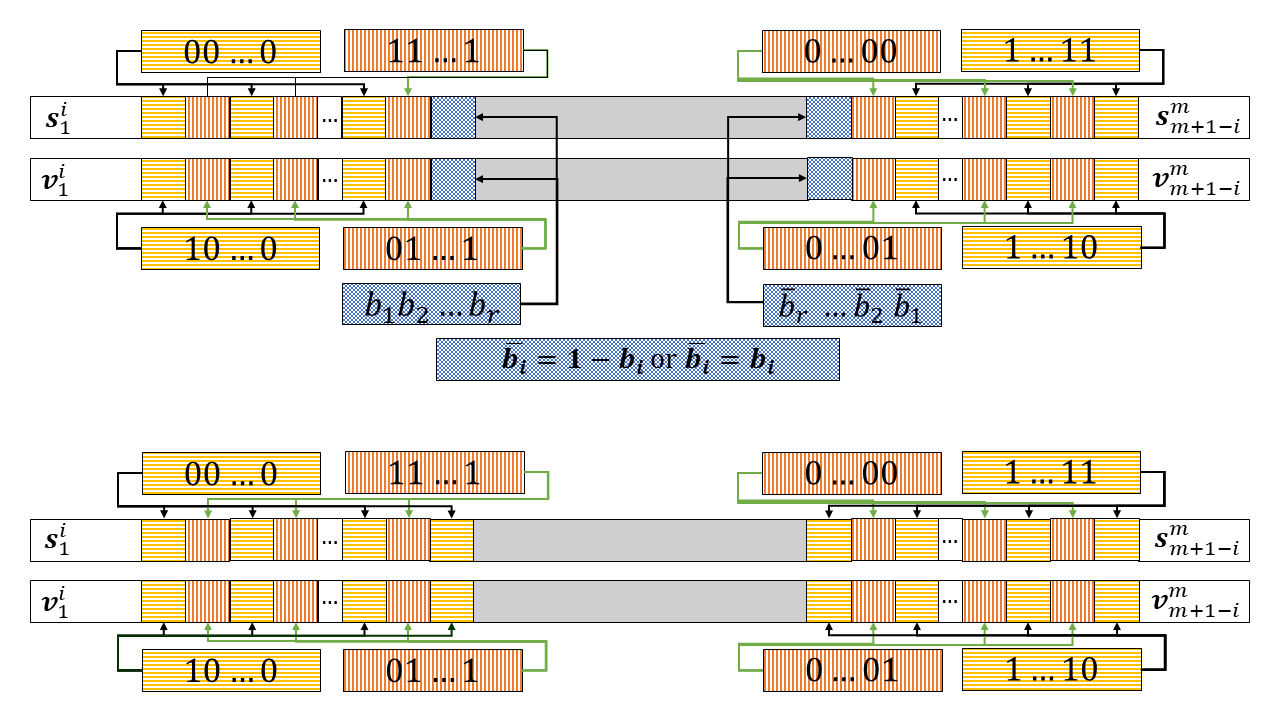}
  \caption{The structure of the strings $\textbf{v} \in \mathcal{V}_{\textbf{s}}$ that are closest to a given string $\textbf{s}$. 
 The first pair of strings illustrates the case where the substrings $00\dots 0$ and $11\dots 1$ in the prefix of \textbf{s}, and $10\dots 0$ and $01\dots 1$ in the prefix of \textbf{v} occur in pairs, ending with a shared substring $b_1 b_2 \dots b_r$. The second pair illustrates the case where the substrings $00\dots 0$ and $11\dots 1$ in the prefix of \textbf{s}, and $10\dots 0$ and $01\dots 1$ in the prefix of \textbf{v} occur in pairs ending with the substring $00\dots 0$ in the prefix of \textbf{s} and $10\dots 0$ in the prefix of \textbf{v}. Note that the substrings may not be of equal lengths. With the exception of the final shared substring (i.e., shared substring $b_1 b_2 \dots b_r$ for the first pair, and the substrings $00\dots 0$ in \textbf{s} and $10\dots 0$ in \textbf{v} for the second pair) all strings are of length at least one. The number of bits in the prefix of each string obtained by alternating the above substrings equals $t+1$. }
   \label{fig:lema_figure7}
\vspace{-0.08in}
\end{figure*}

\begin{table}[H]
\centering
\begin{tabular}{|l|l|l|l|l|}
\hline
$s_+$           & 0 & 0 & 1 & 1 \\ \hline
$v_+$           & 0 & 1 & 0 & 1 \\ \hline
Set Difference & 2 & 4 & 2 & 4 \\ \hline
\end{tabular} 
\caption{Four different assignments of values for $(s_+,v_+)$ and the resulting set cardinalities $|C_{m-i-2}(\textbf{s}) \setminus C_{m-i-2}(\textbf{v})|$.} \label{tab:yellow1}
\end{table}

\begin{table}[H]
\centering
\begin{tabular}{|l|l|l|}
\hline 
$b$ & 0 & 1 \\ \hline
Set Difference        & 3 & 3 \\ \hline 
\end{tabular} 
\caption{Cardinalities of the set difference $|C_{m-i-2}(\textbf{s}) \setminus C_{m-i-2}(\textbf{v})|$ for different choices of $b$.} \label{tab:yellow2}
\end{table}

\begin{table}[H]
\centering
\begin{tabular}{|l|l|l|l|l|}
\hline
$s_+$           & 0 & 0 & 1 & 1 \\ \hline
$v_+$           & 0 & 1 & 0 & 1 \\ \hline
Set Difference & 2 & 4 & 2 & 4 \\ \hline
\end{tabular}
\caption{Values of $|C_{m-i-i'-2}(\textbf{s}) \setminus C_{m-i-i'-2}(\textbf{v})|$ for the setting where the bits $s_{i+1}$ and $v_{i+1}$ are followed by a run of $i'$ $0$s, and where $\sigma_{i+2}^{i+1+i'} = (1, 1, \dots 1)$ and $\sigma_+ = 1$.}  \label{tab:yellow3}
\end{table} 

\begin{table}[H]
\centering
\begin{tabular}{|l|l|l|}
\hline 
$b$ & 0 & 1 \\ \hline
Set Difference        & 3 & 3 \\ \hline 
\end{tabular} 
\caption{The possible values of $|C_{m-i-i'-2}(\textbf{s}) \setminus C_{m-i-i'-2}(\textbf{v})|$ for the case that the bits $s_{i+1}$ and $v_{i+1}$ are followed by a run of $i'$ $0$s, and such that $\sigma_{i+2}^{i+1+i'} = (1, 1, \dots 1)$ and $\sigma_+ \in \{ 0,2\}$.} \label{tab:yellow4}
\end{table}

\begin{table}[H]
\centering
\begin{tabular}{|l|l|l|l|l|}
\hline
$s_+$           & 0 & 0 & 1 & 1 \\ \hline
$v_+$           & 0 & 1 & 0 & 1 \\ \hline
Set Difference & 2 & 2 & 4 & 2 \\ \hline
\end{tabular} 
\caption{The possible values of $|C_{m-i-i'-3}(\textbf{s}) \setminus C_{m-i-i'-3}(\textbf{v})|$ for $\sigma_+ =1$ corresponding to the four binary assignments for $(s_+, v_+)$ under the setting illustrated in Figure~\ref{fig:lema_figure3}.} \label{tab:red1}
\end{table} 

\begin{table}[H]
\centering
\begin{tabular}{|l|l|l|}
\hline 
$b$ & 0 & 1 \\ \hline
Set Difference        & 2 & 2 \\ \hline 
\end{tabular} 
\caption{The possible values of $|C_{m-i-i'-3}(\textbf{s}) \setminus C_{m-i-i'-3}(\textbf{v})|$ for different choices of $b$ under the setting illustrated in Figure~\ref{fig:lema_figure3}.} \label{tab:red2}
\end{table}

\begin{table}[H]
\centering
\begin{tabular}{|l|l|l|l|}
\hline
$s_+$           & 0 &  1 & 1 \\ \hline
$v_+$           & 0 &  0 & 1 \\ \hline
Set Difference & 2 &  4 & 2 \\ \hline
\end{tabular} 
\caption{The possible values of  $|C_{m-i-i'-r-2}(\textbf{s}) \setminus C_{m-i-i'-r-2}(\textbf{v})|$ for $\sigma_+ =1$ corresponding to three binary assignments $(s_+, v_+)$ under the setting illustrated in Figure~\ref{fig:lema_figure3b}.} \label{tab:blue1}
\end{table}

\begin{table}[H]
\centering
\begin{tabular}{|l|l|l|}
\hline 
$b$ & 0 & 1 \\ \hline
Set Difference        & 2 & 2 \\ \hline 
\end{tabular} 
\caption{Cardinalities of the set difference $|C_{m-i-i'-r-2}(\textbf{s}) \setminus C_{m-i-i'-r-2}(\textbf{v})|$ for different choices of $b$ under the setting illustrated in Figure~\ref{fig:lema_figure3b}.} \label{tab:blue2}
\end{table}


\ifCLASSOPTIONcaptionsoff
  \newpage
\fi


\end{document}